\documentclass[nofootinbib,aps,prx,amsfonts,amsmath,amssymb,twocolumn,superscriptaddress, floatfix]{revtex4-2}
\usepackage{graphicx}
\usepackage{subcaption}
\usepackage{caption}
\usepackage{overpic}
\usepackage{physics}
\usepackage{amsthm}
\usepackage{mathrsfs}
\usepackage{mathtools}
\usepackage{enumerate}
\usepackage{xcolor}
\usepackage{hyperref}
\hypersetup{
    colorlinks,
    linkcolor=blue,
    citecolor=blue,
    urlcolor=blue
}
\usepackage[ruled]{algorithm2e}
\usepackage{adjustbox}
\usepackage{float}
\usepackage{wasysym}
\usepackage{url}

\usepackage{enumitem}
\setlist[enumerate,1]{label={(\alph*)}}

\newtheorem{theorem}{Theorem}[section]
\newtheorem{definition}[theorem]{Definition}
\newtheorem{lemma}[theorem]{Lemma}

\def\orcid#1{\kern .08em\href{https://orcid.org/#1}{\includegraphics[keepaspectratio,width=0.7em]{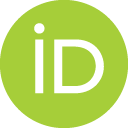}}} 

\newcommand{\idx}{\operatorname{idx}}
\newcommand{\setbuilder}[2]{
    \left\{ #1 \mathrel{\left|\vphantom{#1 #2}\right.} #2 \right\}
}
\DeclarePairedDelimiterX\Set[2]{\{}{\}}{#1\,\delimsize\vert\,#2}

\usepackage{tikz}
\usetikzlibrary{quantikz,tikzmark,calc,decorations.pathreplacing,calligraphy}

\begin{document}
\preprint{APS/123-QED}

\title{Quantum Dueling: An Efficient Solution for Combinatorial Optimization}
\author{Letian Tang\orcid{0000-0001-8882-9187}}
\email{letiant@andrew.cmu.edu}
\affiliation{
Carnegie Mellon University, Pittsburgh, PA 15213, United States
}

\author{Haorui Wang\orcid{0000-0002-9237-1794}}
\affiliation{ University of Cambridge, Cambridge, CB2 1TN, United Kingdom }

\author{Zhengyang Li\orcid{0000-0003-1863-3423}}
\affiliation{
Tongji University, Shanghai, 200070, China}

\author{Haozhan Tang\orcid{0000-0001-6719-846X}}
\affiliation{
Carnegie Mellon University, Pittsburgh, PA 15213, United States
}
 
\author{Chi Zhang\orcid{0000-0002-8658-2931}}
\thanks{These authors contributed equally to this work.}
\affiliation{
Wuhan University, Wuhan, 430072, China 
}

\author{Shujin Li\orcid{0000-0002-5730-9098}}
\thanks{These authors contributed equally to this work.}
\affiliation{
Carnegie Mellon University, Pittsburgh, PA 15213, United States
}

\date{\today}

\begin{abstract}

In this paper, we present a new algorithm for generic combinatorial optimization, which we term quantum dueling. Traditionally, potential solutions to the given optimization problems were encoded in a ``register'' of qubits. Various techniques are used to increase the probability of finding the best solution upon measurement. Quantum dueling innovates by integrating an additional qubit register, effectively creating a ``dueling'' scenario where two sets of solutions compete. This dual-register setup allows for a dynamic amplification process: in each iteration, one register is designated as the 'opponent', against which the other register's more favorable solutions are enhanced through a controlled quantum search. This iterative process gradually steers the quantum state within both registers toward the optimal solution. With a quantitative contraction for the evolution of the state vector, classical simulation under a broad range of scenarios and hyper-parameter selection schemes shows that a quadratic speedup is achieved, which is further tested in more real-world situations. In addition, quantum dueling can be generalized to incorporate arbitrary quantum search techniques and as a quantum subroutine within a higher-level algorithm. Our work demonstrates that increasing the number of qubits allows the development of previously unthought-of algorithms, paving the way for advancement of efficient quantum algorithm design.


\end{abstract}

\maketitle

\section{Introduction}
\label{sec:introduction}

Since the realization that quantum systems can be used to establish a new computational infrastructure \cite{Manin_1980,Feynman_1982}, many quantum techniques have been developed for general optimization problems. For example, the Grover adaptive search (GAS) iteratively maintains a current best guess to find the optimal solution \cite{Durr_1996, Bulger_2003, Baritompa_2005}. Meanwhile, quantum annealing and quantum adiabatic evolution encode the problem as a Hamiltonian and find the lowest energy state \cite{Falco_1988, Hauke_2020,farhi_2000}. Lastly, variational techniques such as QAOA have recently gained significant attention, where we interchange two different Hamiltonians to find the desired state \cite{Farhi_2014, Hadfield_2019, Fuchs_2022}. 


In this article, we propose a new general strategy for quantum approximation. Our algorithm is inherently based on the Grover search, which amplifies the amplitude components of the state vector that satisfy a given condition \cite{Grover_1996, Brassard_1998, Brassard_2002, Bennett_1997, Boyer_1998, Zalka_1999}. In previous designs, the candidates for the problem are encoded as basis vectors for the Hilbert space. Our algorithm doubles the number of qubits used. This means that the basis vectors of our Hilbert space represent a pair of two potential candidates from the search space as opposed to one. As an example, after obtaining a well-optimized state with the procedures that would be introduced later, a measurement will give us a pair of candidates for the original problem instead of one. Then, we can use a classical comparison to find the better one as the final output of the program.

But how do we optimize the state under such a representation? We realize that a sequence of controlled Grover amplifiers can achieve such an objective. Let the two sets of qubits that are used to represent two candidates of the given problem be two registers. Each time, we can use one register as the control and apply Grover's gate to amplify states in the other register that are more optimal (than the control). Since we have two choices for the control register each time, different choices will direct the evolution of the state vector to different paths. After a sufficient number of gates, we will obtain a state with a significant probability of finding the best solution after measurement. In this paper, we will define an adjustable parameter to describe the sequence of our choices, making our algorithm variational. Since we have two candidates, one optimizing the other in each gate application, we term the algorithm ``quantum dueling.''

In this work, we examined the performance of quantum dueling via classical simulation and certain quantitative analysis. We found that compared with the classical brute-force method, quantum dueling bolsters a quadratic speed-up under optimal parameters. As a fully quantized algorithm, we believe that it would be useful as a quantum subroutine for generic optimization. It is worth noting that the algorithm does show some limitations. Most notably, a good choice of parameter is required when the search space is significantly large ($10^6$), and we would like to find the global optimal with high probability; without a well-considered parameter, the state will stop being further boosted once some threshold is reached. However, we do believe that if we can find an approximate solution to the evolution of the state vector during the algorithm, there should be ways to determine the optimal parameter rigorously. Nevertheless, the main purpose of our research is less about the algorithm itself, but to demonstrate that spending extra qubits and enlarging the search space would provide more flexibility in quantum algorithm design, sometimes leading to previously unthought-of algorithms. 

The structure of the paper is organized as follows. Initially, Sec.~\ref{sec: algorithm} delineates the combinatorial optimization problem central to this study and introduces the fundamental construction of quantum dueling. Subsequent to this foundation, Sec.~\ref{sec:general_understanding} presents the recursive formula for state evolution, accompanied by an in-depth examination of solution distributions under the naive parameter choice of $\alpha_{i} = \beta_{i} = 1$. Sec.~\ref{sec:cluster_representation} provides a more computationally efficient version of the algorithm, while Sec.~\ref{sec:hamiltonian_version} extends the discussion to encompass the entire algorithm within a Hamiltonian framework, elucidating its connections to established quantum algorithms like Grover's search and QAOA (Quantum Approximate Optimization Algorithm). To address the nuances of parameter selection in quantum dueling, Secs.~\ref{sec:parameter_scheme_1},\ref{sec:parameter_scheme_2} dissect the algorithm's performance across a variety of solution distributions with alternate parameter schemes. For certain unique solution distributions, Sec.~\ref{sec:parameter_heuristics} employs a classical greedy approach to obtain locally optimal parameters. Building upon these foundational elements, Sec.~\ref{sec:application} explores the practical applications of quantum dueling. In Sec.~\ref{sec:discussion}, we broaden the perspective to consider the mechanisms of dueling, its extended variants, and its potential as a subroutine within larger quantum or hybrid algorithmic structures. The paper culminates in Sec.~\ref{sec:conclusion} with a comprehensive conclusion that encapsulates the insights and implications derived from our exploration of quantum dueling.

\section{Algorithm}
\label{sec: algorithm}

The general combinatorial optimization problem can be formulated as follows. Consider a search space $S$ of size $N$, without loss of generality, we can assume that $N = 2^n$, where $n\in \mathbb{N}$. We define a measure function $v: S \to \mathbb{R}$, and a function $f: S \to \{0,1\}$ that labels solutions with 1 and others 0. \footnote{The need for $f$ is justified in Sec.~\ref{sec:general_understanding}.} Let $M$ denote the total number of solutions in $S$. We wish to find a solution $x \in S$ such that $f(x) = 1$ and $v(x)$ is minimized. In other words, $x$ satisfies:
\begin{equation}
\label{eq:prob_def}
    v\left( x \right) = \min \Set{v\left( y \right)}{y \in S \wedge f\left( y \right) = 1}
\end{equation}

Such a problem can be seen as an extension of the search problem, where we are looking for a solution $x$ with $f(x) = 1$, without a measure function $v$. For search, Grover has shown that there is a quantum algorithm that performs quadratically better than the classical counterpart. In his algorithm, verifying whether a candidate is a solution is achieved via an oracle gate, which is essentially a black box that can be constructed based on the problem itself. In this paper, we will show that such quantum advantage can be extended to the combinatorial optimization problem via a fully quantized approach.

Algorithm~\ref{alg:classical_optimizer} shows the classical algorithm for combinatorial optimization. At any stage after initialization, we store a current best guess and iterate through the sample space (or randomly select an element). We repeatedly conduct this procedure to find a better solution. The complexity of the algorithm is $O(N) = O(2^{n})$. 
\begin{algorithm}
        \caption{Classical Brute-Force Optimizer}\label{alg:classical_optimizer}
        \SetKw{False}{false} 
        \SetKw{True}{true}
        \KwIn{Measure function $v$, solution indicator $f$ in the form of an oracle}
        \KwOut{An element $x$ in the search space satisfying $f(x)=1$ and $v(x)$ (approximately) minimized}
        Initialize some $\mathsf{best}$ to any element from the search space $S$\;
        \For{all other elements $x \in S$}{
            Use an oracle to compare $x$ and $\mathsf{best}$\;
            \lIf{$x$ is better}{$\textsf{best} \leftarrow x$}
        }
\end{algorithm}

Amplitude amplification has offered us a promising way to approach quantum combinatorial optimization problems with the potential of maintaining the quantum advantage. Brassard \cite{Brassard_2002} formulated the general amplitude amplification algorithm using unitaries, including the oracle and the diffusion operator, to solve the original search problem starting from an arbitrary initial state vector.

Despite the difference in problem formulation, we can still use the general idea of amplitude amplification by iteratively applying the oracle $O$ and the diffusion operator $D$. But the major problem here is to encode the measure function $v(x)$ into our algorithm design. One possible solution, as mentioned in \cite{Mat_2023}, is to encode the function $v(x)$ in the problem Hamiltonian, $\hat{H}_{p} = \sum\limits_{x} v(x) \ket{x}\bra{x}$ in a Hamiltonian-based approach, letting the system gradually progress to the global minima of $v(x)$ via a QAOA-like quantum gradient descent. 

In this paper, we consider a different method. Rather than using the Hamiltonian to represent the function, we encode the clause of comparison $v(x) < v(y)$ into the oracle, which is then used to build a new algorithm in a circuit-based approach. However, due to the ``No Cloning Principle", it is impossible to achieve comparison via the quantum oracle with only one $n$-qubit register, using $B = \Set{\ket{x}}{x\in S}$ as a basis for a Hilbert space $\mathscr{H}_S$. 

To address this problem, we employ two registers with $n$-qubtis each, based on the idea that we need at least two copies of the same information to achieve comparison. In other words, the Hilbert space of our concern is $\mathscr{H} = \mathscr{H}_S \otimes \mathscr{H}_S$. This means that the basis is now $B = \Set{\ket{kl}}{k,l\in S}$, $\ket{kl} = \ket{k} \otimes \ket{l}$. The index $k$ denotes elements represented by the first register and $l$ elements from the second register. Thus, the state vector can be decomposed as the following:
\begin{equation}
    \label{eq:expand_full}
    \left| \psi  \right\rangle = \sum\limits_{k = 1}^{N}\sum\limits_{l = 1}^{N} {\psi _{kl}}\ket{kl} \\
\end{equation}
       
A good interpretation is to put $\psi_{kl}$ in a matrix format, as in Fig.~\ref{fig:blank_grid}.

\begin{figure}
    \centering
    \includegraphics[width=\columnwidth]{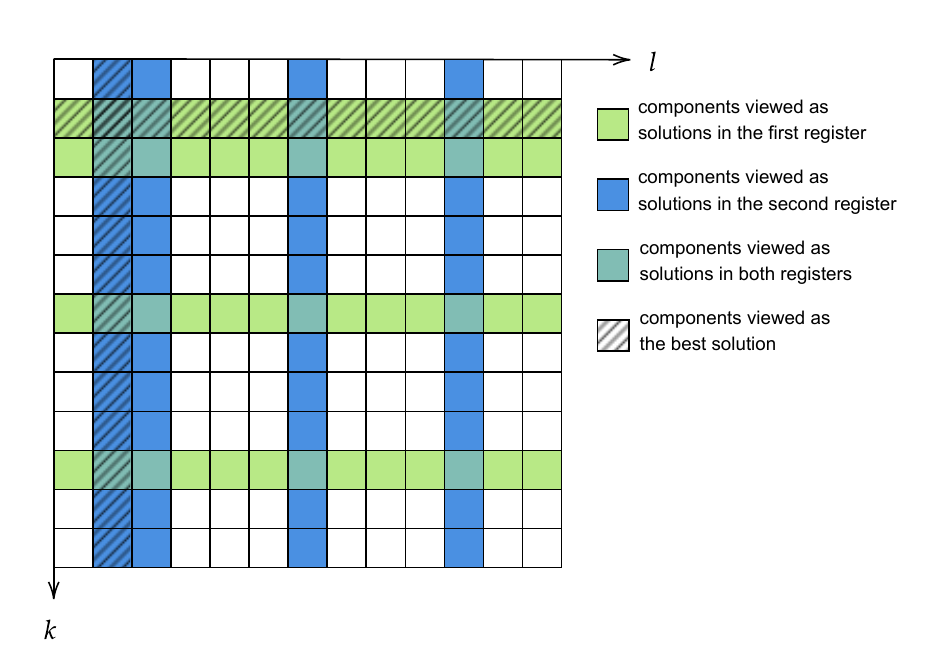}
    \caption{Matrix representation of the components $\psi_{kl}$ in the augmented Hilbert Space $\mathscr{H} = \mathscr{H}_S \otimes \mathscr{H}_S$. The example features the problem with $N=13$, $M=4$, $f(x)=1$ when $x=2,3,7,11$ and $0$ otherwise, where $v(x)=x$ for all $x\in S$. Indices $k$ and $l$ stand for the labels for elements in the first and the second register, respectively. Each cell in the grid can be filled with a number that represents the component $\psi_{kl}$ as defined in Eq.~(\ref{eq:expand_full}). Cells such that either $k$ or $l$ is a solution are highlighted. In this problem, the best solution is $x=2$. Cells at which either $k=2$ or $l=2$ are shaded. When the state vector is a basis vector (or a linear combination of basis vectors) corresponding to the cells in the shaded region, the best solution will be obtained upon measurement.}  
    \label{fig:blank_grid}
\end{figure}

We term our new algorithm ``quantum dueling,'' as we now use two registers to ``compete'' against each other throughout the quantum state evolution in the augmented Hilbert Space $\mathscr{H} = \mathscr{H}_{S} \otimes \mathscr{H}_{S}$. 


In quantum dueling, a measurement of both registers independently at the end of the quantum state evolution is needed and thus will yield two potential solutions. We need to evaluate the $f$ and $v$ values for both candidates, taking the better result as the final output. For clarity, we say that an element $x\in S$ is better than $y\in S$ if and only if $f(x) = 1 \wedge (f(y) = 0 \vee  v(x) < v(y))$. This paper uses logical operators $\wedge$, $\vee$ to represent and, or operations. Meanwhile, we will use $\&$, $|$ operators to represent and, or operations in $\{0,1\}$ space. If neither candidate is a solution to the original problem or a better approximation ratio is desired, we can repeat the algorithm.

On the augmented Hilbert space $\mathscr{H} = \mathscr{H}_{S} \otimes \mathscr{H}_{S}$, we define the quantum oracle as follows: 
\begin{equation}
    \label{eq:oracle_definition_1}
    \mathcal{O}\ket{xy}= (-1)^{f(x) \& [v(x)<v(y)]}\ket{xy} = o(x,y)\ket{xy}
\end{equation}

The result of the oracle applied to an arbitrary state is the linear combination of the oracle applied to its basis decomposition by linearity. Therefore, 
\begin{equation}
    \label{eq:oracle_definition_2}
    \mathcal{O}= \sum_{x=1}^{N}\sum_{y=1}^{N}{o(x,y)\ket{xy}\bra{xy}}
\end{equation}

The oracle operates in two registers and outputs the ``better'' result.  For states $x,y\in S$, our quantum oracle will multiply a basis state by a factor of $-1$ if and only if $f(x) = 1$ and $v(x) < v(y)$, which is our criterion for ``betterness''. Thus, it can distinguish the better states from the entire search space, similar to the oracle in traditional amplitude amplification.

Here, we use the square bracket to denote conversion from a proposition to $\{0,1\}$ space. For example, the clause $[v(x)<v(y)]$ returns 1 if $v(x)<v(y)$ is satisfied, and 0 otherwise. The same definition also applies to clause $f(x)$. For simplicity, in the future, the term $(-1)^{f(x) \& [v(x)<v(y)]}$ will become $o(x,y)$.

Usually, we can achieve the oracle by evaluating the $f$ and $v$ values of the two input elements and making a comparison. All these steps can be achieved using quantum arithmetic techniques with the help of additional qubits in polynomial time. 

Note that in Eqs.~(\ref{eq:oracle_definition_1}) and (\ref{eq:oracle_definition_2}), we use $x$ and $y$ as the free indices rather than symbols $k$ and $l$. This is because $x$ and $y$ can each represent basis states from either the first or the second register. Substituting $x$, $y$ with $k$, $l$ respectively gives the oracle that compares the first register against the second, and substituting $x$, $y$ with $l$, $k$ gives the oracle that compares the second against the first. That is we have oracles $\mathcal{O}_{1\leftarrow 2}$ and $\mathcal{O}_{2\leftarrow 1}$ satisfying: 
\begin{equation}
    \label{eq:oracles_definition_1}
    \left\{\begin{gathered}
        \mathcal{O}_{1\leftarrow 2}\ket{kl} = o(k,l)\ket{kl} \hfill \\
        \mathcal{O}_{2\leftarrow 1}\ket{kl} = o(l,k)\ket{kl} \hfill \\
    \end{gathered}\right.
\end{equation}

Equivalently,
\begin{equation}
    \label{eq:oracles_definition_2}
    \left\{\begin{gathered}
        \mathcal{O}_{1\leftarrow 2} = \sum_{k=1}^{N}\sum_{l=1}^{N} o(k,l) \ket{kl}\bra{kl} \hfill \\
        \mathcal{O}_{2\leftarrow 1} = \sum_{l=1}^{N}\sum_{k=1}^{N} o(l,k) \ket{kl}\bra{kl} \hfill \\
    \end{gathered}\right.
\end{equation}

Similarly, we define the diffusion operators $\mathcal{D}_{1\leftarrow 2}$ and $\mathcal{D}_{2\leftarrow 1}$ in the augmented Hilbert space $\mathscr{H} = \mathscr{H}_{S} \otimes \mathscr{H}_{S}$ as follows: 
\begin{equation}
    \label{eq:diffusion_definition}
    \begin{cases}
        \mathcal{D}_{1\leftarrow 2} = D \otimes I \\
        \mathcal{D}_{2\leftarrow 1} = I \otimes D \\
    \end{cases}
\end{equation}
where $D = H^{\otimes n}(2\ket{0}\bra{0}-I)H^{\otimes n} = (2\ket{m}\bra{m}-I)$, $H$ is the Hardmard gate, $\ket{m} = \frac{1}{\sqrt{N}}\sum\limits_{x=1}^{N}|x\rangle$ is the equal superposition of all the basis states.

To understand how the oracles and diffusion operators on the augmented Hilbert space impact some state vector $\ket{\psi}$ in the augmented Hilbert space. we consider it's decomposition:
\begin{equation}
    \label{eq:redefinition_of_state_Vector}
    \begin{aligned}
        \ket{\psi} &= \sum\limits_{l=1}^{N} \ket{\psi_{l}} \ket{l} \\
                   &= \sum\limits_{k=1}^{N} \ket{k} \ket{\psi_{k}} \\
    \end{aligned}
\end{equation}

$\ket{\psi_{l}}$, $\ket{\psi_{k}}$ can be defined as: 
\begin{equation}
    \label{eq:psik,psil_def}
    \begin{cases}
        \ket{\psi_{l}} = \sum\limits_{k = 1}^{N} \psi_{kl} \ket{k} \\
        \ket{\psi_{k}} = \sum\limits_{l = 1}^{N} \psi_{kl} \ket{l} \\
    \end{cases}
\end{equation}
where $\psi_{kl}$ are the state vector's components. An interesting property of $\ket{\psi_{l}}$, $\ket{\psi_{k}}$ is that their norm squared represent the probability of obtaining the corresponding results $l$ and $k$ when a measurement is conducted for either the first or the second register.
\begin{equation}
    \label{eq:inner_product_psi}
    \begin{cases}
        \langle \psi_{l}  | \psi_{l} \rangle = \sum\limits_{k=1}^{N} \left| \psi_{kl}\right|^{2} = P_{l}  \\
        \langle \psi_{k} | \psi_{k} \rangle = \sum\limits_{l=1}^{N} \left| \psi_{kl}\right|^{2} = P_{k}  \\
    \end{cases}
\end{equation}

We then define a sequence of oracles $\{O_y\}$ that act only in the smaller Hilbert space $\mathscr{H}_S$. We can interpret the register corresponding to the notation $y$ as the control register that directs the oracle to act on the register corresponding to the notation $x$ in a specific way.

Formally speaking, we have:
\begin{equation}
    \label{eq:sub_oracle_definition}
    \begin{gathered}
        O_y \ket{x} = o(x,y) \ket{x} \hfill \\
        \Longrightarrow O_y = \sum_{x=1}^{N}o(x,y)\ket{x}\bra{x} \hfill \\
    \end{gathered}
\end{equation}

This allows us to define the corresponding Grover gates $\left\{G_y\right\}$ that satisfy $G_y = D O_y$ for all $1\leqslant x \leqslant N$. Similarly to what was discussed previously, the notation $x$ and $y$ can be replaced with $k$ or $l$ according to implementation.

Thus, it is possible to re-express the oracles in the following manner: 
\begin{equation}
    \label{eq:redefinition_of_oracles}
    \begin{cases}
        \mathcal{O}_{1\leftarrow 2} = \sum\limits_{l=1}^{N} O_{l} \otimes \ket{l}\bra{l} \\
        \mathcal{O}_{2\leftarrow 1} = \sum\limits_{k=1}^{N} \ket{k}\bra{k} \otimes O_{k} \\
    \end{cases}
\end{equation}

The oracle $\mathcal{O}_{1\leftarrow 2}$ uses the second register as a control and acts on the first register, while the oracle $\mathcal{O}_{2\leftarrow 1}$ uses the first as a control and operates on the second. The operation of either of the two oracles entangles the two registers. The generation of entanglement between the registers may serve as a major feature of quantum dueling.

The result of the oracles on an arbitrary state vector $\ket{\psi}$ will be:
\begin{equation}
    \label{eq:oracle_operation_in_expansion}
    \begin{cases}
        \mathcal{O}_{1\leftarrow 2}\ket{\psi} = \sum\limits_{l=1}^{N}\left(O_l\ket{\psi_l}\right)\ket{l} \\
        \mathcal{O}_{2\leftarrow 1}\ket{\psi} = \sum\limits_{k=1}^{N}\ket{k}\left(O_k\ket{\psi_k}\right) \\
    \end{cases}
\end{equation}

When viewed in this way, it is intuitive to perform a controlled Grover search to the state vector. In one iteration, we have
\begin{equation}
    \label{eq:dueling_operation_in_expansion}
    \begin{cases}
        \mathcal{G}_{1\leftarrow 2}\ket{\psi} = (D\otimes I)\mathcal{O}_{1\leftarrow 2}\ket{\psi} = \sum\limits_{l=1}^{N}\left(D O_l\ket{\psi_l}\right)\ket{l} \\
        \mathcal{G}_{2\leftarrow 1}\ket{\psi} = (I\otimes D)\mathcal{O}_{2\leftarrow 1}\ket{\psi} = \sum\limits_{k=1}^{N}\ket{k}\left(D O_k\ket{\psi_k}\right) \\
    \end{cases}
\end{equation}
We define these gates as dueling operators. They can also be written as: 
\begin{equation}
    \label{eq:redefinition_of_dueling}
    \begin{cases}
        \mathcal{G}_{1\leftarrow 2} = \sum\limits_{l=1}^{N} DO_{l} \otimes \ket{l}\bra{l} \\
        \mathcal{G}_{2\leftarrow 1} = \sum\limits_{k=1}^{N} \ket{k}\bra{k} \otimes DO_{k} \\
    \end{cases}
\end{equation}

In matrix visualization, the result of applying $\mathcal{G}_{1 \leftarrow 2}$ is that for each column corresponding to $l$, components such that $k$ is ``better''  than $l$ will get boosted from components that are not better. Fig.~\ref{fig:dueling_G1} illustrates how this is done in one column starting from one cell at which $k$ is not ``better'' than $l$. 

\begin{figure}
    \centering
    \includegraphics[width=\columnwidth]{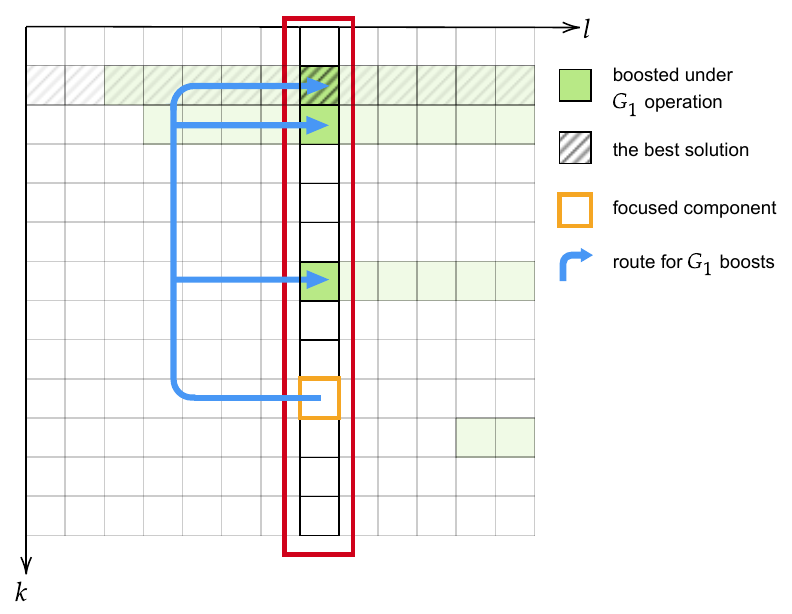}
    \caption{An illustration of the components of the state in the augmented Hilbert Space $\mathscr{H} = \mathscr{H}_{S}\otimes \mathscr{H}_{S}$ after the application of the dueling operator $\mathcal{G}_{1\leftarrow 2}$. The problem setup is the same as Fig.~\ref{fig:blank_grid}. The components boosted under $\mathcal{G}_{1\leftarrow 2}$ operation are marked green. For a specified column with a given $l$, the amplitude $\psi_{kl}$ of an arbitrary non-highlighted component will be distributed equally to the highlighted component within the same column. To better illustrate our idea, we pay attention to the focused component and its corresponding column, observing how its amplitude is redistributed under the operation of $\mathcal{G}_{1\leftarrow 2}$ to other components in the same column. This example features the column $l=8$ and the focused component $k=10$, $l=8$. Note that the best solution $k=2$ is more likely to receive amplitude from other cells under $\mathcal{G}
_{1}$ operation, since the row $k=2$ has the most cells marked green. }
    \label{fig:dueling_G1}
\end{figure}

Similarly, an illustration of the amplifying result of $\mathcal{G}_{2\leftarrow 1}$ can be found in Fig.~\ref{fig:dueling_G2}.

\begin{figure}
    \centering
    \includegraphics[width=\columnwidth]{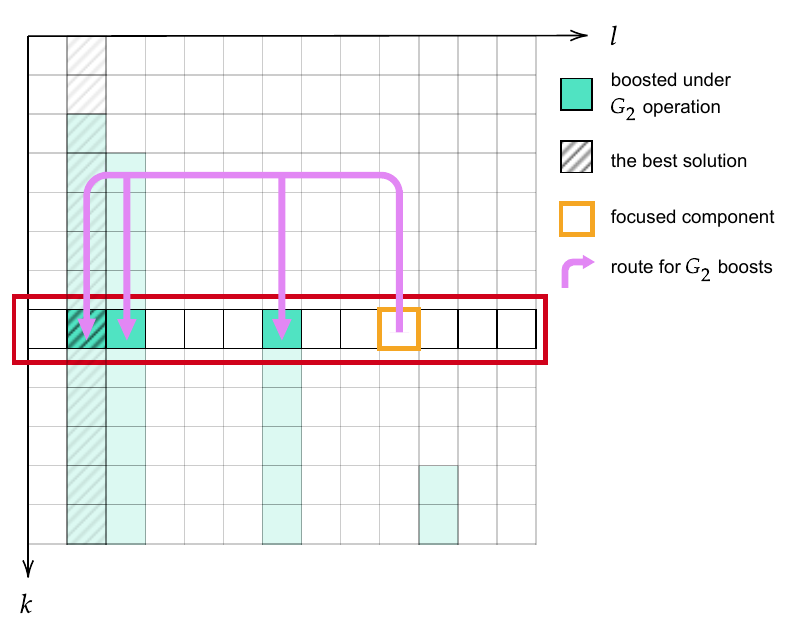}
    \caption{An illustration of the components of the state in the augmented Hilbert Space $\mathscr{H} = \mathscr{H}_{S}\otimes \mathscr{H}_{S}$ after the application of the dueling operator $\mathcal{G}_{2\leftarrow 1}$. The problem setup is the same as Fig.~\ref{fig:blank_grid}. The components whose amplitudes are increased under $\mathcal{G}_{2\leftarrow 1}$ operation are marked light blue. For a specified row with given $k$, the amplitude $\psi_{kl}$ of an arbitrary non-highlighted component will be redistributed equally to the highlighted component within the same row. This example features the row $k=8$. The focused component illustrated is $k=8$, $l=10$. Note that the best solution $l=2$ is more likely to receive amplitude from other cells under $\mathcal{G}
    _{2}$ operation, because the column $l=2$ has the most cells marked light blue. The alternating sequence of $\mathcal{G}_{1\leftarrow 2}$ and $\mathcal{G}_{2\leftarrow 1}$ acting on the state vector increases the amplitude of components corresponding to the best solution under either operation significantly. }
    \label{fig:dueling_G2}
\end{figure}

Dueling operators $\mathcal{G}_{1\leftarrow 2}$, $\mathcal{G}_{2\leftarrow 1}$ give us two symmetrical ways to boost the probability amplitude corresponding to the best solutions. If we apply gates $\mathcal{G}_{1\leftarrow 2}$ and $\mathcal{G}_{2\leftarrow 1}$ alternatively, we will slowly evolve the state towards the desired direction. We use two arrays of integers $\{\alpha_i\}$ and $\{\beta_i\}$ to quantify the exact sequences of gates applied to the state. In the $i$-th iteration, we apply $\mathcal{G}_{1\leftarrow 2}$ by $\alpha_i$ times and $\mathcal{G}_{2\leftarrow 1}$ by $\beta_i$ times. The resulting Algorithm~\ref{alg:dueling} is the quantum dueling algorithm we proposed in this paper. A circuit diagram can be found in Fig.~\ref{fig:Dueling}. 

\begin{figure}
    \centering
    \includegraphics[width=\columnwidth]{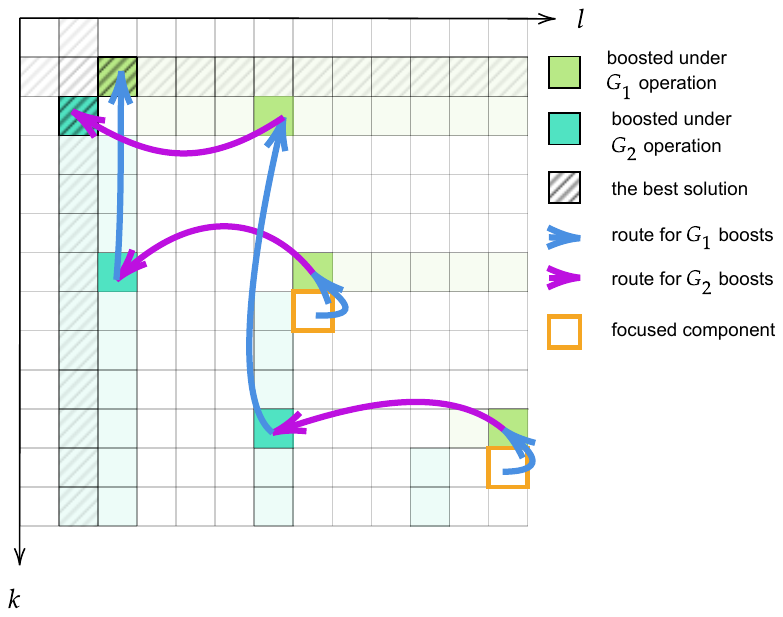}
    \caption{Some possible boosting routes for the given target components after the iterative application of dueling operators $\mathcal{G}_{1\leftarrow 2}$ and $\mathcal{G}_{2\leftarrow 1}$. The problem setup is the same as Fig.~\ref{fig:blank_grid}. The routes for amplitude boost under $\mathcal{G}_{1\leftarrow 2}$ operations are represented by blue arrows, while those representing amplitude boost under $\mathcal{G}_{2\leftarrow 1}$ operations are represented by violet arrows. An arbitrary non-highlighted component's amplitude can be boosted to the components signifying the best solution after a few iterative operations, even in the worst case. When $\{\alpha_{i}\}$ and $\{\beta_{i}\}$ take values larger than 1, the length of the boosting routes is generally shortened so that their amplitudes can be boosted to the components corresponding to the best solutions quickly. }
    \label{fig:dueling_arrow}
\end{figure}

\begin{algorithm}
        \caption{Quantum Dueling}\label{alg:dueling}
        \SetKw{False}{false} 
        \SetKw{True}{true}
        \KwIn{Measure function $v$ represented in quantum arithmetic; two integer sequences of at least length $p$: \{$\alpha_i$\} and \{$\beta_i$\}}
        \KwOut{An element $x$ in the search space satisfying $f(x)=1$ and $v(x)$ (approximately) minimized}
        $\ket{\psi}\leftarrow H^{\otimes2n}\ket{0}$\tcp*{initialize state vector}
        \For{$i \leftarrow 1$ \KwTo $p$}{
        $\ket{\psi}\leftarrow\mathcal{G}_{1\leftarrow 2}^{\alpha_i} \ket{\psi}$\tcp*{update the first register}
          $\ket{\psi}\leftarrow\mathcal{G}_{2\leftarrow 1}^{\beta_i} \ket{\psi}$\tcp*{update the second register}
        }
        Measure both registers. Let the result in the first register be $x_1$ and the result of the second register be $x_2$. Output the better result.
\end{algorithm}

\begin{figure}
    \centering
    \includegraphics[width=\columnwidth]{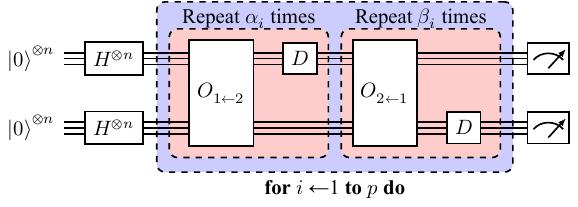}
    \caption{Quantum dueling circuit visualized.}
    \label{fig:Dueling}
\end{figure}

The route for a given component boosted under the iterative application of dueling operators is shown in Fig.~\ref{fig:dueling_arrow}, where $\alpha_i = \beta_i = 1$. The intuitive understanding is that each time we use the first register to boost states in the second register and then use the second register to boost states in the first register. In combination, the state will evolve in the desirable direction. 

It is important to note here that at any stage after initialization, the general state vector $\ket{\psi}$ remains in a pure state, while each register is not necessarily pure, as the two registers become increasingly entangled. Therefore, we need to use density operators to represent the mixed state of each individual register:
\begin{equation}
    \label{eq:density_operator_register}
    \begin{cases}
        \rho_{1} = \sum\limits_{l=1}^{N} \ket{\psi_{l}}\bra{\psi_{l}} \\
        \rho_{2} = \sum\limits_{k=1}^{N} \ket{\psi_{k}}\bra{\psi_{k}} \\
    \end{cases}
\end{equation}







\section{Analysis} 

\label{sec:analysis}

In this section, we will dive into several mathematical properties of quantum dueling, and conduct simulations of dueling over various solution distributions based on classical hardware. 

To provide convenience for both mathematical analysis and computer simulation, we may want to figure out the recursive formula of quantum dueling on each component of the augmented Hilbert space $\mathscr{H}$.

In Eq.~(\ref{eq:expand_full}), we use indices $kl$ to represent the components of a vector in $\mathscr{H}$. For an operator such as $\mathcal{G}_{1\leftarrow 2}$, its components should thus be denoted via four index symbols in the format of $klk'l'$. Under that tradition, the evolution of the state vector after applying gates $\mathcal{G}_{1\leftarrow 2}$ and $\mathcal{G}_{2\leftarrow 1}$ can be put as
\begin{equation}
    \label{eq:index_decomposition}
    \left\{\begin{gathered}
         \left(\mathcal{G}_{1\leftarrow 2} |\psi\rangle\right)_{kl} = \sum_{k'=1}^N\sum_{l'=1}^{N} \mathcal{G}_{1\leftarrow 2,klk'l'} \psi_{k'l'} \\
        \left(\mathcal{G}_{2\leftarrow 1} |\psi\rangle\right)_{kl} = \sum_{k'=1}^N\sum_{l'=1}^{N} \mathcal{G}_{2\leftarrow 1,klk'l'} \psi_{k'l'}  \\
    \end{gathered}\right.
\end{equation}
where by Eqs.~(\ref{eq:redefinition_of_dueling}) and (\ref{eq:index_decomposition}), components of $\mathcal{G}_{1\leftarrow 2}$ and $\mathcal{G}_{2\leftarrow 1}$ can be written as 
\begin{equation}
    \label{eq:original_dueling}
    \left\{\begin{gathered}
        \mathcal{G}_{1\leftarrow 2,klk'l'} = \delta_{ll'}\left(\frac{2}{N}-\delta_{kk'}\right)o(k',l') \\
        \mathcal{G}_{2\leftarrow 1,klk'l'} = \delta_{kk'}\left(\frac{2}{N}-\delta_{ll'}\right)o(l',k')
    \end{gathered}\right.
\end{equation}
where $\delta$ denotes the Kronecker delta, and function $o$ is defined previous satisfying $(-1)^{f(x)\&[v(x)<v(y)]} = o(x,y)$ for all $x,y\in S$.

\subsection{General Understanding}
\label{sec:general_understanding}

Eqs.~(\ref{eq:index_decomposition}) and (\ref{eq:original_dueling}) provide a mechanism for a simulation algorithm on the basis of a classical computer that will build the groundwork of the analysis of quantum dueling. Using the two formulae directly under the matrix multiplication schema gives us an algorithm with time complexity of $O(N^4)$ per Grover iteration via classical simulation.

Such an algorithm can be easily optimized. Combining and simplifying Eqs.~(\ref{eq:index_decomposition}) and (\ref{eq:original_dueling}) guarantees a simulation strategy of $O(N^3)$ per Grover iteration, as seen in the following equation:
\begin{equation}
    \label{eq:dueling_speedup_to_O(N^3)}
    \left\{\begin{gathered}
        \left(\mathcal{G}_{1 \leftarrow 2}|\psi\rangle\right)_{kl} = \sum_{k'=1}^{N}\left(\frac{2}{N}-\delta_{kk'}\right)o(k',l)\psi_{k'l} \\
        \left(\mathcal{G}_{2 \leftarrow 1}|\psi\rangle\right)_{kl} = \sum_{l'=1}^{N}\left(\frac{2}{N}-\delta_{ll'}\right)o(l',k)\psi_{kl'} 
    \end{gathered}\right.
\end{equation}

Fig.~\ref{fig:initial_simulation} documents the result of this simulation for a near-uniform solution distribution. An intuitively chosen parameter, $\alpha_i = \beta_i = 1$, was chosen as a starting point of analysis. \footnote{It must be noted that the algorithm we used employed slightly different technique due to historical developments of this research project. We assumed $\alpha_i = \beta_i = 1$ and combined the alternating two steps as a single gate. In this scheme, the brute force method has a $O(N^6)$ complexity, which can be subsequently simplified to $O(N^4)$.}

\begin{figure*}
    \centering
    \begin{subfigure}[b]{0.48\textwidth}
        \includegraphics[trim={1.5cm 0 3.7cm 1cm},width=\textwidth,clip]{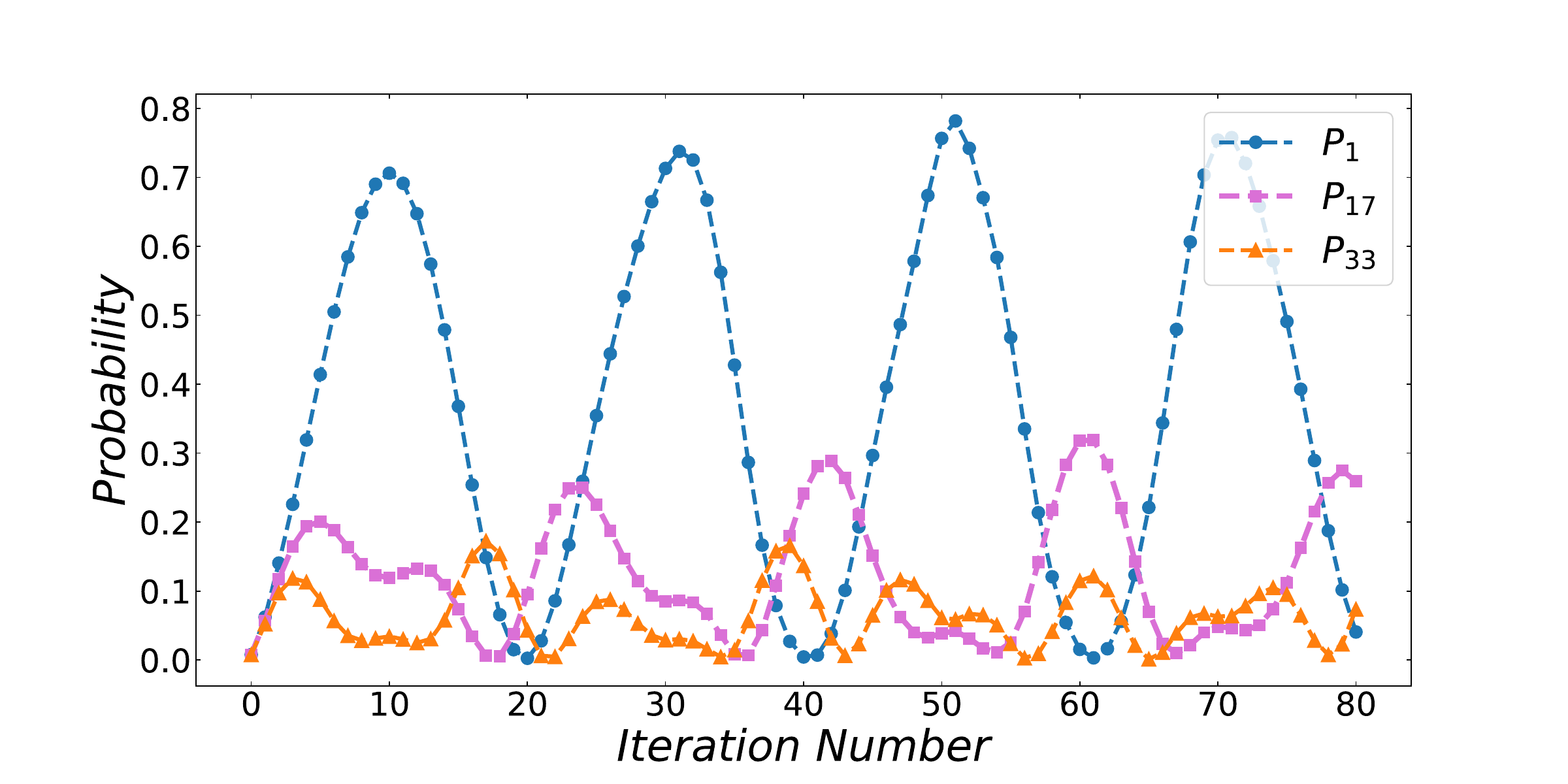}
        \caption{}
        \label{fig:initial_simulation_combined_evolution}
    \end{subfigure}
    \begin{subfigure}[b]{0.48\textwidth}
        \includegraphics[trim={1.5cm 0 3.7cm 1cm},width=\textwidth,clip]{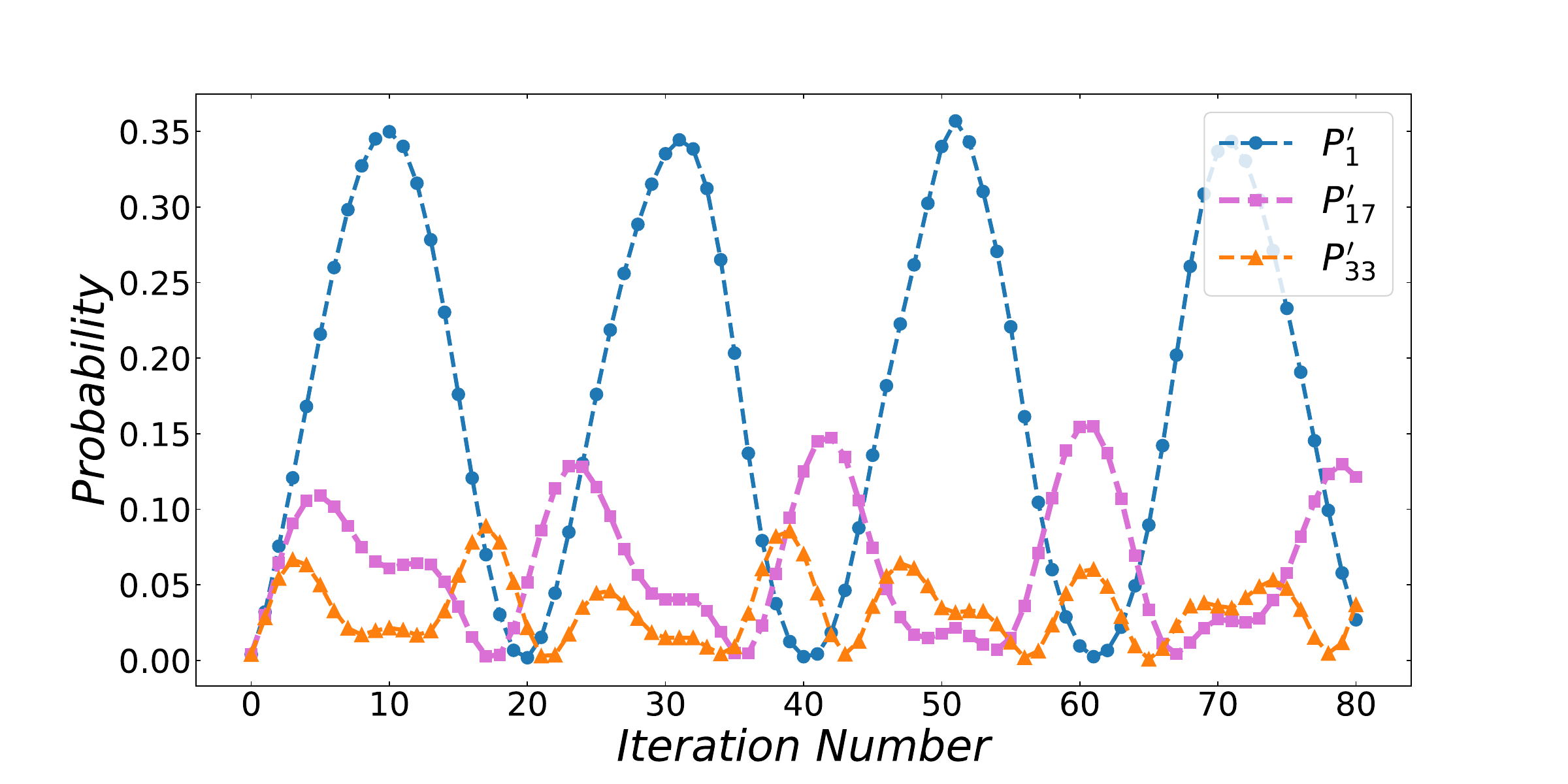}
        \caption{}
        \label{fig:initial_simulation_first_evolution}
    \end{subfigure}
    \begin{subfigure}[b]{0.48\textwidth}
        \includegraphics[width = \textwidth]{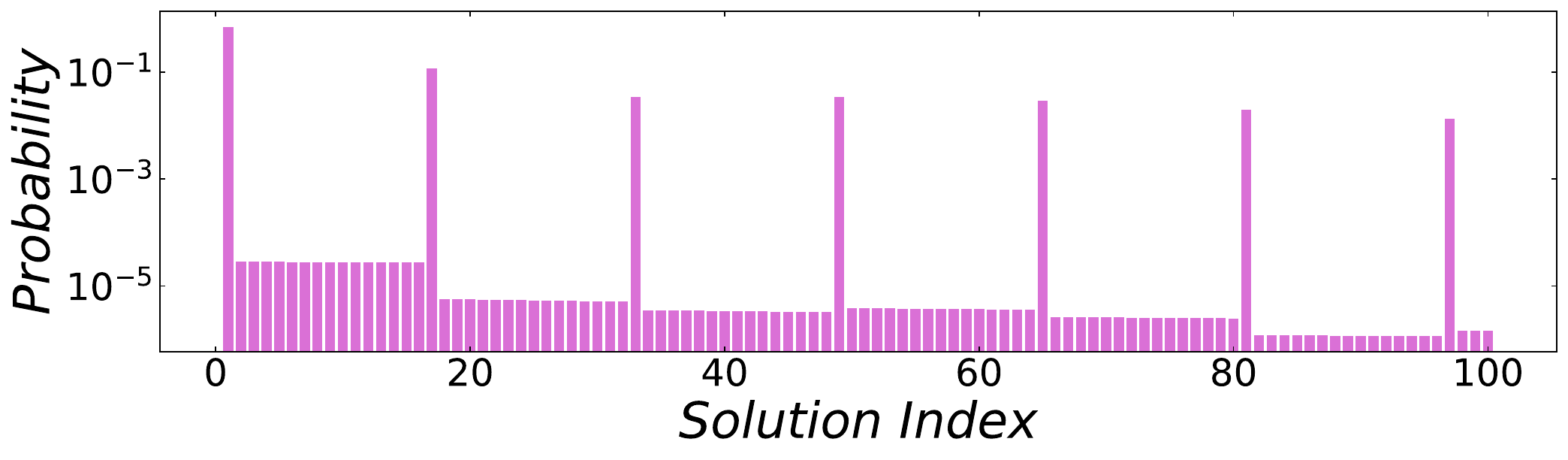}
        \caption{}
        \label{fig:initial_simulation_combined_distribution}
    \end{subfigure}
    \begin{subfigure}[b]{0.48\textwidth}
        \includegraphics[width = \textwidth]{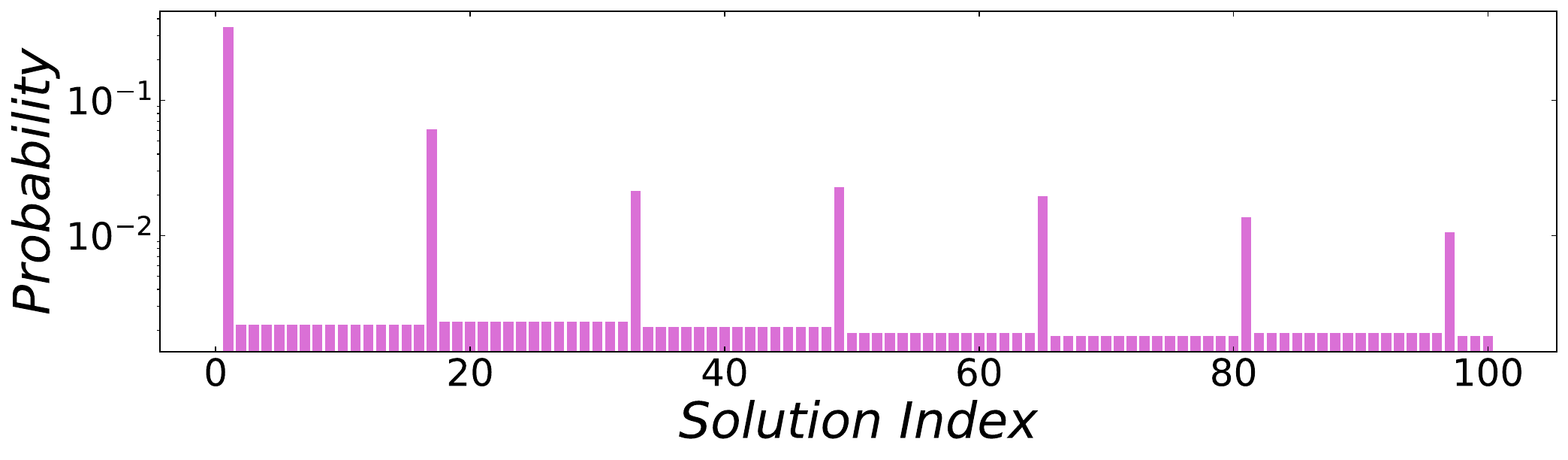}
        \caption{}
        \label{fig:initial_simulation_first_distribution}
    \end{subfigure}
    
    \caption{Performance of quantum dueling for a search space of $N = 256$ with $M = 16$ solutions, which are distributed as $v(x) = x$, $f(x) = [x \equiv 1 \pmod {16}]$, where $1\leqslant x \leqslant 256$. The parameters are chosen naively such that $\alpha_{i} = \beta_{i} =1$. (a) Success probability $P_1$ of finding the most optimized solution (at $x = 1$), along with that of the second (probability $P_{17}$ at $x=17$) and the third (probability $P_{33}$ at $x=33$) as a function of number of iterations performed.  
    (b) The last step of Algorithm~\ref{alg:dueling} combines results from both registers to produce an output. If instead we only measure the first register, a set of parameters similar to those in (a) and (c) can be used to quantify the performance of quantum dueling. These figures show the success probability $P'_1$ of finding the most optimized solution (at $x = 1$), along with that of the second (probability $P'_{17}$ at $x=17$) and the third (probability $P'_{33}$ at $x=33$) as a function of number of iterations performed, on the assumption that only the first register is measured. 
    (c) The (local) maximum success probability of measuring the most optimized solution is achieved after $i_{\text{max}} = 10$ iterations. The figure visualizes, in a logarithmic scale, the probability distribution of elements $1\leqslant x \leqslant 100$ after this iteration.
    (d) The local maximum success probability (after only measuring the first register) is achieved after $i'_{\text{max}} = 10$ iterations. The figure visualizes, in a logarithmic scale, the probability distribution (after only measuring the first register) of elements $1\leqslant x \leqslant 100$ after this iteration.
    }
    \label{fig:initial_simulation}
\end{figure*}


The probability of obtaining the most optimized solution after combining measurements from both registers, which we term the success probability, exhibits an oscillating trend similar to that of the Grover algorithm. However, several differences remain. First, in Figs.~\ref{fig:initial_simulation_first_evolution} and \ref{fig:initial_simulation_combined_distribution}, the local maximum of $P_1$ and $P'_1$ vary from peak to peak. For given $\{\alpha_i\}$ and $\{\beta_i\}$ arrays, we would like to set the termination count $p$ such that the first local maximum is reached. Second, it is clear that unlike Grover algorithm, such a local maximum is not asymptotically close to $1$. To compensate this problem, we can repeat the algorithm to boost success probability. In principle, if known information is not sufficient to determine $p$, we can vary rotation count for each repetition just as done in \cite{Brassard_2002}, though this option is not of great interest to us. Overall, Fig.~\ref{fig:initial_simulation} suggests that it is difficult to find an exact solution for the state vector in quantum dueling. However, an approximate solution might be obtained with appropriate mathematical expertise. 

Since Fig.~\ref{fig:initial_simulation} only documents a single case, we ran the simulation for a variety of scenarios, on the assumption that $\alpha_i = \beta_i = 1$. Table~\ref{tab:initial_simulation} and Fig.~\ref{fig:initial_simulation_all_M} summarize the performance of quantum dueling for different distributions of solution and different values of $M$, while $N$ remains fixed. The first local maximum success probabilities and the corresponding iteration counts vary. For uniform distribution, as in Table~\ref{tab:initial_simulation}, the performance of quantum dueling is hindered by the increasing number of non-solution elements with smaller measure function $v$ values than the most optimized solution, but the resulting success probability remains valid for efficient algorithm design. A better selection of parameters $\{\alpha_i\}$ and $\{\beta_i\}$ can compensate for this issue. We will tackle this problem in \ref{sec:parameter}. 

\begin{table}[b]
    \centering
    \begin{ruledtabular}
        \begin{tabular}{c c c c c}
            $f(x)$ & $P_{\text{max}}$ & $p_{\text{max}}$ & $P'_{\text{max}}$ & $p'_{\text{max}}$ \\
            \colrule
            $[x \equiv 1 \pmod{16}]$ & 0.7061 & 10 & 0.3498 & 10 \\
            $[x \equiv 8 \pmod{16}]$ & 0.4497 & 8 & 0.2257 & 8 \\
            $[x \equiv 0 \pmod{16}]$ & 0.2730 & 5 & 0.1399 & 5 \\
            $[x \leqslant 16]$ & 0.0903 & 2 & 0.0549 & 2 \\
            $[x > 240]$ & 0.0112 & 1 & 0.0056 & 1 \\
            $[x=1 \vee x>241]$ & 0.9919 & 8 & 0.5035 & 8 \\
        \end{tabular}
    \end{ruledtabular}
    \caption{$P_{max}$, $p_{max}$, $P'_{max}$ and $p'_{max}$ with respect to different solution distributions.}
    \label{tab:initial_simulation}
\end{table}

\begin{figure}
    \centering
    \begin{subfigure}{0.5\columnwidth}
        \centering
        \includegraphics[width = \textwidth]{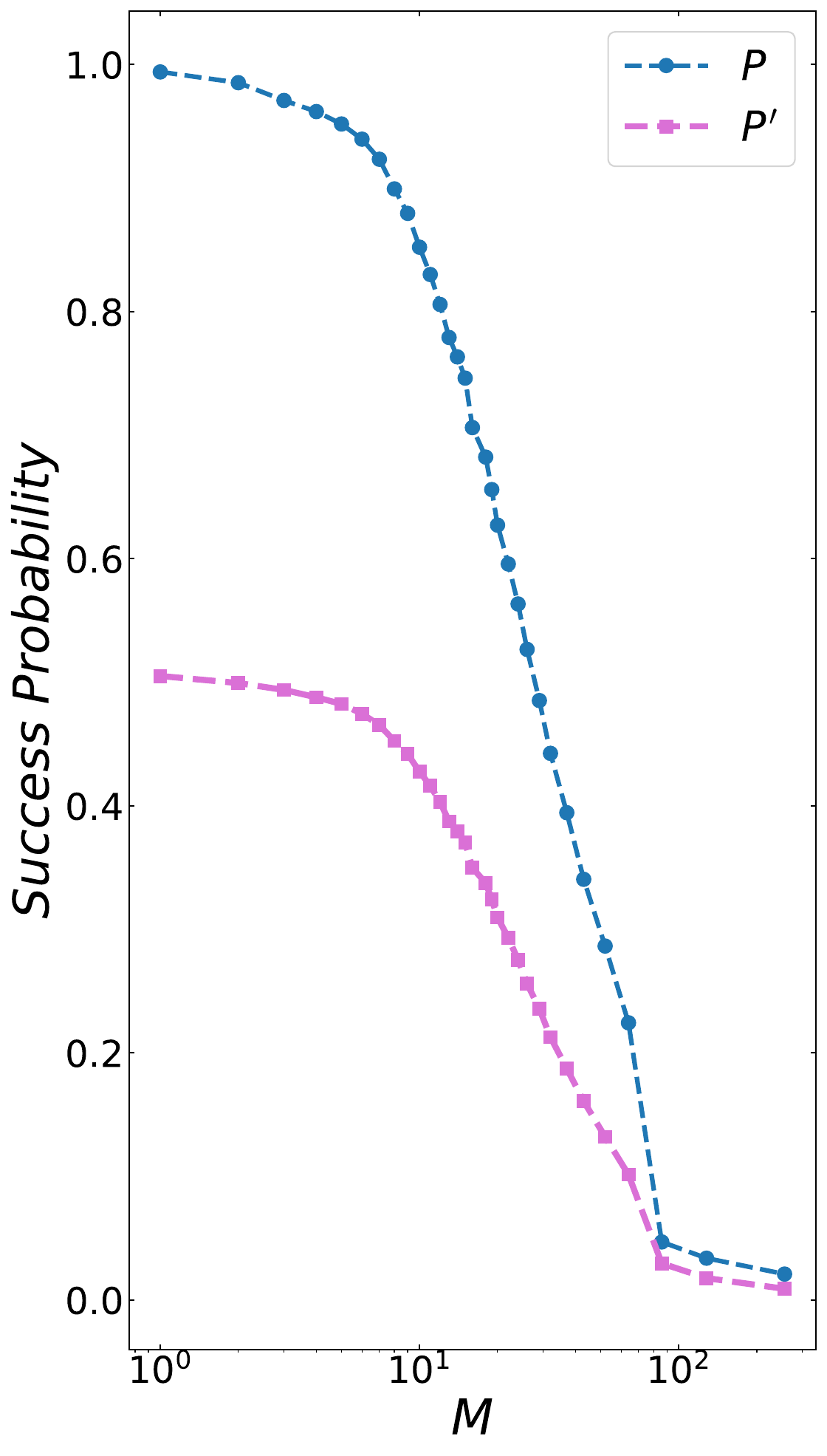}
        \caption{}
        \label{fig:success_probability_all_M}
    \end{subfigure}%
    \begin{subfigure}{0.5\columnwidth}
        \includegraphics[width = \textwidth]{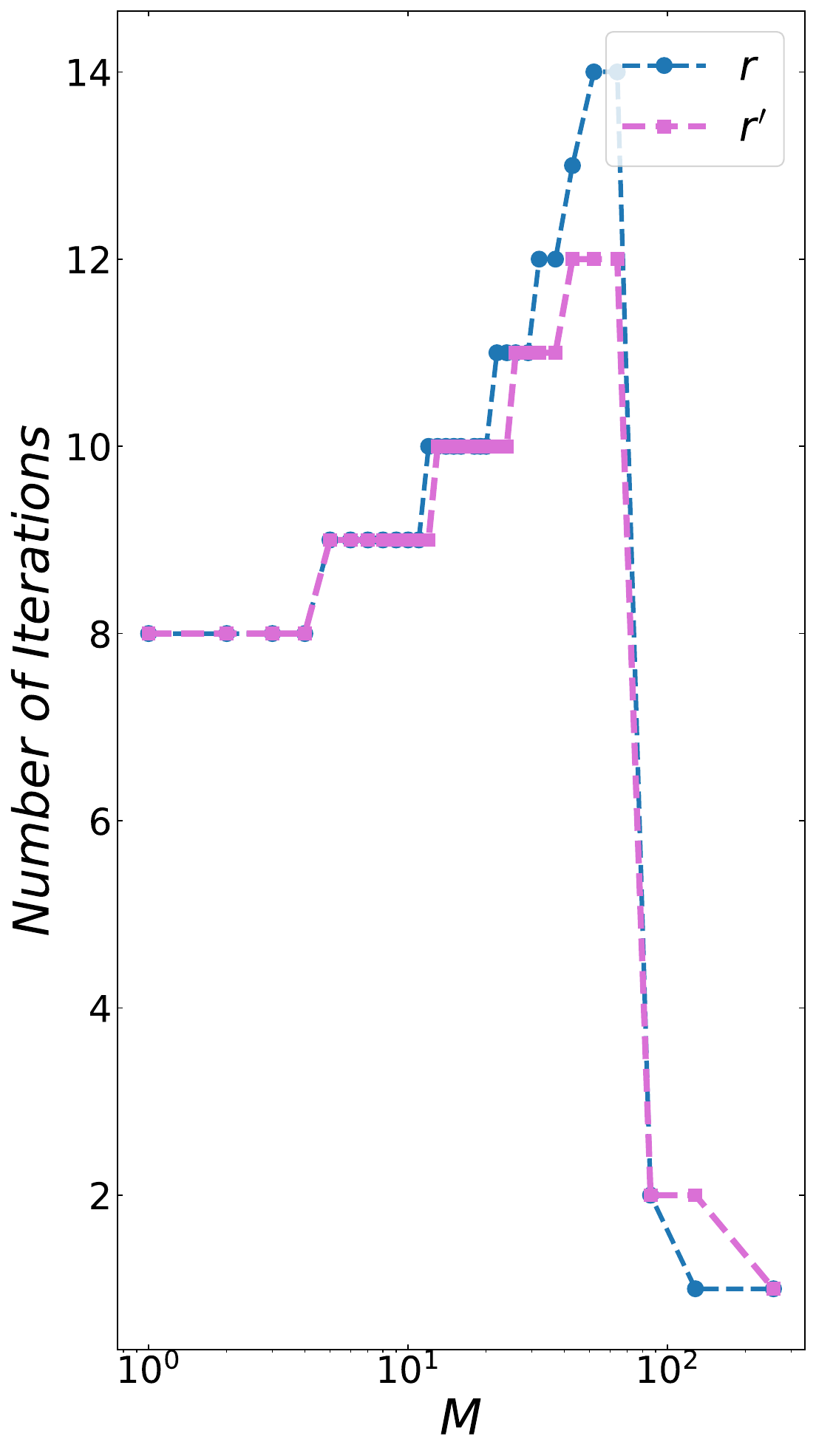}
        \caption{}
        \label{fig:iteration_count_all_M}
    \end{subfigure}
    \caption{Performance of quantum dueling for different $M$, assuming $N = 256$, $v(x) = x$, and a near uniform solution distribution $f(x) = \left[x - 1 \equiv 0 \pmod{\left\lceil\frac{N}{M}\right\rceil}\right].$ (a) The first local maximum success probability $P_{\text{max}}$ and the first local maximum first-register success probability $P'_{\text{max}}$, as a function of $M$, which is in logarithmic scale. (b) The number of iterations $p_{\text{max}}$ after which $P_{\text{max}}$ is achieved, and the number of iterations $p'_{\text{max}}$ after which $P'_{\text{max}}$ is achieved, as a function of $M$, which is in logarithmic scale.  More details regarding this formula can be found in Appendix~\ref{app:distribution}.}
    \label{fig:initial_simulation_all_M}
\end{figure}

For non-uniform distributions, quantum dueling behaves poorest when all solutions have $v$ values that are within a range in which little or no non-solutions lie, with success probability dropping to less than $10\%$. The worst case happens when all non-solutions take $v$ values greater than those of all solutions. In this case, the components of the state vector that represent non-solutions in both registers take most of the weight of the state vector, and such weight cannot be transferred to solutions via amplitude amplification. On the other hand, if there is a clear separation between the desired best solution and other solutions, i.e., between measure function values $v$ of the most optimized solution and that of the second most optimized solution lies all or most of the non-solutions, the algorithm achieves the best result, reaching a success probability of over $99\%$. This means that the maximum probability of quantum dueling depends on the separation between solutions in the distribution of measurement function values.

Fig.~\ref{fig:initial_simulation_all_M} documents the maximum success probability and the required number of iterations for uniform solution distribution when we slowly increase the number of solutions $M$ and decrease their separation. We can see success probability slowly drops, and the required number of iterations slowly increases for an increasing $M$. However, after a critical point $M_0=\frac{N}{4}$, both success probability and (especially) corresponding iteration counts drop significantly, indicating a collapse in the efficacy of the algorithm. As an explanation, for each $y$, define a sequence of operators $G_{y} = DO_{y}$. The result of dueling gates can be decomposed as:
\begin{equation}
    \label{eq:dueling_operation_in_expansion_with_partial_Grover_gates}
    \begin{cases}
        \mathcal{G}_{1\leftarrow 2}\ket{\psi} = \sum\limits_{l=1}^{N}\left(G_l\ket{\psi_l}\right)\ket{l} \\
        \mathcal{G}_{2\leftarrow 1}\ket{\psi} = \sum\limits_{k=1}^{N}\ket{k}\left(G_k\ket{\psi_k}\right) \\
    \end{cases}
\end{equation}
Given $y\in S$ and an initial state vector in $\mathscr{H}_S$ in uniform distribution, assuming the size of target space is $m$, by \cite{Nielsen_2018}, the preferred number of repetitions of $\{G_{y}\}$ gates to reach $100\%$ success probability $r_{\text{max}}$ must satisfy
\begin{equation}
    \label{eq:Grover_r_max}
    (2r_{\text{max}}+1)\theta = \frac{\pi}{2}
\end{equation}
where $\theta = \arcsin\sqrt{\frac{m}{N}}$. When $m > M_0 = \frac{N}{4}$, Eq.~(\ref{eq:Grover_r_max}) suggests $r_{\text{max}} < 1$, meaning that after one iteration, the success probability has been moved past its peak, breaking the amplitude amplification mechanism of Grover algorithm. Consider dueling gates applied to an initial state vector $H^{\otimes 2n}\ket{0}$. When $M > M_0$, for significant values of $y$, $m$ will be sufficiently close to or greater than $M_0$ such that $G_y$ will no longer boost the weight of the target. Therefore, a collapse is observed for quantum dueling, at least with $\alpha_i = \beta_i = 1$.

The results of such simulations justify several definitions made in previous sections. First, when we defined the problem of combinatorial optimization in Sec.~\ref{sec: algorithm}, we introduced the $f$ function to mark out solutions against non-solutions. However, one can also define combinatorial optimization without using the $f$ function by forcing the measure function $v$ of non-solutions to be infinity. In this case, we essentially have $M = N$, leading to the collapse of efficacy as shown in Fig.~\ref{fig:initial_simulation_all_M}. For optimization problems without $f$ or cases where $M>M_0$, we can augment the search space by introducing several extra qubits and mark out an element as  solution if the extra qubits take some arbitrarily chosen value. Alternatively, we can separate solutions into different groups and search these groups separately.

Second, going back to the design of Algorithm~\ref{alg:dueling}, one might want to modify $O_x$ such that it selects for element that are better than $x$, i.e., 
\begin{equation}
    \label{eq:alternate_oracle}
    o(x,y) = (-1)^{[f(x) > f(y)] \,\mathbin{|}\, \left(\,[f(x) = f(y)] \, \mathbin{\wedge} \, [v(x) < v(y)] \, \right)}
\end{equation}
where $|$ is the equivalent of logical disjunction on $\{0,1\}$. However, this redefinition would produce a total order in $S$, meaning that we can find functions $v'(x)$ and $f'(x) = 1$ $\forall x \in S$ such that the algorithm after redefinition is equivalent to the original Algorithm~\ref{alg:dueling} with these new functions $v'$ and $f'$. This, in turn, implies $M = N$, so we would see our algorithm collapsing in efficacy under naive parameters $\alpha_{i} = \beta_{i} = 1$. This is why we define the oracle as it was in Sec.~\ref{sec: algorithm}.

Despite the problems described above, in most common cases where the solution distributions are not that extreme, and the number of solutions is controlled below the critical point $M_{0}=\frac{N}{4}$, quantum dueling with naively chosen parameters $\alpha_{i} = \beta_{i} = 1$ performs quite well, increasing the success probability to a constant regardless of the increase of $N$. 

Under the extreme cases where the number of solutions exceeds the critical limit, or all of the solutions are distributed higher than most of the non-solutions in the measure function spectrum, we can add a few more ancilla qubits, making the entire Hilbert space larger and thus making the solutions dilute. This treatment will be analyzed in Sec.~\ref{sec:opnl}.

Another promising way is adjusting the variational parameters, along with classical optimization heuristics, which will be discussed in detail in Sec.~\ref{sec:parameter}.

\subsection{Cluster Representation}
\label{sec:cluster_representation}

Analysis in Sec.~\ref{sec:general_understanding} suggests that it is difficult, yet potentially possible, to quantify the behavior of quantum dueling mathematically by carefully studying the behavior of the system using the recursive formulas in Eq.~\ref{eq:dueling_speedup_to_O(N^3)}. While we have not yet solved such evolution, this section provides some insights to simplify the formulation of quantum dueling. 

We observe in Fig.~\ref{fig:initial_simulation_first_distribution} that for a given first-register component $|\psi_l\rangle$, in a $v$ range where no solutions lie, the non-solutions seem to have the same weight within the total state vector. Additionally, if several solutions share equal measure function values, they also have the same weight at any stage. To formalize these conclusions, we define the following relation ${\sim}$ on $S$ satisfying Definition~\ref{def:cluster_relation}.

\begin{definition}
    \label{def:cluster_relation}
    ${\sim}$ is an equivalence relation on $S$ such that for all $x,y \in S$, $x{\sim} y$ if and only if either of the following conditions is true:
    \begin{enumerate}
        \item $f(x) = f(y) = 1 \wedge  v(x) = v(y)$;
        \item $f(x) = f(y) = 0 \wedge \neg  \exists\;  z \in S \; f(z) = 1 \wedge \min\left\{v(x),v(y)\right\} \leqslant v(z) < \max\left\{v(x),v(y)\right\}$. 
    \end{enumerate}
\end{definition}

\begin{figure}
    \centering
    \includegraphics[width=\columnwidth]{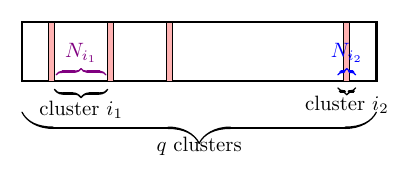}
    \caption{A graphic illustration of the solution space $S$ under cluster representation when all the elements are sorted. The cluster $i_{1}$ is an example of all the non-solutions with their measure function values lying between two adjacent solutions being classified together. The cluster $i_{2}$ illustrates a case in which all the solutions with the same measure function values are classified as one. The total $q$ clusters make up the quotient search space $S/\sim$.}
    \label{fig:cluster_illustration}
\end{figure}

We verify that ${\sim}$ is an equivalence relation in Appendix~\ref{app:cluster}. Intuitively, ${\sim}$ produces equivalence classes containing either all solutions with the same $v$ value or all non-solutions between two adjacent solutions when all elements are sorted by $v$ value, which we term ``clusters''. We can observe this intuitively by Fig.~\ref{fig:cluster_illustration}. Our observation from Fig.~\ref{fig:initial_simulation_first_distribution} thus becomes Theorem.~\ref{thm:cluster_property}, which states elements from the same equivalence class take the same weight in the state vector of quantum dueling. We leave its proof in Appendix.~\ref{app:cluster}.
\begin{theorem}
    \label{thm:cluster_property}
    Let $k_1,k_2,l_1,l_2 \in S$ satisfying $k_1 {\sim} k_2$ and $l_1 {\sim} l_2$. At any stage after initialization, $\langle k_1l_1| \psi \rangle = \langle k_2l_2 | \psi \rangle$.
\end{theorem}

 To formulate the evolution of the state vector on the language of clusters, we need to consolidate our understanding of their properties. It is clear the $f$ values of elements from the same cluster are equivalent, so we can allow $f$ to take in a cluster as the argument, outputting the $f$ value of the cluster's elements. While the $v$ values of elements within a cluster might be distinct, for a $\mathcal{S} \in S/\sim$, we can use $\max_v \mathcal{S}$ and $\min_v \mathcal{S}$ to denote the minimum and maximum $v$ value. The details regarding this construction can be found in Appendix~\ref{app:cluster}.

For simplicity, we would like to sort all the clusters and index them accordingly. We first define total order $\leqslant_{v}$ and $\leqslant_{f}$ to rank the $f$ and $v$ values of corresponding clusters. Thus, we are able to define $\leqslant_{\sim}$ as in Definition~\ref{def:v_sim_orders}. The detailed statements and proofs are provided in Appendix~\ref{app:cluster}.
\begin{definition}
\label{def:v_sim_orders}
$\leqslant_\sim$ is the lexicographic combination of $\leqslant_v$ and $\leqslant_f$. That is, for $\mathcal{S},\mathcal{T} \in S/{\sim}$, $\mathcal{S} \leqslant_\sim \mathcal{T}$ if and only if $\mathcal{S} <_v \mathcal{T}$ (that is, $\neg \mathcal{T} \leqslant_v \mathcal{S})$) or ($\mathcal{S} \sim_v \mathcal{T}$ and $\mathcal{S} \leqslant_f \mathcal{T}$).
\end{definition}

Therefore, we could use the relation $\leqslant_{\sim}$ to sort all the clusters by indexing them from 1 to $q=\left|S/\sim\right|$. Such relabelling gives rise to a new, contracted basis for the state vector formulated in Definition~\ref{def:new_basis}. For clarity, all indices under this basis will be denoted in Greek letters.

\begin{definition}
    \label{def:new_basis}
    Consider $\mathcal{S} \in S/{\sim}$ and $\xi = \idx \mathcal{S}$. Function $\idx$ is a bijection from $(S/{\sim})  $ to $[n]$ (set of integers from $1$ to $n$), such that for $\mathcal{S},\mathcal{T}$, $\idx \mathcal{S} \leqslant \idx \mathcal{T}$ if and only if $\mathcal{S} \leqslant_\sim \mathcal{T}$. 
    
    We then Define $N_\xi$ and $|\xi\rangle$ as:
    \begin{equation}
        \label{eq:new_basis}
        \left\{\begin{gathered}
            N_\xi = |\mathcal{S}| \hfill \\
            |\xi\rangle = \frac{1}{\sqrt{N_\xi}} \sum_{x\in\mathcal{S}}|x\rangle \hfill
        \end{gathered}\right.
    \end{equation}
\end{definition}

In this case, the Hilbert Space $\mathscr{H}_{S}$ formulated in Sec.~\ref{sec: algorithm} will give rise to a subspace $\mathscr{H}_{\sim}$ such that $\{\ket{\xi}\}$ is a basis. Theorem~\ref{thm:cluster_property} suggests that at any stage after initialization in Algorithm~\ref{alg:dueling}, $\ket{\psi} \in \mathscr{H}_{\sim} \otimes \mathscr{H}_{\sim}$.

Under the new basis, we can use $\kappa$ and $\lambda$ to index the basis vectors from the first and the second register. Let $q = \left|S/{\sim}\right|$, Eqs.~(\ref{eq:expand_full}), (\ref{eq:redefinition_of_state_Vector}) thus become:
\begin{equation}
    \label{eq:expand_cluster}
    \left\{\begin{gathered}
        |\psi\rangle = \sum_{\lambda=1}^q|\psi_\lambda\rangle  |\lambda\rangle \hfill \\
        |\psi\rangle = \sum_{\kappa=1}^q|\kappa\rangle |\psi_\kappa\rangle \hfill \\
        |\psi\rangle = \sum_{\kappa=1}^{q}\sum_{\lambda=1}^{q} \psi_{\kappa \lambda} |\kappa \lambda \rangle \hfill
    \end{gathered}\right.
\end{equation}

Meanwhile, after applying the Hadmard gates in Algorithm~\ref{alg:dueling}, the state vector is:
\begin{equation}
    \label{eq:initial_state_cluster}
    |\psi_0\rangle = \sum\limits_{\kappa=1}^{q} \sum\limits_{\lambda=1}^{q} \frac{\sqrt{N_{\kappa}N_{\lambda}}}{N} |\kappa \lambda \rangle
\end{equation}

To obtain the formula for state evolution similar to what we have done in Eq.~(\ref{eq:dueling_speedup_to_O(N^3)}), we first find an expression for the criterion of the oracle, $[v(x)<v(y)]\&f(x)$, where we allow $f$ to take the index of a cluster (denoted in Greek letter) as an input, outputting the $f$ value of that cluster. We will list the detailed proof in Appendix~\ref{app:cluster}.

To formulate the oracle in terms of clusters, we define $O_\xi$ such that for $x \in \mathcal{S}$, where $\mathcal{S}$ satisfies $\idx \mathcal{S} = \xi$, $O_\xi = O_x$. Thus:
\begin{equation}
    \label{eq:oracle_cluster}
    O_\xi |\eta\rangle = (-1)^{f(\eta) \& [\eta < \xi]} |\eta\rangle
\end{equation}
where the $(-1)^{f(\eta) \& [\eta < \xi]}$ term can be simplified to $o(\eta, \xi)$. We call this process contracting the oracle.

Similarly, we then contract the Grover gates $\left\{G_x\right\}$ into $\left\{G_\xi\right\}$ as:
\begin{equation}
    \label{eq:Grover_formula_cluster}
    G_\xi = \left(2|m\rangle\langle m| - I\right)O_\xi
\end{equation}
where $|m\rangle$, the mean state vector, can be decomposed in the cluster basis:
\begin{equation}
    \label{eq:mean_cluster}
    |m\rangle = \sum_{\eta = 1}^q \sqrt{\frac{N_{\eta}}{N}} |\eta\rangle 
\end{equation}

It is intuitive to rewrite Eq.(\ref{eq:dueling_operation_in_expansion}) as:
\begin{equation}
    \label{eq:Grover_on_1&2_cluster}
    \left\{\begin{gathered}
        \mathcal{G}_{1\leftarrow 2} |\psi\rangle = \sum_{\lambda = 1}^{q} G_{\lambda} \ket{\psi_{\lambda}} \otimes \ket{\lambda} \hfill \\
        \mathcal{G}_{2\leftarrow 1} |\psi\rangle = \sum_{\kappa = 1}^{q} \ket{\kappa} \otimes G_{\kappa} \ket{\psi_{\kappa}} \hfill 
    \end{gathered}\right.
\end{equation}

To show Eq.~(\ref{eq:Grover_on_1&2_cluster}) is correct, we notice that there is only a constant term difference between corresponding $|\psi_l\rangle$ and $|\psi_\lambda\rangle$ and between corresponding $|\psi_k\rangle$ and $|\psi_\kappa\rangle$. By linearity of quantum gates, there should be no difference in the formula for state evolution, i.e., Eq.~(\ref{eq:Grover_on_1&2_cluster}).

As we have done in Sec.~\ref{sec:general_understanding}, we combine Eqs.(\ref{eq:oracle_cluster})--(\ref{eq:Grover_on_1&2_cluster}) to get the most simplified equation for state evolution:
\begin{equation}
    \label{eq:cluster_state_evolution}
    \left\{\begin{gathered}
        \left(\mathcal{G}_{1\leftarrow 2}|\psi\rangle\right)_{\kappa\lambda} = 2\sqrt{\frac{N_{\kappa}}{N}} h_\lambda - o (\kappa,\lambda) \psi_{\kappa \lambda} \hfill \\
        \left(\mathcal{G}_{2\leftarrow 1}|\psi\rangle\right)_{\kappa\lambda} = 2\sqrt{\frac{N_{\lambda}}{N}} h_\kappa -o(\lambda,
    \kappa) \psi_{\kappa \lambda} \hfill \\
    \end{gathered}\right.
\end{equation}
where values $h_\kappa$ and $h_\lambda$ can be pre-processed as:
\begin{equation}
    \label{eq:cluster_pre_processed}
    \left\{\begin{gathered}
        h_\lambda = \sum_{\kappa'=1}^{q}\sqrt{\frac{N_{\kappa'}}{N}} o (\kappa',\lambda) \psi_{\kappa' \lambda} \hfill \\
        h_\kappa = \sum_{\lambda'=1}^{q} \sqrt{\frac{N_{\lambda'}}{N}} o(\lambda', \kappa) \psi_{\kappa \lambda'}
    \end{gathered}\right.
\end{equation}

In total, Eqs.~(\ref{eq:initial_state_cluster}), (\ref{eq:cluster_state_evolution}), and (\ref{eq:cluster_pre_processed}) allow us to simplify quantum dueling to the most contracted form. In Grover algorithm, such a strategy leads to only 2 clusters: one made up of solutions, and one made up of non-solutions. This would allow us to solve the state vector after an arbitrary number of iterations \cite{Brassard_2002}. In quantum dueling, however, the simplified form remains complex, hindering efforts to find a quantitative expression. Nevertheless, we are now guaranteed with a $O(q^2)$ per iteration simulation strategy, which will prompt us to understand quantum dueling from a new and comprehensive perspective.

\subsection{Hamiltonian Version of Dueling}
\label{sec:hamiltonian_version}


Aside from its circuit interpretation using unitary quantum gates, \cite{Roland_2003, Barkoutsos_2020} have shown that Grover search can be transplanted into a Hamiltonian setup, similar to Hamiltonian-based algorithms like QAOA. As an algorithm inspired by Grover search, quantum dueling can also be represented by the time evolution of a series of Hamiltonians. 

Extending the Hamiltonian forumlation in \cite{Roland_2003, Barkoutsos_2020} into the augmented Hilbert space $\mathscr{H} = \mathscr{H}_{S} \otimes \mathscr{H}_{S}$, we can define our mixer Hamiltonians $H_{M}^{1\leftarrow 2}$, $H_{M}^{2\leftarrow 1}$ as follows: 
\begin{equation}
    \label{eq:mixer_Hamiltonian_def}
    \begin{cases}
        H_{M}^{1\leftarrow 2} = \left(I-\ket{m}\bra{m}\right) \otimes I \\
        H_{M}^{2\leftarrow 1} = I \otimes \left(I-\ket{m}\bra{m}\right) \\
    \end{cases}
\end{equation}
where $\ket{m} = \frac{1}{\sqrt{N}} \sum\limits_{x} \ket{x}$ is the equal superposition of all basis states in $\mathscr{H}_{S}$.

In this way, the diffusion operators $\mathcal{D}_{1\leftarrow 2}$, $\mathcal{D}_{2\leftarrow 1}$ in the augmented Hilbert space $\mathscr{H}$ will be generated as long as we choose the evolution time of each mixer Hamiltonian to be $\pi$, as already mentioned in \cite{Roland_2003, Barkoutsos_2020}.
\begin{equation}
    \label{eq:mixer_Hamiltonian_def2}
    \begin{cases}
        \mathcal{D}_{1\leftarrow 2} = e^{-i\pi H_{M}^{1\leftarrow 2}} \\
        \mathcal{D}_{2\leftarrow 1} = e^{-i\pi H_{M}^{2\leftarrow 1}} \\
    \end{cases}
\end{equation}

When it comes to the problem Hamiltonians $H_{P}^{1\leftarrow 2}$ and $H_{P}^{2\leftarrow 1}$, due to our specific oracle constructions, they will be the sum of projectors onto the basis $\ket{kl}$ in the augmented Hilbert space $\mathscr{H}$ that satisfy the corresponding constraints: $o(k,l)=-1$ and $o(l,k)=-1$, respectively.

Therefore, we can write:
\begin{equation}
    \label{eq:problem_Hamiltonian_def}
    \begin{cases}
        H_{P}^{1\leftarrow 2} = -\sum\limits_{k, l}^{o(k,l) = -1} \ket{kl}\bra{kl} \\
        H_{P}^{2\leftarrow 1} = -\sum\limits_{k, l}^{o(l,k) = -1} \ket{kl}\bra{kl} \\
    \end{cases}
\end{equation}

In this way, the oracles $\mathcal{O}_{1\leftarrow 2}$, $\mathcal{O}_{2\leftarrow 1}$ can be written in the form:
\begin{equation}
    \label{eq:problem_Hamiltonian_def2}
    \begin{cases}
        \mathcal{O}_{1\leftarrow 2} = e^{-i\pi H_{P}^{1\leftarrow 2}} \\
        \mathcal{O}_{2\leftarrow 1} = e^{-i\pi H_{P}^{2\leftarrow 1}} \\
    \end{cases}
\end{equation}


Evolving the quantum state $\ket{\psi}$ with respect to time, we obtain:
\begin{align}
\left. \prod_{i=1}^{p} \right( & \left(e^{-i \pi H_{M}^{2\leftarrow 1}} e^{-i\pi H_{P}^{2\leftarrow 1}} \right)^{\beta_{i}} \nonumber \\
& \times \left( e^{-i\pi H_{M}^{1\leftarrow 2}} e^{-i\pi H_{P}^{1\leftarrow 2}} \right)^{\alpha_{i}} \left. \vphantom{\prod_{i=1}^{p}} \right) \ket{\psi} 
\end{align}

As mentioned in \cite{chiang_2023}, QAOA can be viewed as a generalization of Grover search algorithm. The time parameters of the problem and the mixer Hamiltonians $H_{P}$, $H_{M}$ are fixed to be $\pi$ in Grover search, while QAOA introduces $2p$ additional degrees of freedom to the state evolution by converting the time parameters to two sets of variational parameters $\{\gamma_{i}\}$ and $\{\eta_{i}\}$ corresponding to the evolution of problem and mixer Hamiltonians, respectively, adding complexity to the quantum state evolution. The notation $p$ stands for the depth of the QAOA circuit, as a common practice starting from \cite{Farhi_2014}. More generally, the problem Hamiltonian need not be restricted to the sum of projectors over the target states. Instead, it could be any Hamiltonian tailored to the specific problem, and the mixer Hamiltonian also has different definitions in different QAOA versions. 

Quantum dueling, on the other hand, is a game changer. The generalizations that dueling makes to the original Grover search are completely orthogonal to those made by QAOA. To realize comparison, we augment the entire Hilbert space, double the number of qubits, and thus construct two groups of Hamiltonians. Evolving each group of Hamiltonians with respect to time, the result would be a controlled unitary operation of one register over the other. In this manner, we assign two new sets of parameters $\{\alpha_{i}\}$, $\{\beta_{i}\}$ to the system, one set for each group, making the algorithm variational. In quantum dueling, the variational parameters are defined in the ``outer layer'' of the Hamiltonians, while the parameters for the ``inner layer'' Hamiltonians remain fixed as $\pi$. Unlike QAOA, the parameters $\{\alpha_{i}\}$, $\{\beta_{i}\}$ have nothing to do with the evolution time of problem and mixer Hamiltonians; instead, they represent the evolution time of the `` averaged'' Hamiltonians $\overline{H_{1\leftarrow 2}}$, $\overline{H_{2\leftarrow 1}}$, if we use some Trotter-like strategies to combine the two Hamiltonians in one group into one Hamiltonian. The total parameter degrees of freedom is also $2p$ for quantum dueling, where $p$ is defined in Algorithm~\ref{alg:dueling}.

Interestingly, the fixed-point algorithm \cite{Grover_Lov_2005, Yoder_2014, Yan_2022} modifies the original Grover search so that the evolution time of each operator is not necessarily $\pi$. Yoder points out in \cite{Yoder_2014} that the quadratic speedup of Grover can be preserved along with the fixed-point properties, meaning that the success probability will be bounded within an error $\delta^{2}$ away from 1 asymptotically. 

This construction suggests that we may add new degrees of freedom to our inner layer of quantum dueling, changing the evolution time $\pi$ of each operator to an arbitrary angle $\theta$ ranging from 0 to $2\pi$. This can be done by defining two new sets of parameters $\{\gamma_{i}\}$, $\{\eta_{i}\}$ in the inner layer, transforming our algorithm into ``fixed point dueling''. The details may be discussed in our future works.

In a similar fashion as QAOA, we can generalize quantum dueling by modifying the problem and the mixer Hamiltonians to different forms, tailored to the specifics of the target problem. This would fundamentally change the dynamics of the system, unveiling new opportunities that can be explored in future research. 

Therefore, we have made an elegant connection between Grover-based algorithms using a unified, Hamiltonian setup, by adding variational parameters and augmenting the Hilbert space, making the state evolution non-trival. The reason for the existence of such representation is simply rooted in quantum mechanics, as the unitary operators could be expressed as time evolution under some specific Hermitian operators, which can be understood as Hamiltonians.



\begin{figure*}
    \centering
    
    \begin{subfigure}{0.6\textwidth}
    \centering
    \includegraphics[width = \textwidth, trim = 0 1cm 0 0]{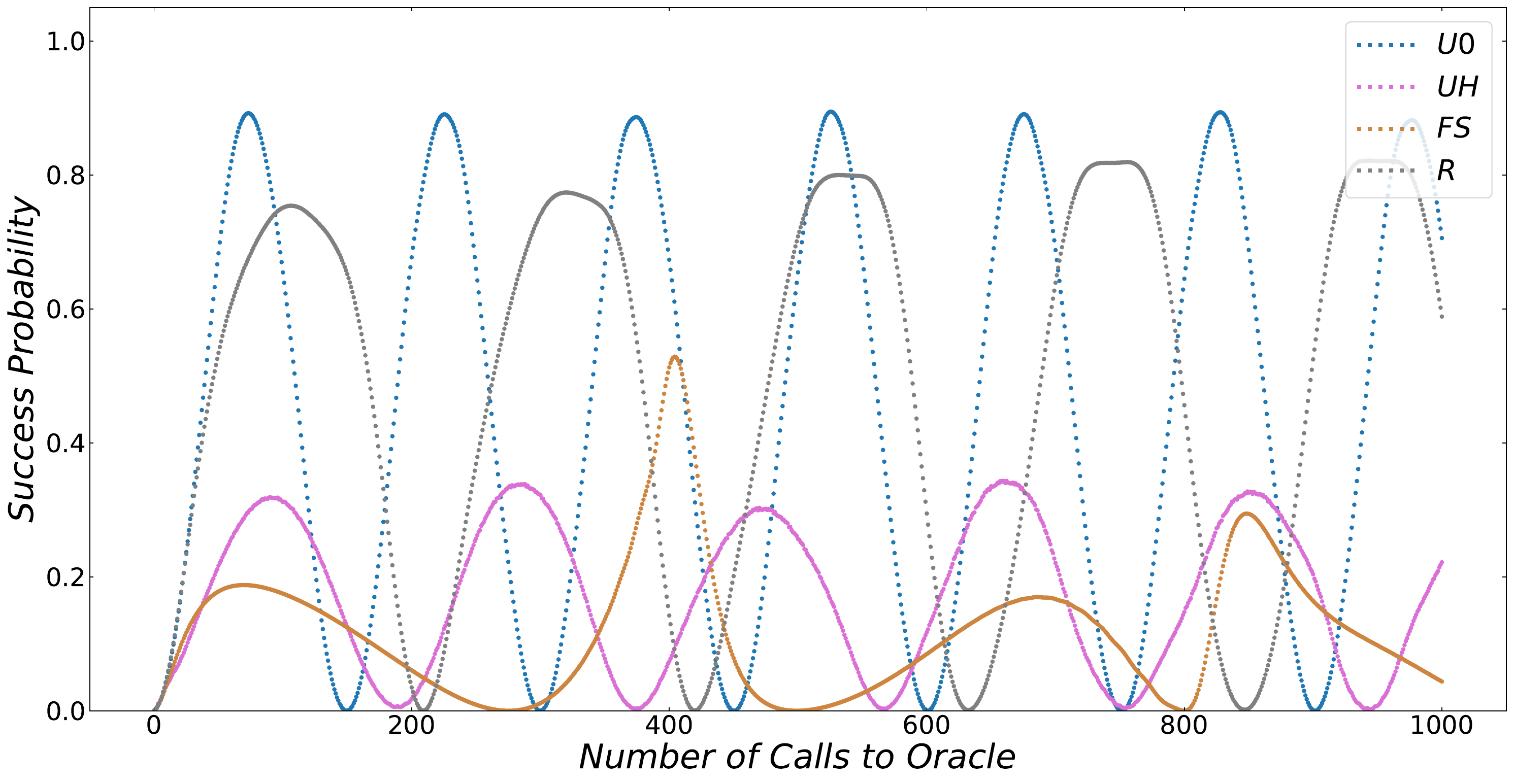}
    \caption{}
    \label{fig:parameter_scheme_1_state_evolution}
    \end{subfigure}
    \begin{subfigure}{0.3\textwidth}
    \centering
    \includegraphics[width = \textwidth]{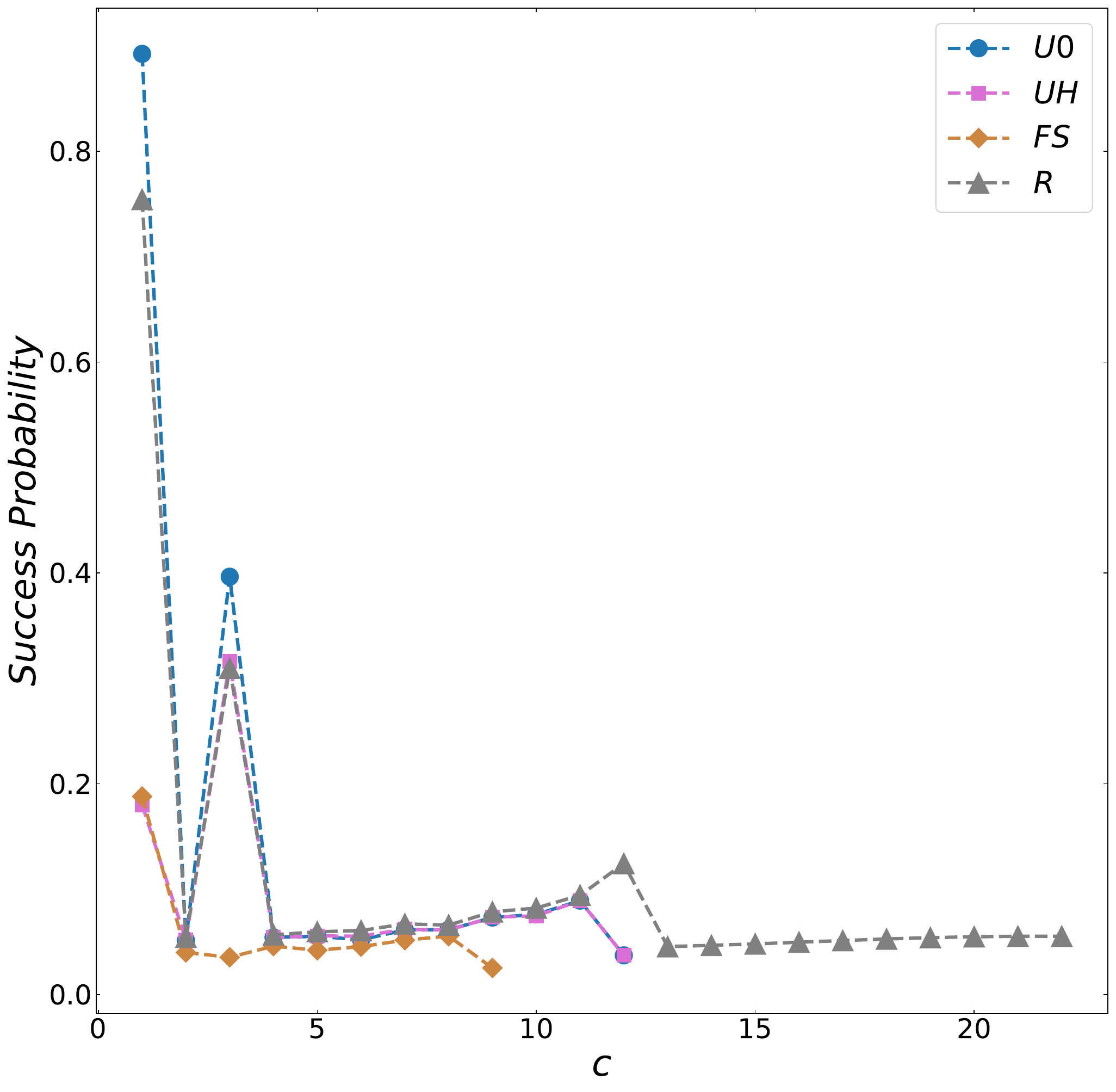}
    \caption{}
    \label{fig:parameter_scheme_1_vary_c}
    \end{subfigure}
    \\
    \begin{subfigure}{0.3\textwidth}
    \centering
    \includegraphics[width = \textwidth]{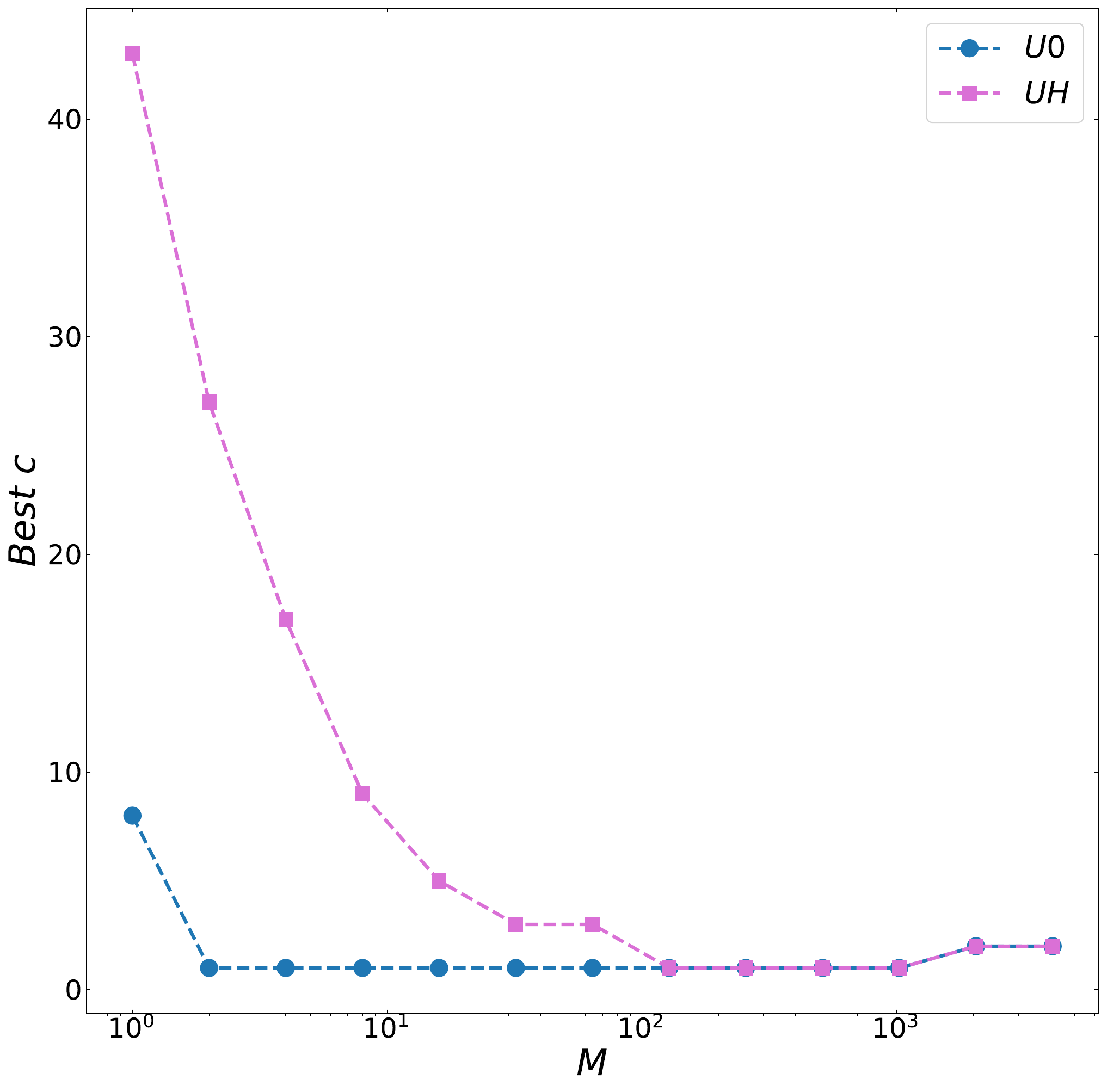}
    \caption{}
    \label{fig:parameter_scheme_1_vary_M_best_c}
    \end{subfigure}
    \begin{subfigure}{0.3\textwidth}
    \centering
    \includegraphics[width = \textwidth]{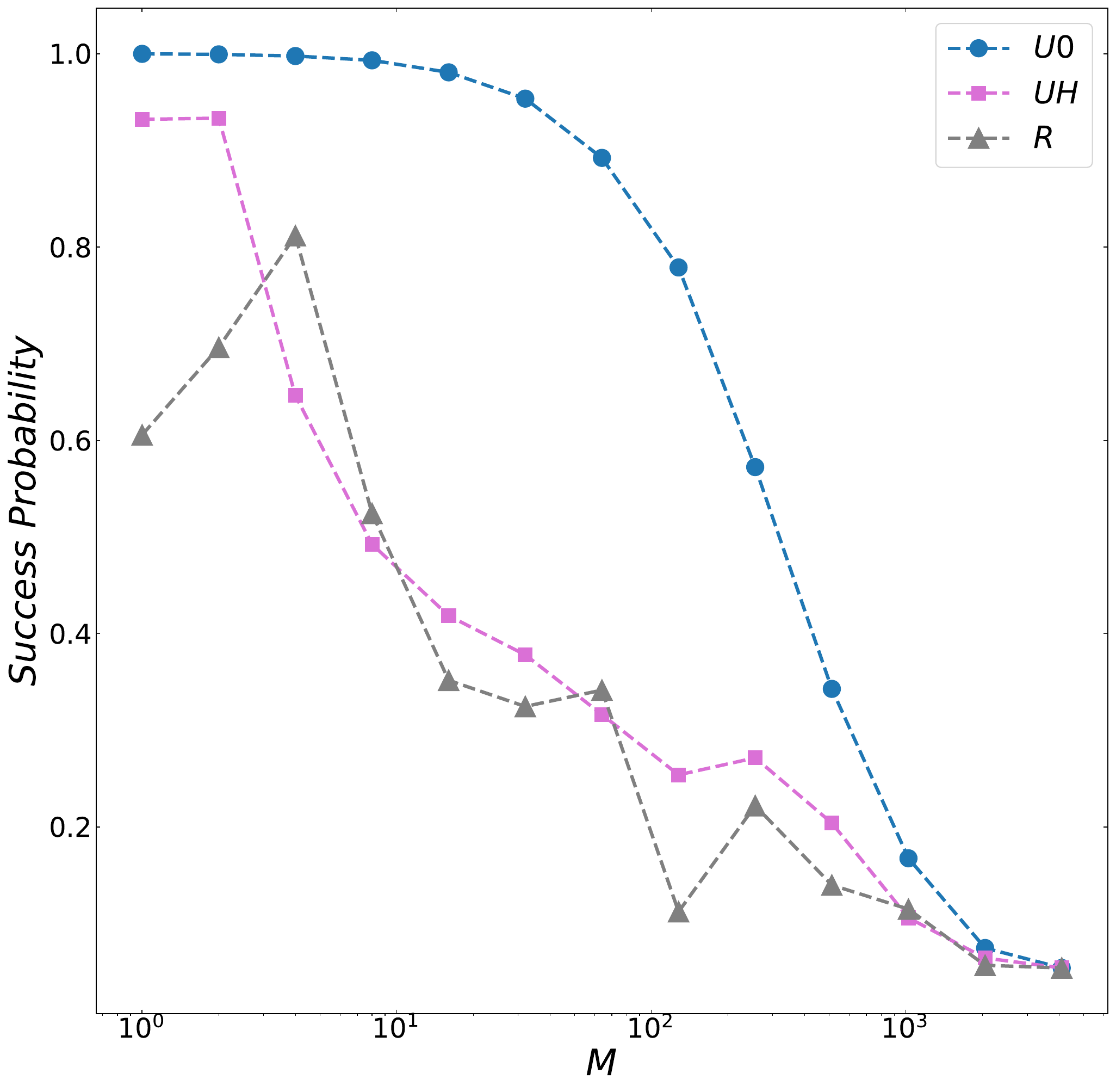}
    \caption{}
    \label{fig:parameter_scheme_1_vary_M_success_prob}
    \end{subfigure}
    \begin{subfigure}{0.3\textwidth}
    \centering
    \includegraphics[width = \textwidth]{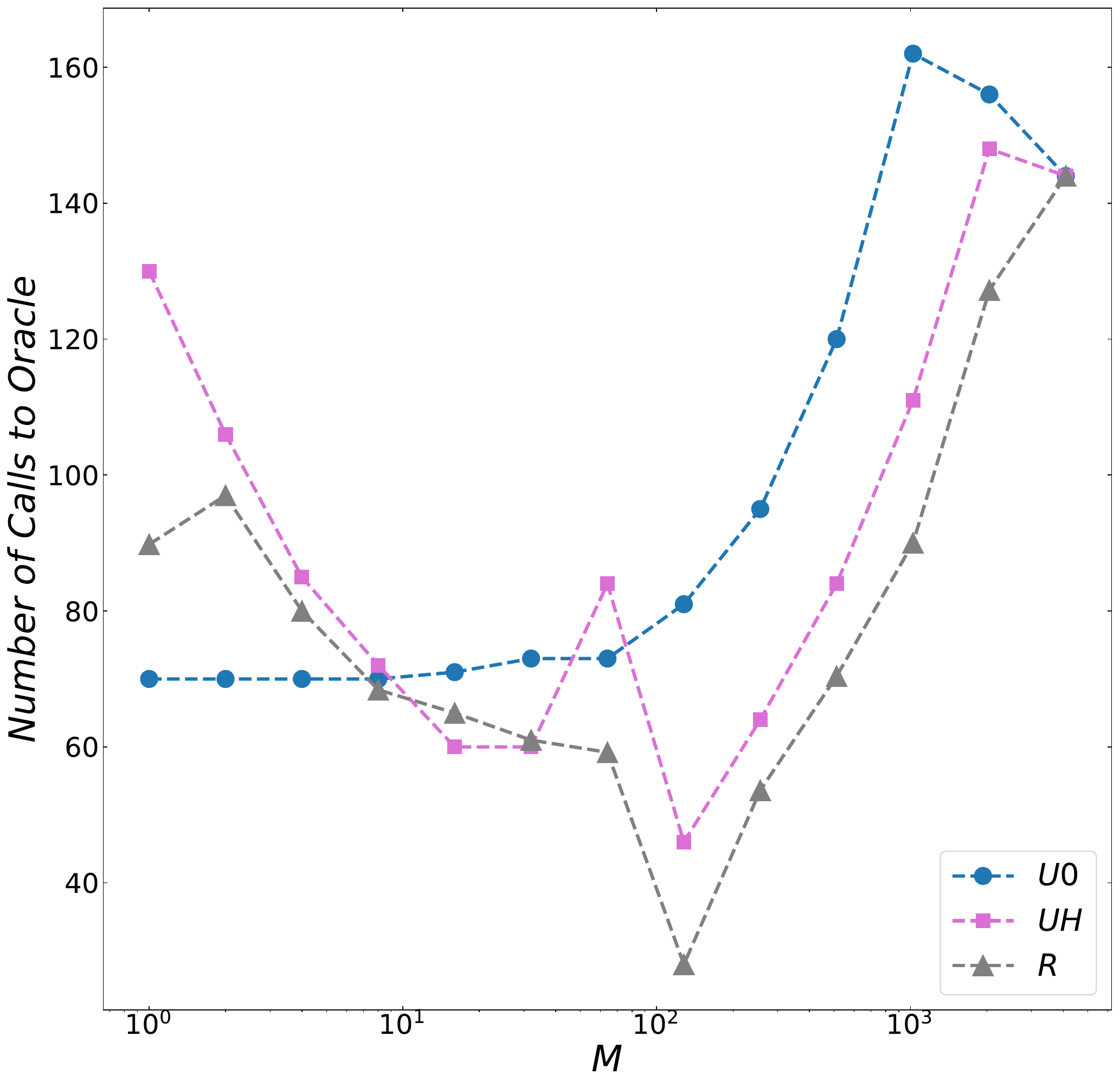}
    \caption{}
    \label{fig:parameter_scheme_1_vary_M_num_calls}
    \end{subfigure}
    \begin{subfigure}{0.3\textwidth}
    \centering
    \includegraphics[width = \textwidth]{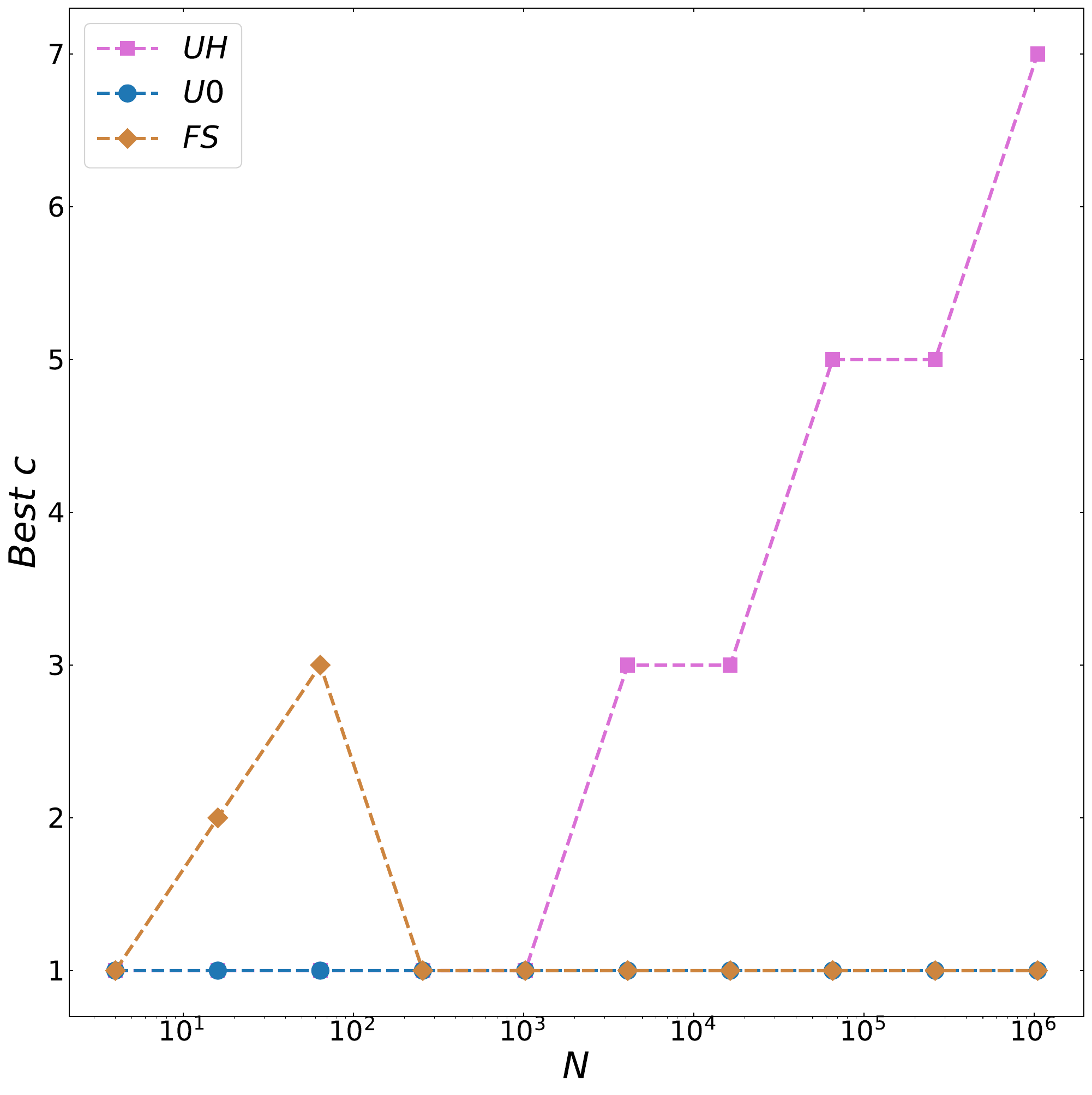}
    \caption{}
    \label{fig:parameter_scheme_1_vary_N_best_c}
    \end{subfigure}
    \begin{subfigure}{0.3\textwidth}
    \centering
    \includegraphics[width = \textwidth]{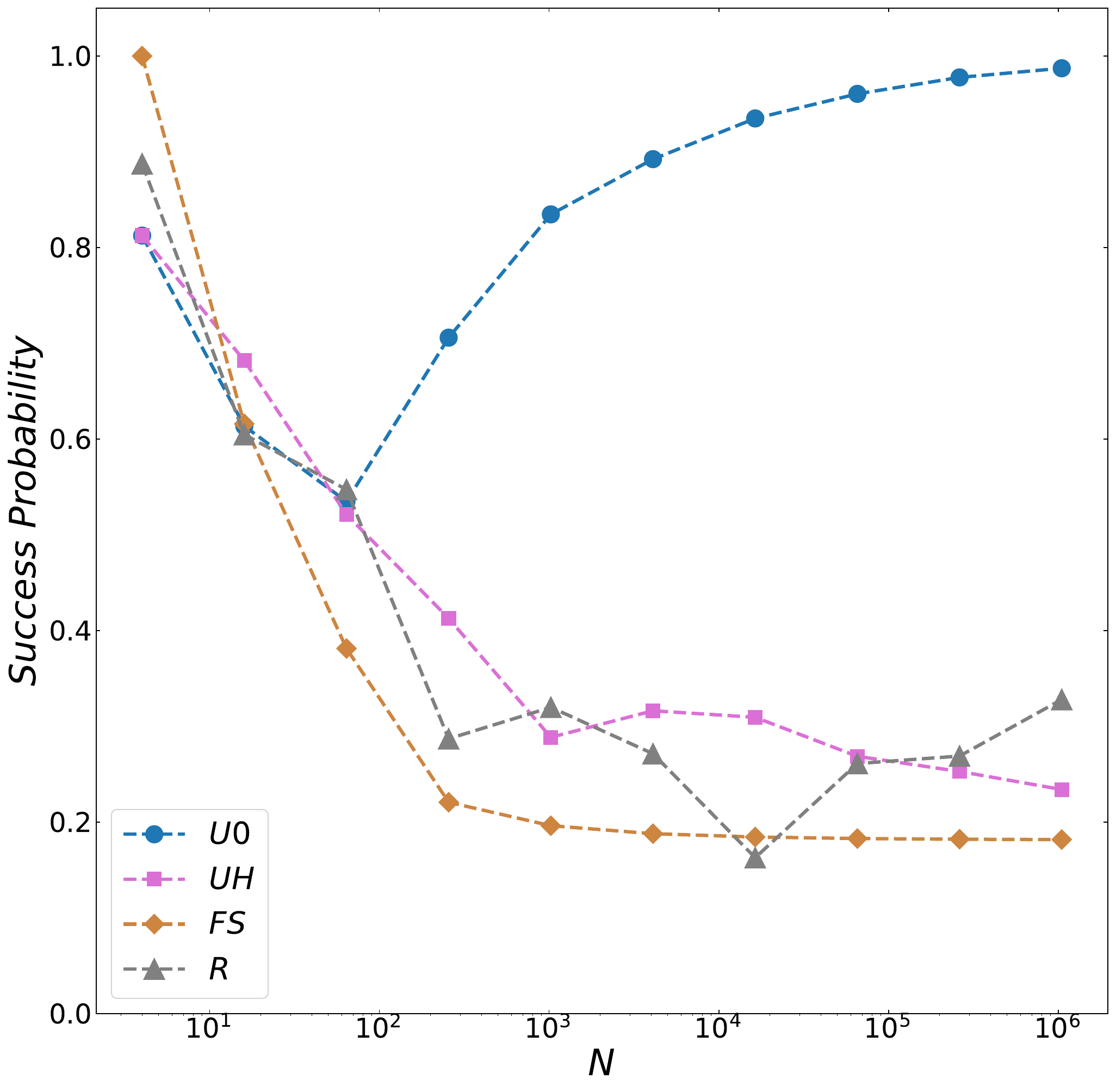}
    \caption{}
    \label{fig:parameter_scheme_1_vary_N_sucess_prob}
    \end{subfigure}
    \begin{subfigure}{0.3\textwidth}
    \centering
    \includegraphics[width = \textwidth]{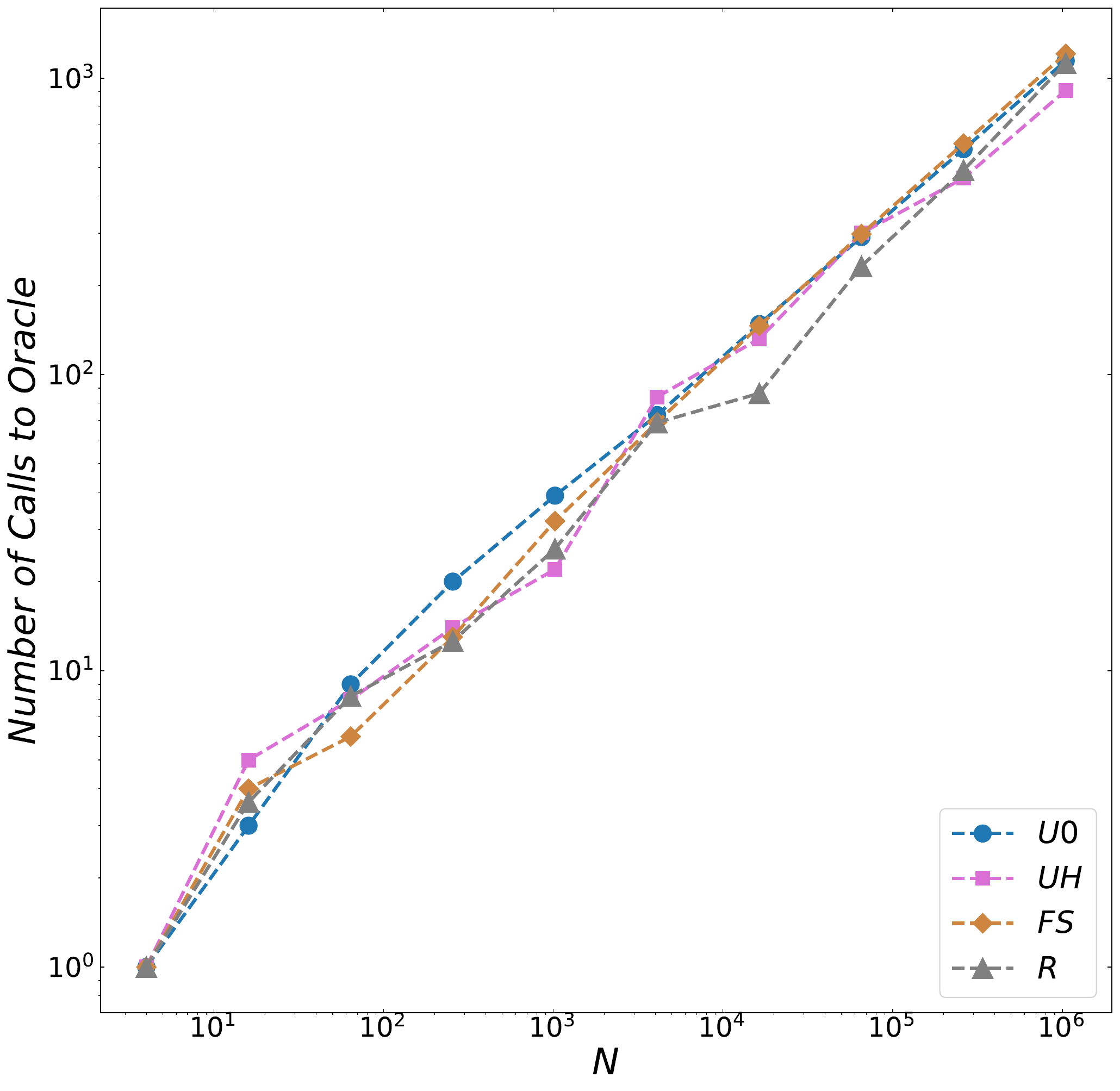}
    \caption{}
    \label{fig:parameter_scheme_1_vary_N_num_calls}
    \end{subfigure}

    \caption{Quantum dueling under parameter $\alpha_{i}=\beta_{i}=c$. The figures show the performance of the algorithm on different solution distributions. In the legends of figures, $\mathsf{U0}$ stands for $f(x) = [x - 1 \equiv 0 \pmod{\left\lceil\frac{N}{M}\right\rceil}]$, $\mathsf{UH}$ stands for $f(x) = [x \equiv \left\lfloor \frac{1}{2} \left\lceil \frac{N}{M} \right\rceil \right\rfloor \pmod{\left\lceil  \frac{N}{M} \right\rceil }  ]$, $\mathsf{FS}$ stands for $f(x) = [\exists y\in\mathbb{N}^* \; x = y^2]$ and $\mathsf{R}$ stands for random distribution for solutions. All figures show trends where $c$ takes the value that results the highest probability at the first peak, except figure (b). The random case has been treated by averaging over several runs. (a) Success probability trend under a large number of calls to Oracle. In this figure $N=4096$, $M=64$. (b) Success probabilities achieved with different $c$ and $N=4096$, $M=64$. (c)(d)(e) show trends of best parameters, success probabilities and algorithm complexity against varying M with $N=4096$ while (f)(g)(h) show them against varying N with $M=\sqrt{N}$.}
    \label{fig:parameter_scheme_1}
\end{figure*}

\section{Parameter Scheme}
\label{sec:parameter}

One of the most important features of Algorithm~\ref{alg:dueling} is the need for two additional sets of hyperparameters $\{\alpha_i\}$ and $\{\beta_i\}$ to dictate the evolution of the state vector. In an equivalent formulation, at any stage after initialization, we can either apply $\mathcal{G}_{1\leftarrow 2}$ or $\mathcal{G}_{2\leftarrow 1}$ to the current state vector. A binary string $s$ could be used to represent the application of $\mathcal{G}_{1\leftarrow 2}$, $\mathcal{G}_{2\leftarrow 1}$ operators, where $s_{i}=0$ means $\mathcal{G}_{1\leftarrow 2}$ is applied, $s_{i}=1$ means $\mathcal{G}_{2\leftarrow 1}$ is applied. In total, there are $2^{\mathsf{depth}}$ possible choices when the length of the string is $\mathsf{depth}$, and any one of the possibilities in principle might be the ones that maximize the success probability upon measurement. 


At first glance, it might seem that such a need for hyperparameters poses a major concern for the algorithm. In QAOA, for example, it has been shown that as the circuit depth $p$ increases, directly obtaining the optimal parameters becomes extremely difficult. Naturally, we may assume that similar problems are also present for quantum dueling.  As we know, different machine learning techniques have been proposed \cite{alam_2020, Xie_2023} as a subroutine of QAOA. Recently, a Hamiltonian backpropagation scheme also came to light \cite{Paster_2023}. Similarly, using machine learning to find optimal hyperparameters for quantum dueling might be one of the choices. However, despite these efforts, we still hope to determine the hyperparameters in advance, designing a relatively easy parameter scheme in most cases. In this way, we hope to resolve the concern of hyperparameters without the need of designing complicated neural networks.  

This section addresses such a concern in detail. In particular, we use ``parameter scheme''  to refer to a way to select parameters $\{\alpha_i\}$ and $\{\beta_i\}$. Quite often, these parameter schemes are based on some other (usually very few) parameters that can be selected to maximize the performance of the algorithm. In our analysis, we let $c$ be an extra parameter and considered two separate parameter schemes: $\alpha_i = \beta_i = c$ and $\alpha_i = 1, \beta_i = c$. This means that the algorithm only relies on two parameters, the extra parameter $c$ and the circuit depth $p$. It is clear that such a parameter scheme is very naive. However, even for such a simple scheme, we show that quantum dueling exhibits strong performance, reaching a quadratic speedup for all our problem setups. In other words, it seems that the parameters in quantum dueling do not pose a huge concern for the practical use of the algorithm. All we need to do is to find a reasonable parameter scheme and use variational or learning techniques to predict very few number of parameters. As will be discussed in Sec.~\ref{sec:dueling_mechanism}, a better, mathematical understanding of the mechanisms of quantum dueling will help us identify what really counts as a ``reasonable parameter scheme'', and such a question should be subject to future research. 


\subsection{Parameter scheme $\alpha_i = \beta_i = c$}
\label{sec:parameter_scheme_1}

From this point onward, we simulate a slightly modified version of quantum dueling, where the stop parameter is replaced with $T$, the total number of calls to dueling gates. Since each dueling gate calls once one of the two oracles, $T$ is also the total number of oracles called during the algorithm---an indicator of time complexity. The modified algorithm is described in Algorithm~\ref{alg:dueling_modified}. 

\begin{algorithm}
        \caption{Quantum Dueling (modified)}\label{alg:dueling_modified}
        \SetKw{False}{false} 
        \SetKw{True}{true}
        \SetKw{Break}{break}
        \KwIn{Measure function $v$ represented in quantum arithmetic; two sufficiently long integer sequences: \{$\alpha_i$\} and \{$\beta_i$\}, a stop timer $T$}
        \KwOut{An element $x$ in the search space satisfying $f(x)=1$ and $v(x)$ (approximately) minimized}
        $\ket{\psi}\leftarrow H^{\otimes2n}\ket{0}$\;
        $\mathsf{cnt} \leftarrow 0$\;
        \For{$i \leftarrow 1$ \KwTo $\infty$}{
            \For{$j \leftarrow 1$ \KwTo $\alpha_i$}{
                $\ket{\psi}\leftarrow\mathcal{G}_{1\leftarrow 2}\ket{\psi}$,  $\mathsf{cnt} \leftarrow \mathsf{cnt}+1$\;
                \lIf{$\mathsf{cnt} = T$}{break all loops}
            }
            \For{$j \leftarrow 1$ \KwTo $\beta_i$}{
                $\ket{\psi}\leftarrow\mathcal{G}_{2\leftarrow 1}\ket{\psi}$,  $\mathsf{cnt} \leftarrow \mathsf{cnt}+1$\;
                \lIf{$\mathsf{cnt} = T$}{break all loops}
            }        
        }
        Measure both registers. Let the result in the first register be $x_1$ and the result of the second register be $x_2$. Output the better result.
\end{algorithm}

With this algorithm, we set parameter $\alpha_i = \beta_i = c$, where $c$ is some positive integer. Thus, the algorithm essentially depends on two parameters, $c$ and $T$. Based on observations made in Sec.~\ref{sec:general_understanding}, we would like to select $T$ such that the state is evolved to the point where the first local maximum success probability is reached. Then, we would like to find a $c$ value that maximizes this first local maximum. 

After simulations across a broad spectrum of scenarios, we compile all outcomes in Fig.~\ref{fig:parameter_scheme_1}. Within the figure, the subplot Fig.~\ref{fig:parameter_scheme_1_state_evolution} illustrates the evolution of the success probability for some given solution distributions with the optimal $c$ selections. Consistently, the observed trend aligns with the discussions in Sec.~\ref{sec:general_understanding}, characterized by an oscillating success probability. An intriguing deviation occurs in the case of the full square number distribution; while the oscillatory pattern persists, there's a notable variation in the local maxima. Since the full square number distribution is the only distribution that differs significantly from the uniform case, this suggests that the evolution of the state vector is intricately tied to the problem setup. However, the mathematical details of such a pattern remain an open question. 

Interestingly, this finding challenges previous practices. Typically, we select $T$ to align with the first local maximum, avoiding waiting for the second one. The rationale is efficiency: the same time spent waiting for the second maximum could be used to rerun the algorithm classically and yield a better result. Yet, the full square number distribution presents a compelling counterpoint. Here, the second maximum significantly outstrips the first, suggesting that delaying for the second maximum could be a viable alternative. If this is the case, it is entirely possible that we can find solution distributions where awaiting the second local maximum is definitively more advantageous. Nonetheless, in the interest of maintaining a consistent analytical approach, our focus remains on employing the first local maximum probability in our evaluations.

In addition, Fig.~\ref{fig:parameter_scheme_1_vary_c} shows the performance of quantum dueling with varying hyperparameter $c$ for some fixed problem setups. In general, $c$ is best chosen as a small odd number unless for extremely sparse $M$. As delineated in Appendix~\ref{app:Grover_search}, for an arbitrary initial state, a single Grover search gate not only alters the mean components of solutions and non-solutions but introduces a $-1$ phase shift in the deviation of non-solution components from their mean. Within quantum dueling, such phase shift is cancelled when Grover gates are applied an even number of times consecutively. Our findings suggest that retaining the $-1$ phase through an odd $c$ introduces greater variability in the state evolution, thereby enhancing the probability of success. Furthermore, the preference for a small optimal $c$ underscores the importance of alternating between different types of gates—--a cornerstone concept in quantum dueling.

For a fixed-size search space, Figs.~\ref{fig:parameter_scheme_1_vary_M_best_c}-\ref{fig:parameter_scheme_1_vary_M_num_calls} depict quantum dueling's effectiveness for the given solution distributions. Trends in Fig.~\ref{fig:parameter_scheme_1_vary_M_best_c} implies that the optimal value of $c$ can be quantified with simple functions based on specific solution distribution without any other knowledge. Specifically, for the uniform distribution $\mathsf{U0}$, the best parameter is consistently $\alpha_i=\beta_i=1$. Results in Figs.~\ref{fig:parameter_scheme_1_vary_M_success_prob} and \ref{fig:parameter_scheme_1_vary_M_num_calls} largely echo the patterns observed in Fig.~\ref{fig:initial_simulation_all_M} where the efficacy of quantum dueling is slowly reduced with an increasing solution concentration. The most notable revelation is the observed non-linear behavior in the number of oracle calls required to attain the first maximum success probability. This unexpected and intriguing finding presents an open avenue for further research and exploration.

Lastly, Figs.~\ref{fig:parameter_scheme_1_vary_N_best_c}-\ref{fig:parameter_scheme_1_vary_N_num_calls} summarize the performance of quantum dueling for an increasing $N$. Similarly, Fig.~\ref{fig:parameter_scheme_1_vary_M_best_c} implies that the practical determination of hyper-parameter $c$ is trivial. By Fig.~\ref{fig:parameter_scheme_1_vary_N_sucess_prob}, even for large $N$, quantum dueling with parameter scheme $\alpha_i = \beta_i = c$ can boost the success probability to a significant quantity. For uniform distribution $\mathsf{U0}$, in addition to the best parameter being $\alpha_i=\beta_i=1$, the success probability is asymptotically close to 1. With the full square number distribution $\mathsf{FS}$, the success probability tends toward a constant value, another intriguing aspect requiring further investigation. For other distributions, the decrease in success probability is so slow that the asymptotic trend is unclear. Nonetheless, it is clear that the number of classical repetitions to boost success probability to a constant, even if needed, is minimal. 

\begin{table}[]
    \centering
    \begin{ruledtabular}
        \begin{tabular}{c c c}
            Distribution & Slope & Intercept\\
            \hline
            $\mathsf{U0}$  &  $0.516\pm0.014$ & $-0.037\pm0.123$\\
            $\mathsf{UH}$  &  $0.491\pm0.020$ & $0.064\pm0.182$\\
            $\mathsf{FS}$  &  $0.534\pm0.010$ & $-0.265\pm0.093$\\
            $\mathsf{R}$   &  $0.508\pm0.016$ & $-0.169\pm0.141$\\
        \end{tabular}
    \end{ruledtabular}
    \caption{Linear regression data of $\ln{(\text{Calls to Oracle})}$ against $\ln{N}$ with parameter $\alpha_{i}=\beta_{i}=c$.}
    \label{tab:parameter_scheme_1_fits}
\end{table}

In particular, the number of oracle calls $T$ required to achieve the first local maximum success probability is delineated in Fig.~\ref{fig:parameter_scheme_1_vary_N_num_calls}. For each type of solution distribution, a linear fit is conducted on the logarithms. The results are shown in Table~\ref{tab:parameter_scheme_1_fits}. Overall, it becomes evident that $T \in O(\sqrt{N})$, indicating that the \( \alpha_i = \beta_i = c \) parameter scheme can elevate success probability significantly within quadratic time for all sampled solution distributions.

\subsection{Parameter scheme $\alpha_i = 1$, $\beta_i = c$.}

\label{sec:parameter_scheme_2}

In Sec.~\ref{sec:parameter_scheme_1} we investigated the performance of quantum dueling with parameter scheme $\alpha_i = \beta_i = c$ in great detail. Our results demonstrated that even such a naive scheme enhances the success probability efficiently to a high value. However, one can argue that such strong performance is exclusive to this scheme. To address such skepticism, we examined an alternative scheme where $\alpha_i = 1$ but $\beta_i$ is set to some parameter $c \in \mathbb{N}^+$. Intuitively, this represents an asymmetric design where the first register is primarily utilized as the control to optimize the second, interspersed with occasional updates.  

\begin{table}[]
    \centering
    \begin{ruledtabular}
        \begin{tabular}{c c c}
            Distribution & Slope & Intercept\\
            \hline
            $\mathsf{U0}$  &  $0.487\pm0.003$ & $0.275\pm0.028$\\
            $\mathsf{UH}$  &  $0.455\pm0.021$ & $0.310\pm0.212$\\
            $\mathsf{FS}$  &  $0.539\pm0.010$ & $-0.312\pm0.102$\\
            $\mathsf{R}$   &  $0.501\pm0.015$ & $0.028\pm0.146$\\
        \end{tabular}
    \end{ruledtabular}
    \caption{Linear regression data of $\ln{(\text{Calls to Oracle})}$ against $\ln{N}$ with parameter $\alpha_{i}=1$, $\beta_{i}=c$.}
    \label{tab:parameter_scheme_2_fits}
\end{table}

The cumulative findings are compiled in Fig.~\ref{fig:parameter_scheme_2} and Table~\ref{tab:parameter_scheme_2_fits}, employing the same methodology as in Sec.~\ref{sec:parameter_scheme_1}. For brevity, we selected plots of state evolution, the impact of varying parameter $c$, and performance changes with different $N$. These result affirm that the trends discussed in Sec.~\ref{sec:parameter_scheme_1} persist for this parameter scheme, with quantum dueling achieving notable success probabilities within $O(\sqrt{N})$ time. Consequently, we can conclude that the hyperparameter required in quantum dueling does not pose a significant challenge for practical implication. Instead, it brings a beneficial versatility to the algorithm, further enhancing its efficacy.

\begin{figure*}
    \centering
    
    \begin{subfigure}{0.6\textwidth}
    \centering
    \includegraphics[width = \textwidth, , trim = 0 1cm 0 0]{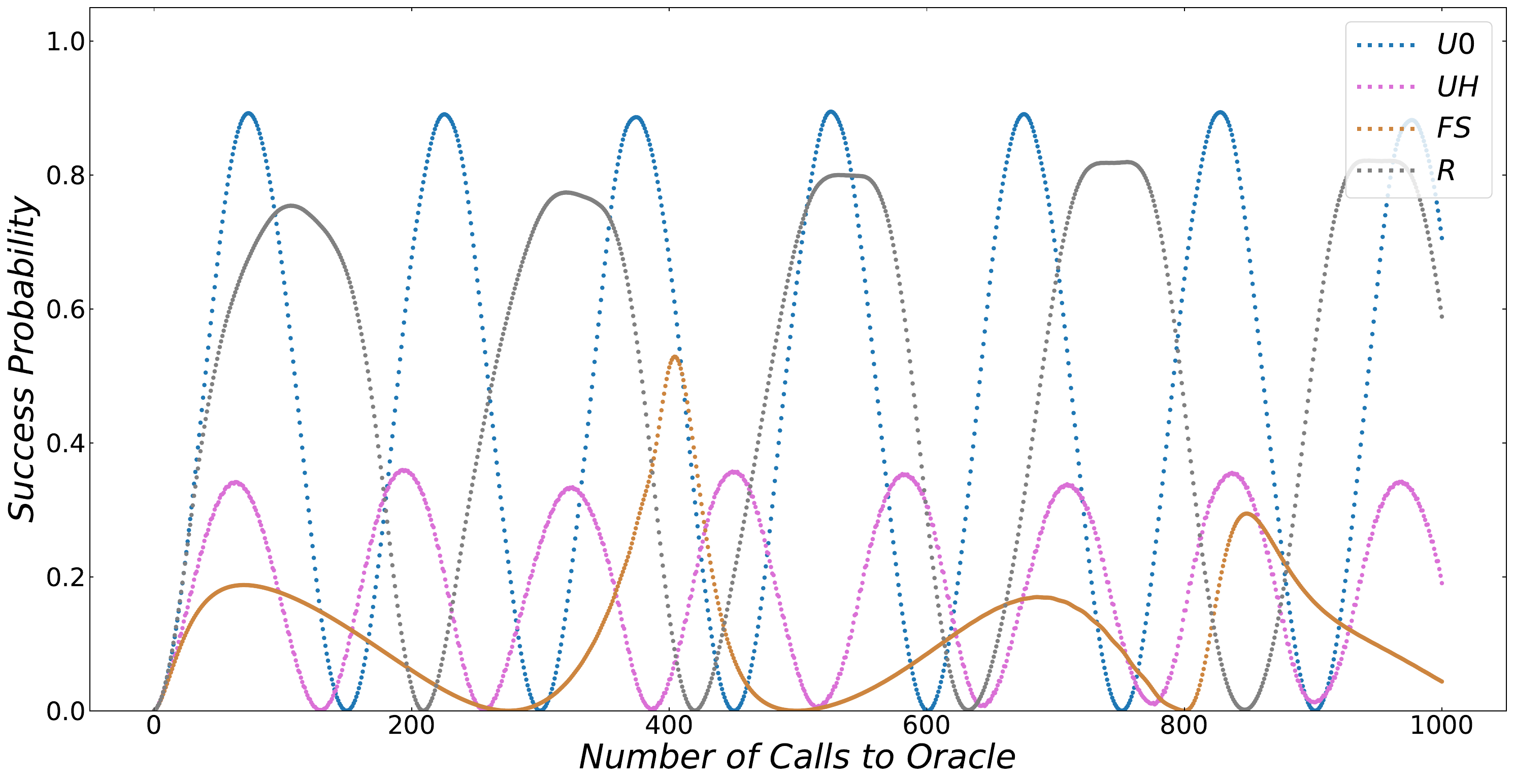}
    \caption{}
    \label{fig:parameter_scheme_2_state_evolution}
    \end{subfigure}
    \begin{subfigure}{0.3\textwidth}
    \centering
    \includegraphics[width = \textwidth]{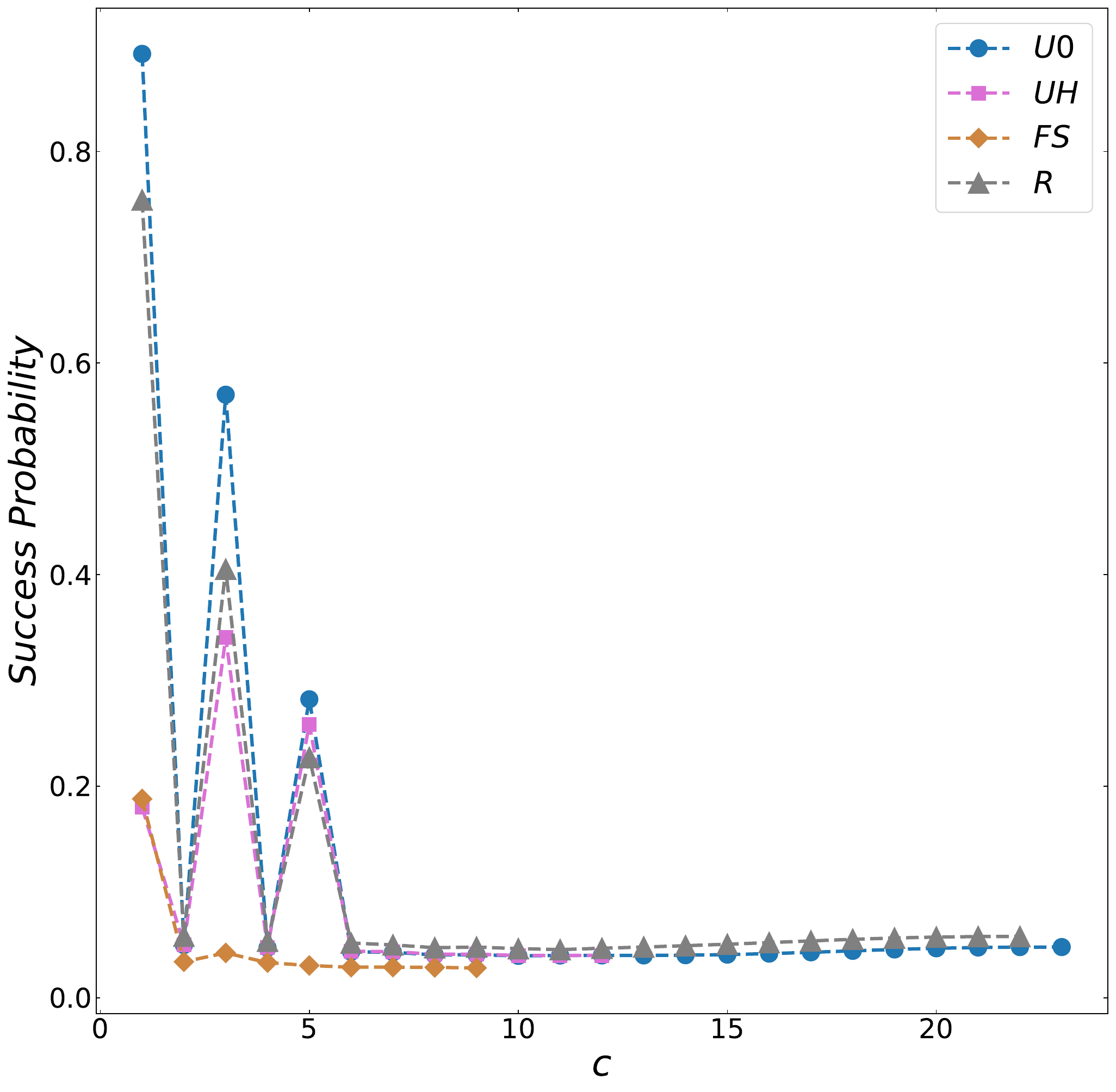}
    \caption{}
    \label{fig:parameter_scheme_2_vary_c}
    \end{subfigure}
    \\
    \begin{subfigure}{0.3\textwidth}
    \centering
    \includegraphics[width = \textwidth]{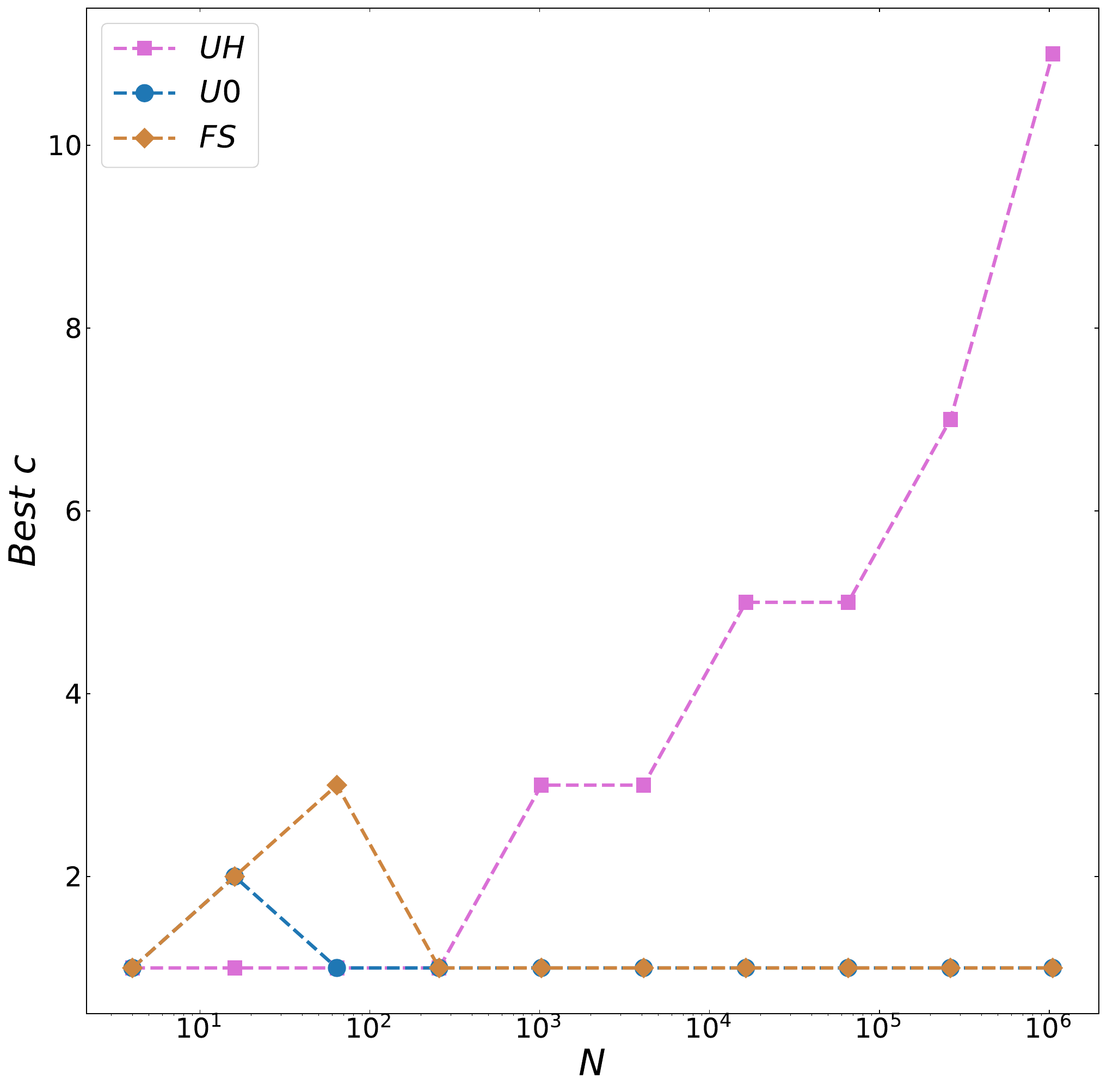}
    \caption{}
    \label{fig:parameter_scheme_2_vary_N_best_c}
    \end{subfigure}
    \begin{subfigure}{0.3\textwidth}
    \centering
    \includegraphics[width = \textwidth]{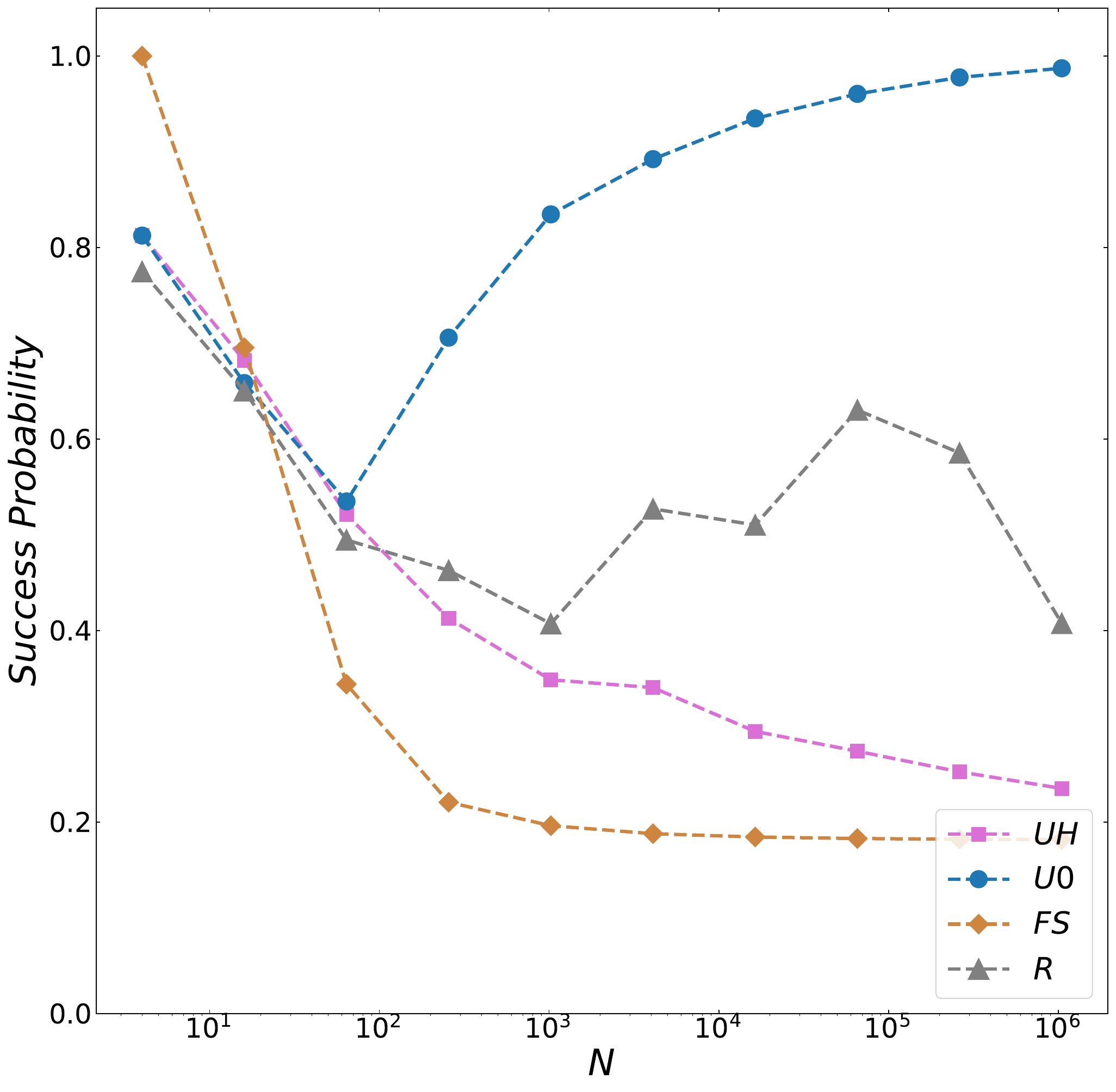}
    \caption{}
    \label{fig:parameter_scheme_2_vary_N_sucess_prob}
    \end{subfigure}
    \begin{subfigure}{0.3\textwidth}
    \centering
    \includegraphics[width = \textwidth]{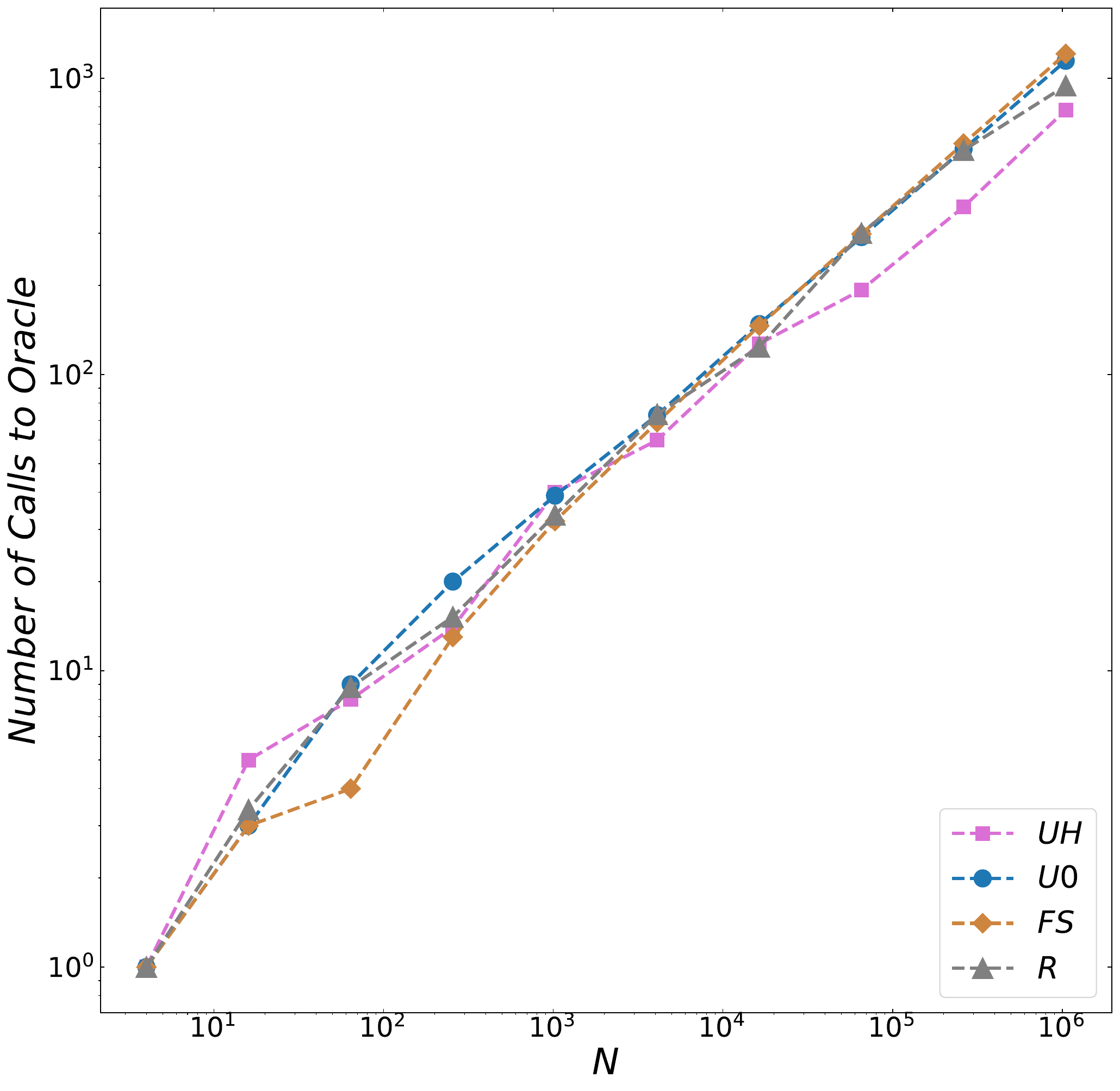}
    \caption{}
    \label{fig:parameter_scheme_2_vary_N_num_calls}
    \end{subfigure}

    \caption{Quantum dueling under parameter $\alpha_{i}=1, \beta_{i}=c$. All presented plots have the same configurations as in Fig.~\ref{fig:parameter_scheme_1}.}
    \label{fig:parameter_scheme_2}
\end{figure*}

\subsection{Quantum Dueling under Locally Optimal Parameters}
\label{sec:parameter_heuristics}

\begin{algorithm}
    \caption{Classical Heuristic Algorithm to Find Locally-optimal Parameters}
    \label{alg:heuristic}
    \SetKw{Continue}{continue}
    \SetKw{Clear}{clear}
    \SetKw{KwOr}{or}
    \KwIn{functions $f$ and $v$, a termination condition and a pruning condition, a maximum search depth $\mathsf{depth}$}
    Initialize $|\psi\rangle$ in classical simulation\;
    \Repeat{\text{termination condition is met}}{
        $|\psi'\rangle\leftarrow|\psi\rangle$ \;
        \Clear $|\psi_\text{best}\rangle$\;
       \For{binary string $s$ of length $\mathsf{depth}$ that passes the pruning condition}{
            \For{$i\leftarrow 1$ \KwTo $\mathsf{depth}$}{
                \lIf{$s_i = 0$}{$|\psi'\rangle\leftarrow\mathcal{G}_{1\leftarrow 2}|\psi'\rangle$}
                \lElse{Apply $|\psi'\rangle\leftarrow\mathcal{G}_{2\leftarrow 1}|\psi'\rangle$}
            }
            \uIf{$|\psi_\text{best}\rangle$ is uninitialized \KwOr success probability of $|\psi'\rangle$ is better than that of $|\psi_\text{best}\rangle$}{
            $|\psi_\text{best}\rangle\leftarrow|\psi'\rangle$\;
            }
       }
       $\ket{\psi}\leftarrow\left|\psi_\text{best}\right\rangle$\;
    }
\end{algorithm}

In Sec~\ref{sec:parameter_scheme_1} and Sec~\ref{sec:parameter_scheme_2}, we found that quantum dueling exhibits a strong performance even for naive parameter schemes. However, it is worth noting that our analysis focuses on non-extreme solution distributions where, when sorted by target function value, the solutions are dispersed into non-solutions in some manner. For cases where all solutions are cramped together without separation, our analysis in Sec.~\ref{sec:general_understanding} shows that quantum dueling, at least for $\alpha_i = \beta_i = 1$, effectively loses efficacy. 

A natural question to ask is for these extreme cases, is it possible that there still might be ways to achieve a significant boost of success probability with a proper selection of hyper-parameters? In other words, we would like to simulate how quantum dueling would perform for a near-optimal parameter selection without the restriction of specific parameter schemes. Then, we would test our simulation against these extreme solutions to explore the limit of quantum dueling. 

Since our research is primarily aided by classical simulation, our purpose is simply to demonstrate the limit of quantum dueling without any practical consideration. We are free to use the benefit of classical simulation, i.e., the ability to store and copy quantum states with knowledge of the entire state vector. Using the binary string $s$ formulation as discussed at the start of Sec~\ref{sec:parameter}, one of the approaches to finding a near-optimal parameter is a greedy strategy that iteratively finds the best sequence of gates to apply in the next $\mathsf{depth}$ gates. If needed, such a search can be pruned to situations where the switch from one gate to another is sparse in the $s$ string, reflecting the $\{\alpha_i\}, \{\beta_i\}$ formulation of hyperparameter. \footnote{In particular, with parameters defined in Algorithm~\ref{alg:heuristic}, let $\mathsf{change}$ be the number of locations $i$ in $s$ where $s_i \neq s_{i-1}$, we can set the pruning condition to be $\mathsf{change} < \mathsf{change\_limit}$, where $\mathsf{change\_limit}$ is some parameter. In principle, this would reduce the computational load and allow us to estimate quantum dueling's performance for an increasing $N$. An example diagram can be found in \cite{codes_on_github}.} The exact algorithm is sketched in Algorithm~\ref{alg:heuristic}.

\begin{figure*}
    \centering 
    \begin{subfigure}{0.45\textwidth} 
        \centering
        \includegraphics[width=\linewidth]{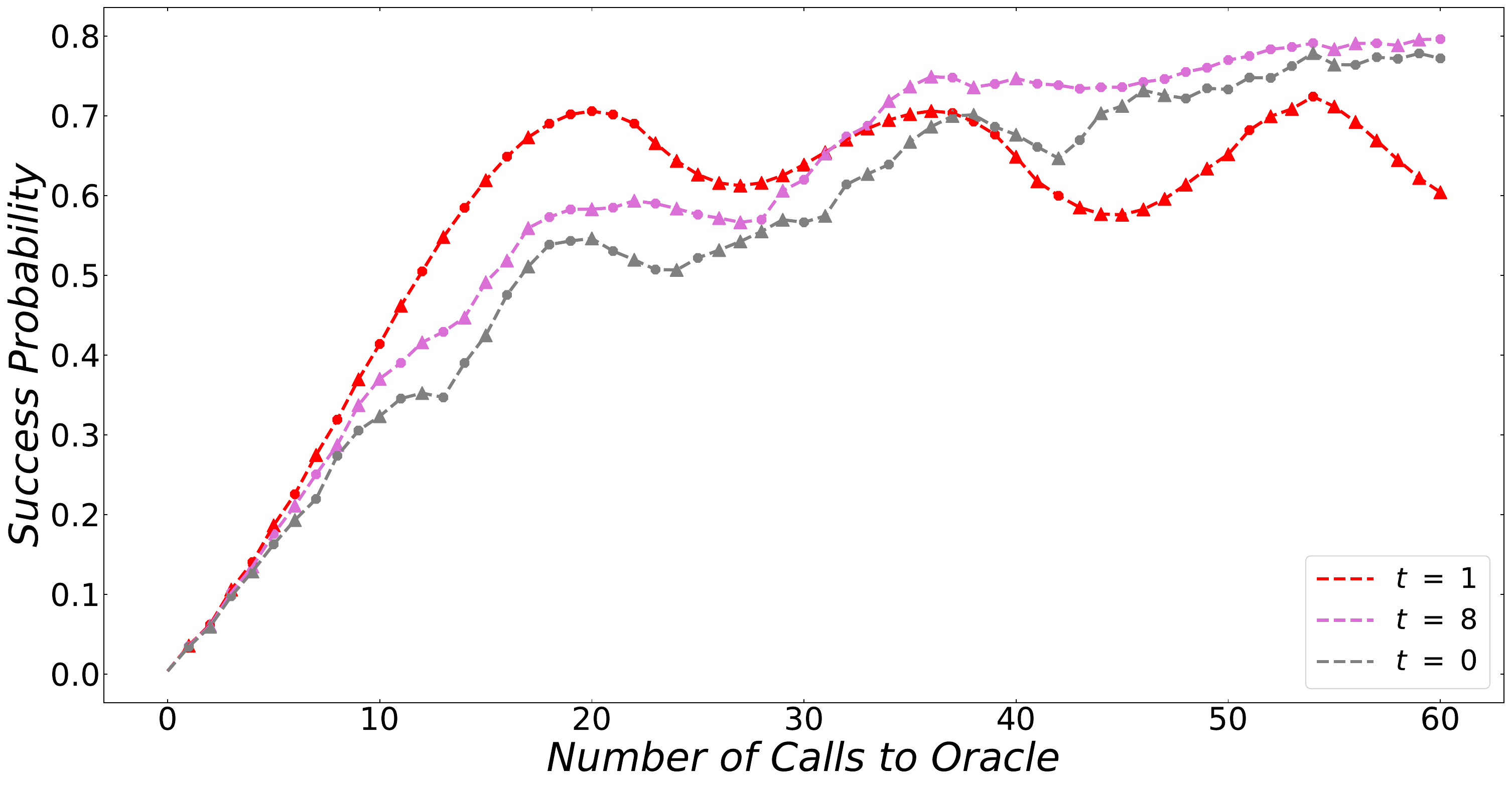}
        \caption{}
        \label{fig:heuristics_uniform_t}
    \end{subfigure}
    \begin{subfigure}{0.45\textwidth}
        \centering
        \includegraphics[width=\linewidth]{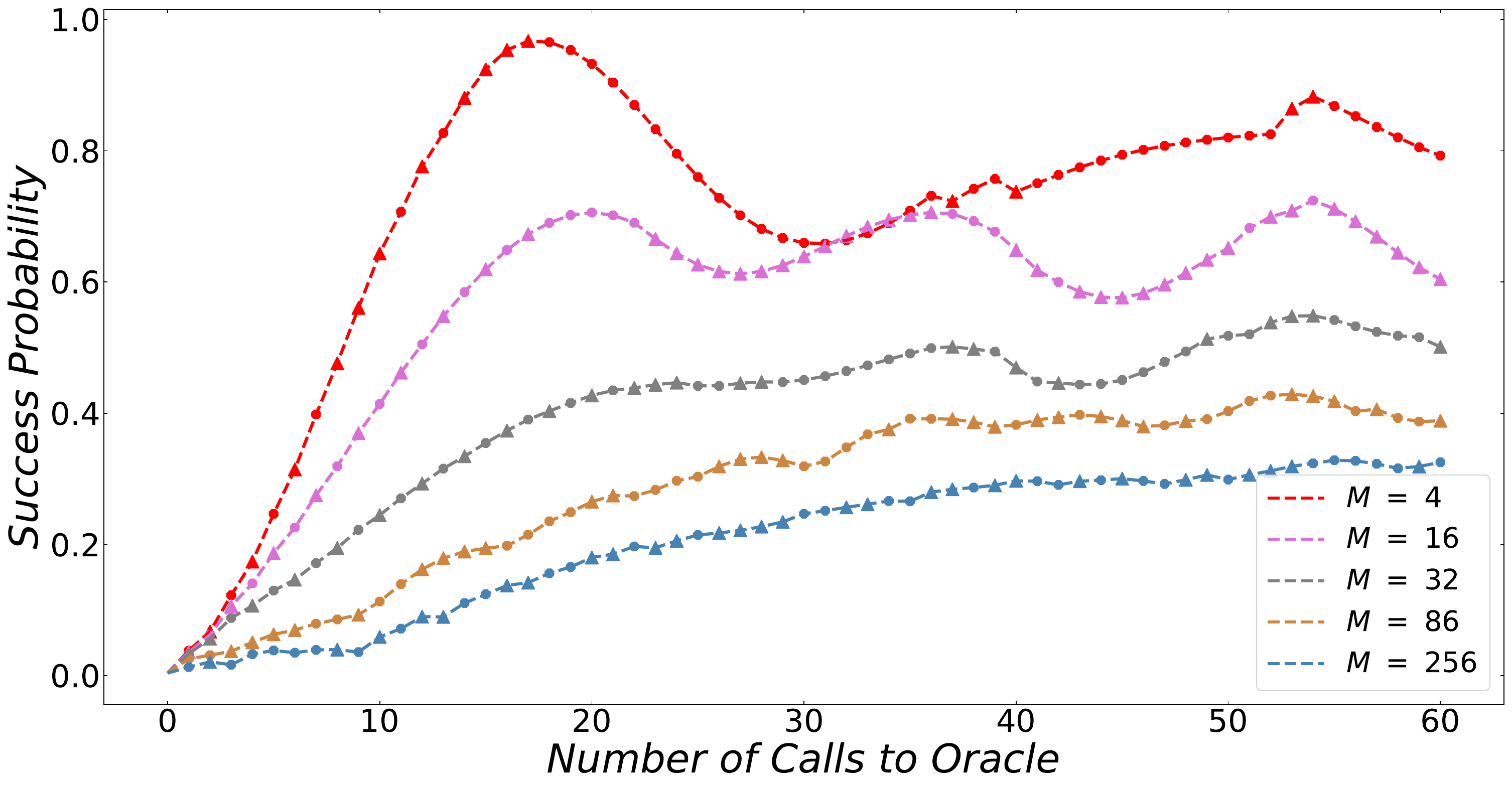}
        \caption{}
        \label{fig:heuristics_uniform_M}
    \end{subfigure}
    \\
    \begin{subfigure}{0.93\textwidth}
        \centering
        \includegraphics[width=0.98\linewidth]{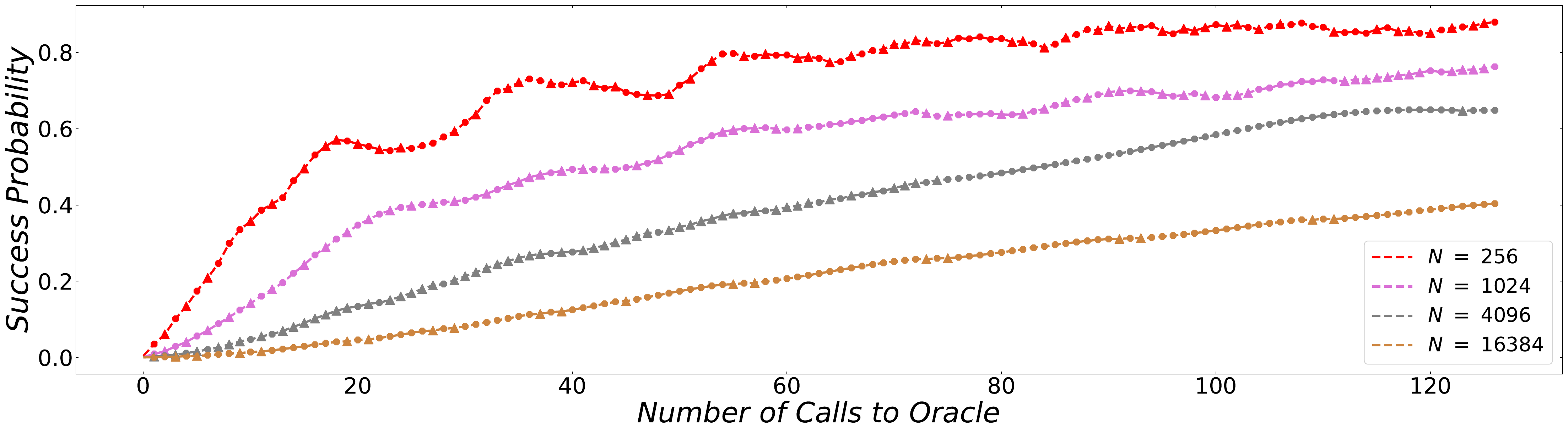}
        \caption{}
        \label{fig:heuristics_uniform_N}
    \end{subfigure}

    \caption{Quantum dueling under near-optimal parameters. For uniform solution distribution, measure function $v(x) = x$, and no pruning of the search tree, we ran the classical heuristic algorithm under $\mathsf{depth}=18$ for three separate sets of data and recorded the evolution of success probability with respect to each application of dueling gate. Data points that are achieved after applying $\mathcal{G}_{1\leftarrow 2}$ are denoted using \CIRCLE, and those achieved after applying $\mathcal{G}_{2\leftarrow 1}$ {\large\UParrow}. (a) $N=256$, $M = 16$, $f(x) = [x \equiv t\pmod{16}]$, where $t$ is 1,8, and 0, respectively. (b) $N=256$, $f(x) = [x - 1 \equiv 0 \pmod{\left\lceil\frac{N}{M}\right\rceil}]$, where $M$ is 4, 16, 32, 86, and 256 respectively. (c) $M = \left\lceil  \sqrt{N} \right\rceil$, $f(x) = [x \equiv \left\lfloor \frac{1}{2} \left\lceil \frac{N}{M} \right\rceil \right\rfloor \pmod{\left\lceil  \frac{N}{M} \right\rceil }  ]$, where $N$ takes the value 256, 1024, 4096, 16384, respectively.}
    \label{fig:heuristics_uniform}
\end{figure*}

Before we tackle these intractable cases, first, we run Algorithm~\ref{alg:heuristic} against several different forms of uniform solutions distribution and summarize the results in Fig.~\ref{fig:heuristics_uniform}. It is clear that under such an optimal scheme, quantum dueling can significantly boost success probability towards a high probability, usually to a better extent than the parameter schemes discussed in Secs.~\ref{sec:parameter_scheme_1} and \ref{sec:parameter_scheme_2}. Note that in Fig.~\ref{fig:heuristics_uniform_t}, for uniform solutions distribution where the element with smallest measure function value happens to be the solution, the heuristics algorithm is able to find the best parameter to be $\alpha_i = \beta_i = 1$ before the maximum probability is reached. \footnote{It is worth noting that for a sufficiently large $N$, Algorithm~\ref{alg:heuristic} is significantly bounded by the $\mathsf{depth}$ parameter. As a result, the algorithm can no longer find $\alpha_i = \beta_i = 1$. See \cite{codes_on_github}.}  Moreover, as shown in Fig.~\ref{fig:heuristics_uniform_N}, for an increasing $N$, we observed longer continuous sequences of $G_1$ and $G_2$ being applied. This justified the $\{\alpha_i\}, \{\beta_i\}$ formulation of hyper-parameter as opposed to the binary string $s$ formulation. Lastly, as seen in Fig.~\ref{fig:heuristics_uniform_M}, the use of better parameters found by heuritics algorithm alleviates the collapses of efficacy at $M = \frac{N}{4}$ as discovered in Sec.~\ref{sec:general_understanding}.

With this in mind, we then test the aforementioned perplexing solution distributions, along with some other non-uniform distributions, and present the results in Fig.~\ref{fig:heuristics_intractable}. It is clear that with proper parameters, quantum dueling manages to achieve a significant boost of success probability regardless of solution distribution, even for these onerous cases. Note that the distribution where all non-solutions have function values higher than solutions is of particular interest to us. Since our oracle is assumed to distinguish between solutions and non-solutions, we can arbitrarily shift the $v$ distributions of all solutions by alternating the construction of the oracle using quantum arithmetic. In principle, we can thus morph all forms of solution distributions to such a distribution. Our simulation suggests that there exists a parameter scheme for this problem that is sufficiently good. If future research can make such a scheme explicit, we will obtain a universal method for all combinatorial optimization and can potentially offer a lower bound for the performance of dueling.

For the worst-case scenario, where all solutions have measure function values higher than all non-solutions and therefore cannot be boosted, quantum dueling can still boost the success probability towards a high constant, albeit at a much slower rate. This result has defied our expectations, and fullying understanding this effect remains an open question.

Lastly, in some cases, a decrease in success probability is observed in order to achieve the maximum, increased success probability after $\mathsf{depth}$ moves. This seems to suggest that aggregations of locally optimal decisions do not necessarily result in the globally optimal decision, because in such cases the probability may experience a longer decrease first. Therefore, the parameter suggested by Algorithm~\ref{alg:heuristic} is not necessarily the optimal parameter and should be taken with a grain of salt.

Overall, this section discussed the performance of quantum dueling with different parameter setups in detail. While quantum dueling in its definition requires a set of hyperparameters, Secs.~\ref{sec:parameter_scheme_1} and \ref{sec:parameter_scheme_2} suggest even a naive scheme can help to find parameters that enable dueling to boost success probability to a sufficient value. In these examples, the degree of freedom in the hyperparameter is essentially reduced to 2, which, in practice, can be determined via manageable methods. Furthermore, in Sec.~\ref{sec:parameter_heuristics} we investigated a near-optimal parameter scheme that could possibly be approached with the use of more complex parameter-determining methods such as machine learning. In this case, quantum dueling will exhibit an even stronger performance for almost all, if not all solution distributions, becoming a powerful tool for generic combinatorial optimization.


\begin{figure*}
    \centering
    \begin{minipage}{\textwidth} 
        \includegraphics[width=\textwidth]{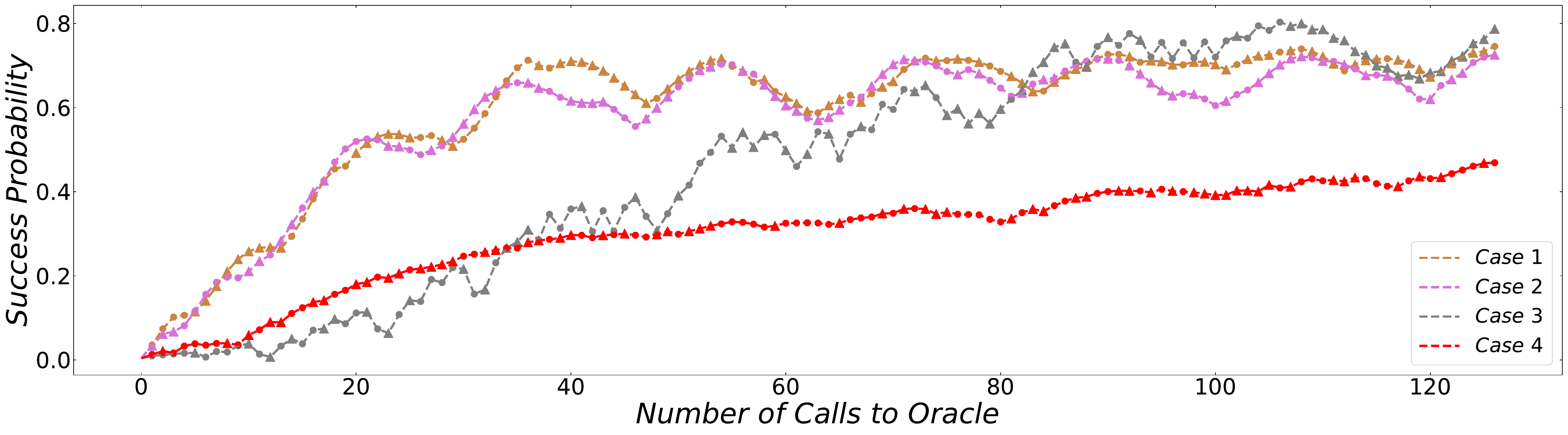}
        \caption{Quantum dueling for different solution distributions. We ran heuristics under $\mathsf{depth} = 18$ to collect relevant data. Data points that are achieved after applying $\mathcal{G}_{1\leftarrow 2}$ are denoted using \CIRCLE, and those achieved after applying $\mathcal{G}_{2\leftarrow 1}$ {\large\UParrow}. Solution distribution in case 1 is $f(x) = [\exists y\in\mathbb{N}^* \; x = y^2]$; solution distribution in case 2 is $f(x) = [1 \leqslant x \leqslant M]$; solution distribution in case 3 is $f(x) = [N-M+1\leqslant x \leqslant N]$; solution distribution in case 4 is $f(x) = 1$.}
        \label{fig:heuristics_intractable}
    \end{minipage}
\end{figure*}

\section{Application}
\label{sec:application}

\subsection{Quantum Dueling without Solution Labeling}
\label{sec:opnl}

The solution distributions discussed in the earlier sections of this paper are often idealized and not reflective of real-world scenarios. In this section, we shift our focus to more practical distributions. Quite often, we encounter cases when all elements in the search space are labeled as solutions or, equivalently, when $M = N$. In this case, the combinatorial optimization problem defined in Sec.~\ref{sec: algorithm} could be reduced to simply finding an element $x$ that minimizes the measure function, namely: 
\begin{equation}
    \label{eq:new_prob_def}
    v\left( x \right) = \min \Set{v\left( y \right)}{y \in S}
\end{equation}

As an example, we consider the specific case when all elements share distinct measure function values. In other words, there is no degeneracy of the measure function $v(x)$.

As previously outlined in Sec.~\ref{sec:general_understanding}, when confronted with the challenge of quantum dueling's diminished effectiveness in scenarios where $M = N$, we proposed two potential solutions. One is to manually add ancilla qubits, sparsing out the solution distribution in the new search space, while the other is to design better parameter schemes. The latter approach has already been discussed in Sec.~\ref{sec:parameter}. In this section, we focus on the former approach.

Formally speaking, we could add $\Delta n$ ancilla qubits to the system, extending the original search space $S$ into $S'$. In this manner, we can define an extra parameter $\eta = \frac{M}{N}$, indicating the ratio of the number of solutions to the total number of elements in the search space. As the number of ancilla qubits $\Delta n$ increases, the ratio $\eta = \frac{M}{N} = \frac{1}{2^{\Delta n}}$ decreases monotonically, so that we are free to choose $\eta$ as the control parameter. Originally, the size of the search space $S$ is $N$, with $N = M$. In this section, to avoid confusion, we use $M$ to represent the size of the original search space, also corresponding to the number of solutions in the augmented search space $S'$. This is due to the fact that we could manually set all the new elements $x'$ added to the original search space to be non-solutions, which means $f(x') = 0$, $\forall \, x' \in S' \setminus S$, where $\setminus$ is the notation for set difference. 


In principle, we can set the value of $v(x')$ arbitrarily. The goal is to achieve a solution distribution such that quantum dueling shows robust performance. For all $x$ and its corresponding $x'$, we can set $v(x')$ to be any value in $\left( v(x), v(x) + \Delta v(x) \right)$. In that way, by results in Sec.~\ref{sec:cluster_representation}, we will end up with a solution distribution that effectively equals $\mathsf{U0}$.

In our simulation, we execute Algorithm~\ref{alg:dueling_iterative} in the augmented search space $S'$ with ancilla qubits added with parameter scheme $\alpha_i = \beta_i = c$. \footnote{The algorithm can be further optimized utilizing the known current best answer, in the similar fashion as Grover Adaptive Search in Appendix~\ref{app:GAS}.}$^,$\footnote{since our solution distribution is effectively $\mathsf{U0}$ this means $\alpha_i = \beta_i = 1$}  We set the number of oracle calls per round $T$ to be the number of calls when the first maximum probability $p_{max}$ is achieved in Algorithm~\ref{alg:dueling_modified}.  Our goal is to achieve some constant success probability $P_{\text{bar}}$. Therefore, we set $a = \left\lfloor \log_{1 - p_{max}} \left( 1 - P_{\text{bar}} \right)  \right\rfloor$ and $T_{\text{last}}$ to correspond to the time when the combined success probably reaches the threshold $P_{\text{bar}}$. We then perform Algorithm~\ref{alg:dueling_iterative} with the $\eta$ value that minimizes the total number of calls to oracle.

As we add more ancilla qubits, making $\eta$ smaller, we expect an increase of the first maximum success probability $p_{\text{max}}$, a corresponding increase of the number of calls to the oracle per round $T$, as well as the total number of calls to the oracle $\mathsf{num}$. In particular, we are interested in the growing trend of the total number of calls to the oracle $\mathsf{num}$ with respect to the size of the original search space $M$ under the optimized $\eta$ for each $M$, because this trend might give us a hint into the computational complexity of quantum dueling under the naive parameter scheme $\alpha_{i} = \beta_{i} = c$ when all $M=N$ originally.   

\begin{algorithm}
        \caption{Iterative Quantum Dueling}
        \label{alg:dueling_iterative}
        \SetKw{False}{false} 
        \SetKw{True}{true}
        \SetKw{Break}{break}
        \KwIn{Functions $f$ and $v$; Success probability threshold $P_{\text{bar}}$; Total number of complete rounds $a$; The stop timer $T$ for each complete round; The stop timer $P_{\text{last}}$ for the last round; Two sufficiently long integer sequences: $\{\alpha_{i}\}$ and $\{\beta_{i}\}$. }
        \KwOut{An element $x$ in the search space $S$ satisfying $f(x) = 1$ and $v(x)$ (approximately) minimized; the total number of oracle calls $\mathsf{num}$}
        $\mathsf{num} \leftarrow 0$ \;
        \For{$i \leftarrow 1$ \KwTo $a$ }{
        Perform Algorithm~\ref{alg:dueling_modified} with parameters: $\{\alpha_{i}\}$, $\{\beta_{i}\}$ and $T$ \;
        }
        Perform Algorithm~\ref{alg:dueling_modified} with parameters $\{\alpha_{i}\}$, $\{\beta_{i}\}$ and $P_{\text{last}}$ \;
        $\mathsf{num} \leftarrow a*T + T_{last}$ \;
        Measure both registers. Let the result in the first register be $x_1$ and the result of the second register be $x_2$. Output the better result.
\end{algorithm}

\begin{figure}
    \centering
    \includegraphics[width = \columnwidth]{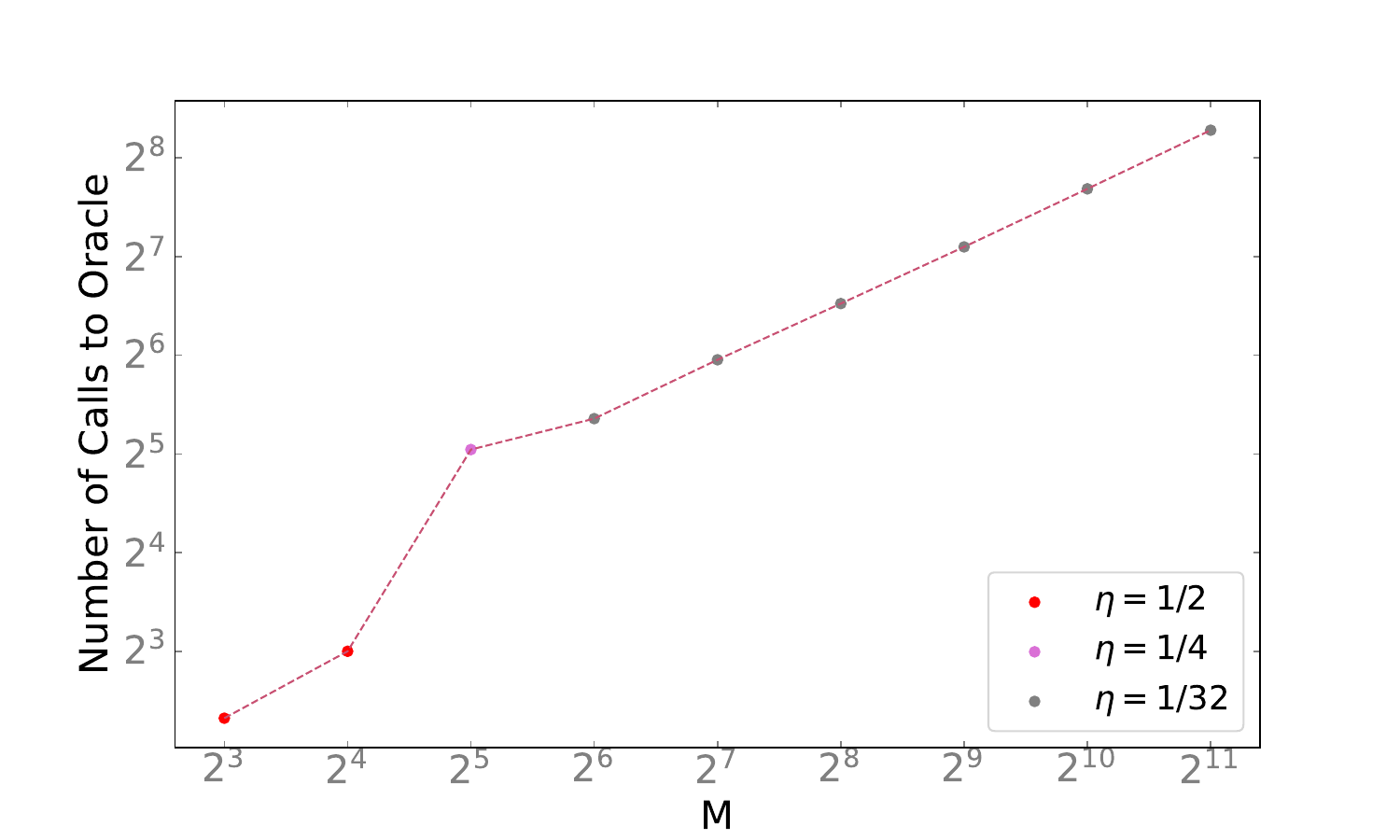}
    \caption{Log-log plot of the number of calls to the oracle with respect to the size of the original search space $M$. For any element $x$ in the original search space $S$, we have the measure function difference of adjacent elements $\Delta v(x) = 2$. As we add ancilla qubits, for every $x \in S$, we add its corresponding $x'$ elements with measure function values $v(x') = v(x) + 1$. We then run Algorithm~\ref{alg:dueling_iterative} under the parameter scheme $\alpha_{i} = \beta_{i} = c$. $P_{\text{bar}}$ is set to be $\frac{2}{3}$. We execute the algorithm under different control parameters $\eta$ and obtain the most optimal $\eta$ for each $M$.}
    \label{fig:opnl_m}
\end{figure}

Fig.~\ref{fig:opnl_m} illustrates the increasing trend of the total number of calls to the oracle with respect to the size of the original search space $M$, given the optimal selection of control parameters $\eta$. Notably, as $M$ enlarges, the increase in oracle calls aligns with an $O(\sqrt{M})$ trend. This suggests that the strategy of adding ancilla qubits is effective for large-scale initial search spaces, reflecting on the computational complexity considerations.

\subsection{Quantum Dueling for PMAX-3SAT}
\label{sec:PMAX_3SAT}

Besides the specific case when $M = N$ mentioned in Sec.~\ref{sec:opnl}, we also want to apply quantum dueling to a practical combinatorial optimization problem, namely PMAX-3SAT.

In PMAX-3SAT, we are given a set of boolean variable $x_1, x_2, \cdots, x_n$, and a literal is one single variable in its original form $x$ or negation form $\neg x$. The disjunction of literals form clause $C$, and since this is 3SAT so each clause contains exactly 3 literals, e.g. $C = x_1 \lor x_2 \lor \neg x_3$. The final formula, denoted as $e$, is a set of such clauses in conjunction, e.g. $F = C_1 \land C_2 \land \cdots \land C_m$ where part of the formula $F_1$ is used to determine whether an assignment satisfies the formula, and the other part $F_2$ is to determine how optimal an assignment is if it satisfies $F_1$. Under PMAX-3SAT, our goal is to find an assignment $x_{1} x_{2} \cdots x_{n}$ that satisfies $F_1$ and satisfies the most number of clauses in $F_2$.

In the encoding stage, our search space has size $N = 2^{n}$, where $n$ is the number of variables. This means that there are in total $N$ possible assignment configurations. We define $y_{i}$ as a binary indicator. When clause $C_{i}$ is satisfied, $y_{i}$ takes the value of 1, otherwise, $y_{i} = 0$. In practice, each assignment $X$ within this search space is a solution if it satisfies the first half of F, $F_1$, which means $\sum\limits_{0 \le i < m'} y_{i} = m'$ where $m'$ is the number of clauses in $F_1$. Furthermore, the measure function for each assignment $X$ is formed by $u(X) = \sum\limits_{m' \le j < m} y_{j}$, which is the number of clauses that are satisfied in $F_2$. In PMAX-3SAT, we want to maximize $u(X)$, and since our formulation of quantum dueling minimizes the measure function, in practice we set the measure function to $v(X) = -u(X)$.

For simplicity, we want to randomly generate PMAX-3SAT formulas and restrict the number of solutions M to be around $\sqrt{N}$. Therefore, we use the fraction $\frac{M}{N}$ to determine the number of clauses we have for the generated formula. Specifically, since for each assignment, the probability that it satisfies a single clause is $\frac{7}{8}$, so each clause is satisfied by around $\frac{7}{8}$ fraction of all possible assignments. Therefore, the number of clauses $m'$ in $F_1$ determines the expected M, that is $E(M) = N \left( \frac{7}{8} \right)^{m'}$. We get $m' = \left\lceil \log_{\frac{7}{8}} \left( \frac{1}{\sqrt{N}} \right) \right\rceil$. We then decide to make $m = 2m'$ so that we have same amount of clauses in $F_2$ to determine how optimize a solution is, so we generate $2\left\lceil \log_{\frac{7}{8}} \left( \frac{1}{\sqrt{N}} \right) \right\rceil$ number of clauses for $N = 1, 2, 4, \cdots, 4096$.

\begin{figure}
    \centering
    \includegraphics[width = \columnwidth]{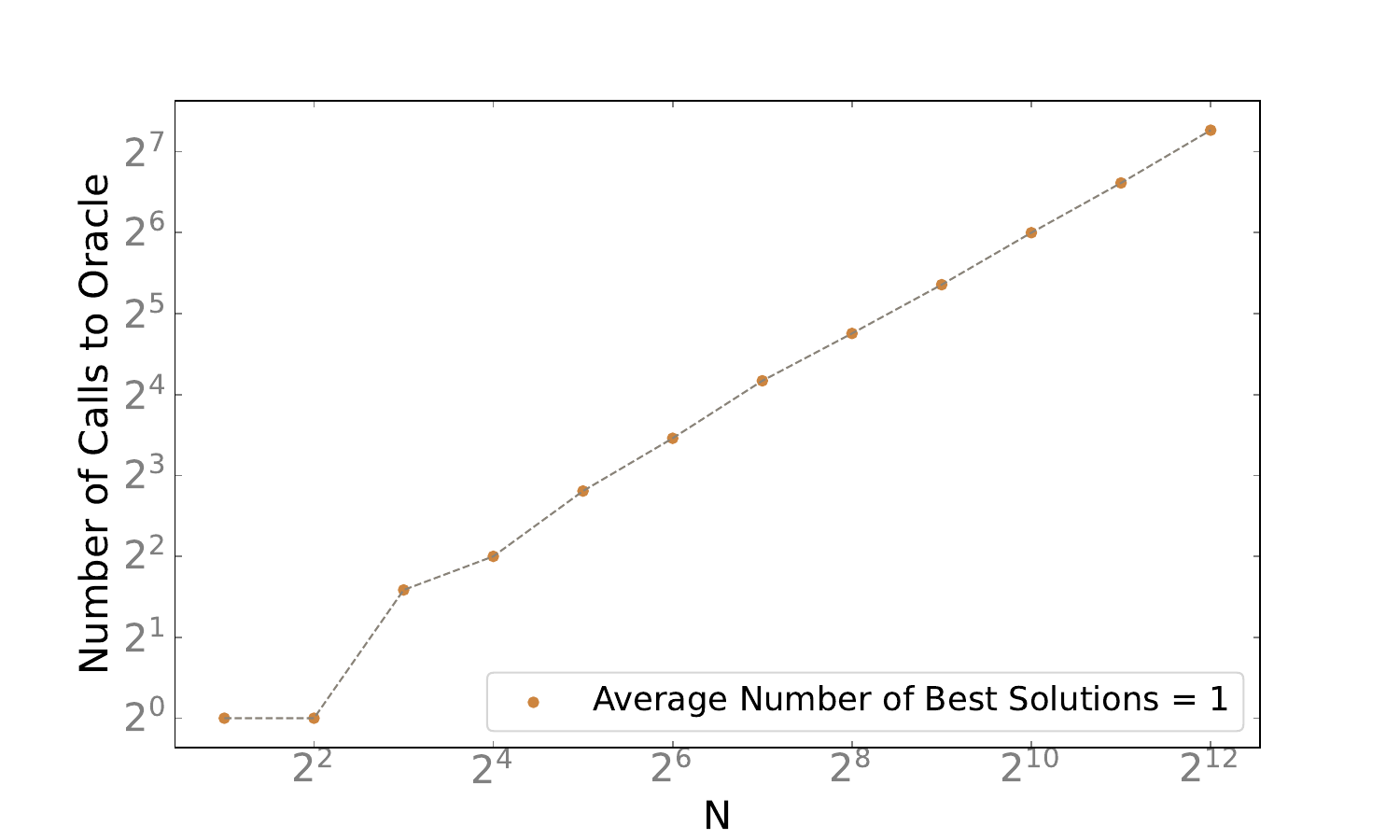}
    \caption{Log-log plot of the number of calls to the oracle $T$ with respect to the size of the search space $N$. Algorithm~\ref{alg:dueling_iterative} is executed iteratively under the parameter scheme $\alpha_{i} = \beta_{i} = c$ to obtain a constant threshold success probability $P_\text{bar} = \frac{2}{3}$. We randomly generate 10000 cases of Boolean formulas for each N. We then take the average value of the total number of calls to the oracle when $N$ is fixed.}
    \label{fig:P3SAT}
\end{figure}

The result of our program simulation is shown in Fig.~\ref{fig:P3SAT}, from which we observe a complexity trend around $O(\sqrt{N})$. Although a more exact determination of time complexity requires running with higher N and more possible parameter scheme, the trend shown in the graph confirms our conclusion that quantum dueling under some parameter scheme has the potential to solve PMAX-3SAT problem efficiently. Furthermore, a lot of known non-polynomial problem reduce to PMAX-3SAT. Being able to solve PMAX-3SAT efficiently proves quantum dueling's potential to solve more non-polynomial problems that are ``hard'' to solve by classical computers.

\section{Discussion}
\label{sec:discussion}

\subsection{Mechanisms of Dueling}
\label{sec:dueling_mechanism}

In previous sections, we have presented many methods to understanding the performance of quantum dueling in a wide range of circumstances. However, one key question remains: How does the algorithm actually work?

One main question regarding quantum dueling is the fact that when we repeatedly apply Grover gates, after $O(\sqrt{\frac{N}{M}})$ times the states will be ``overcooked.'' If we try to apply Grover gates again, the corresponding quantum states will no longer be optimized but worse. 

As an intuitive answer, quantum dueling circumvents this problem because the number of Grover iterations required for such a phenomenon to happen is not the same when we choose different elements from the search space as the ``opponent'' in dueling. Consider different $y \in S$ and corresponding $G_y$ gates as used in Eq.~(\ref{eq:redefinition_of_oracles}). For better $y$, the number of rotations required for overcooking is larger. As long as $\{\alpha_i\}$ and $\{\beta_i\}$ are sufficiently small compared with $\sqrt{N}$, by the time that overcooking takes place among the worse elements, the amplitudes are already boosted away towards the more desirable elements, so this effect generally does not occur before quantum dueling already boosted the amplitudes to the largest.

Another puzzling issue is that at times during the algorithm, we do not actually apply Grover gates on a uniform linear combination of all the solutions and also non-solutions. Rather, the components $\psi_{kl}$ follow an unknown trend, where its absolute value is usually decreasing with $k$ and $l$ representing worse and worse candidates for the search space. In such a case, the Grover amplifier should only have a boosting effect for the projection of the state onto the space spanned by the mean of solutions and the mean of non-solutions identified by the oracle. In that way, we should expect the Grover gates not to have boosted quantum states significantly such that the success probability can by any means be boosted to a significant value. 

Our results, especially those in Sec.~\ref{sec:parameter}, show the opposite. While in some extreme cases, a naive parameter scheme fails to boost success probability significantly, in almost all other situations, quantum dueling boosted success probability towards a significant value even for a large $N$. It seems that the iterative procedures of switching between $\mathcal{G}_{1\to 2}$ and $\mathcal{G}_{2\to 1}$ gates somehow hybridized the state in a way that the issue raised is completely circumvented. To truly understand how this can happen, we need to find an approximate solution for the evolution of the state vector. In fact, we believe solving such a task is not only extremely important for the quantum dueling algorithm itself, but would be interesting mathematical research that might offer new insights to the field of quantum computing in general.

\subsection{Generalizations of Quantum Dueling}
\label{sec:generalization}

In Sec.~\ref{sec:parameter}, we found that while quantum dueling in principle relies a sequence of external parameters that need to be fixed, the degree of freedom in these parameters can be drastically reduced with a good parameter scheme. In such an event, quantum dueling demonstrate a strong performance, exhibiting a quadratic speedup compared with classical search. Yet, are there any other ways to improve our result furthermore? For example, can we boost the success probability to arbitrarily high value without classical repetition of the algorithm as done in Sec~\ref{sec:application}?

As discussed in Sec.~\ref{sec:dueling_mechanism}, the main reason behind quantum dueling's dependency on solution distribution is not its design idea itself, but the underlying issues with Grover's algorithm. In Sec.~\ref{sec:hamiltonian_version}, we mentioned that quantum dueling can be further generalized with arbitrary controlled phase factors and also different mixer and problem Hamiltonian setups. Putting things further, quantum dueling should be even more generalized to the idea of augmenting the Hilbert space and transforming optimization into comparison, utilizing the underlying entanglement principles. If that is the case, in principle, any quantum search algorithm that can perform amplitude amplification can be used as an underlying mechanism for quantum dueling. For example, adiabatic search as algorithms as discussed in \cite{Roland_2003, Morley_2019} can be used in place of Grover gates in the construction of quantum dueling. If we can find a way to realize comparison via adiabatic means, it would be possible to realize quantum dueling even in NISQ (noisy intermediate-scale quantum) devices. A more thoroughly thought quantum search scheme can not only reduce quantum dueling's dependence on solution distributions and hyperparameters but also offer unexpected gains.


\subsection{Quantum Dueling as a Quantum Subroutine}

\label{sec:subroutine}

In Appendix~\ref{app:GAS}, we note that quantum dueling can, in some ways, be seen as a quantized version of Grover Adaptive Search (GAS), which also achieved a quadratic speedup---the theoretical limit of quantum computing for generic optimization problems.  This raises a pertinent question: if both algorithms attain similar speedups, what distinguishes quantum dueling from GAS? Why is such algorithm design important?

Quantum dueling's primary strength is its fully quantum nature, in contrast to the hybrid approach of GAS, which relies on periodic measurements to optimize the state based on the objective function. This intrinsic advantage of quantum dueling allows it to function as an integrated subroutine within broader quantum algorithms, a capability not feasible with hybrid models. As an example, consider a quantum circuit with two registers; so its basis states can be written as $\ket{x}\ket{y}$ in some Hilbert space $\mathscr{H}_x \otimes \mathscr{H}_y$. Given some entangled state, we would like to optimize only the second register based on some objective function that is dependent on both $x$ and $y$, say $v(x,y)$. In other words, starting with some $\ket{\psi_0} \in \mathscr{H}_x \otimes \mathscr{H}_y$ we would like to achieve some quantum state $\ket{\psi} \mathscr{H}_x \otimes \mathscr{H}_y$ such that for each $\ket{x}$, the projected state of $\ket{\psi_0}$, i.e., $\bra{x}\ket{\psi_0}$ is optimized to $\bra{x}\ket{\psi}$ \footnote{This is partial inner product.} such that $|\bra{y_\text{optimal}(x)}\left(\bra{x}\ket{\psi}\right)|^2$ is sufficiently large. Here $y_\text{optimal}(x)$ is a function of $x$ that represent the $y$ that optimize $v(x,y)$ for each $x$.

Note that such a task cannot be completed with a hybrid algorithm such as GAS. We cannot use measurement to obtain a current best guess and further optimize. Rather, a measurement on the register corresponding to $y$ will leave us with only one $y$ value and the corresponding projection to $\mathscr{H}_x$, losing the broader context provided by the quantum state. Quantum dueling, however, is a fully quantized algorithm without the need of measurement. Therefore, it can, in principle, complete such a task. There are two issues that need to be resolved. First, since quantum dueling would essentially be run parallel for all possible $x$, the best hyper-parameter for different $x$ values may be different. Choosing a single fixed hyper-parameter might not produce desirable results. This can be resolved by using a quantized, controlled operation to automatically determine which gate to apply based on the specifics of $x$, embedded within the quantum circuit itself. Of course, this would require enormous quantum computational power that is not currently practical. However, this is possible theoretically and realizable in the future. Second, there are no assumptions made about the initial state. In all our experiments we initialized the state to be a uniform linear combination of all possible elements from the search space, however, the initial state in this case will depend on the algorithm itself. It is unclear how quantum dueling will perform for a more generic initial state. For this issue, note that the dependency of quantum dueling does not originate from the idea of dueling itself, but the dependence on initial state in Grover search. Therefore, we can consider other quantum search algorithms as discussed in Sec.~\ref{sec:generalization}.

In the realm of optimization problems, it's rare to encounter scenarios devoid of any structure, where the only solution is a brute-force approach across all possible candidates. More often, these problems contain inherent properties that allow an efficient solution. While an algorithm like Grover Adaptive Search (GAS) offers a quantum-accelerated brute-force method, its application is confined to these generic cases. Quantum dueling, on the other hand, transcends this limitation. Its utility is not restricted to just accelerating brute-force searches but extends to enhancing any process aimed at identifying the best solution from a set of candidates,  even if such a process is quantum-based. Such versatility makes quantum dueling an invaluable tool not only for direct optimization but also as a subroutine in more complex quantum algorithms where an optimization process is required.

\section{Conclusion}
\label{sec:conclusion}

In this work, we proposed a quantum algorithm for combinatorial optimization. The algorithm, quantum dueling, operates on an augmented Hilbert space with twice the size of qubits. At each time, we can choose one register as a control to apply a controlled Grover gate to the other register, selecting states that are ``better'' than the control. In such a way, the state slowly evolves towards the optimal direction. We analyzed the performance of quantum dueling via computer simulation and offered some mathematical contraction that could potentially be useful to find an approximate solution for the evolution of the state vector.  Using such a contraction, we investigated the performance of quantum dueling by setting a variety of parameter selection schemes. We demonstrated that while quantum dueling requires a set of hyper-parameters, even a naive set of parameter schemes can result in an algorithm reaching significant success probability in a quadratic speedup compared with a classical brute-force search. We further tested our hypothesis to more practical scenarios such as the case where there is no solution labeling and PMAX-3SAT.

Overall, it is clear that quantum dueling opens a new route by which optimization problems can be realized, characterized by its efficiency, versatility, and general applicability. From a conceptual perspective, we demonstrated how, by multiplying the number of qubits used and utilizing the underlying entanglement principle, we can create previously unthought-of and interesting quantum algorithms that can offer new insights into theoretical quantum computing. Previously, such a procedure was only seen in quantum error correction, where several physical qubits are used to simulate one logical qubit \cite{Roffe_2019}. Our paper essentially applied such ideas to quantum algorithm design in general. Mathematically, there is so much to be learned. Can we truly understand how Algorithm~\ref{alg:dueling} somehow bypasses the apparent impossibilities to effectively boost the state vector and finally provide a quantitative solution? Can we find a way to understand how hyper-parameters impact the performance of quantum dueling and apply such methodologies to other variational algorithms such as QAOA (quantum approximate optimization algorithm)? Most importantly, can we apply the idea of dueling to more specific optimization problems with underlying structures, utilizing a fully quantized oracle of comparison and more in highly entangled memory space? Trying to answer these questions will open up new ideas that can help advance the development of quantum algorithm design.

Codes to reproduce the data and figures in this paper are available online \cite{codes_on_github}.


\appendix

\section{Grover Search}
\label{app:Grover_search}

\begin{algorithm}
\caption{Grover Search}
\label{alg:GS}
\SetKw{False}{false}
\SetKw{True}{true}
\KwIn{A parameter $r_{max}$ for maximum rotation counts}
\KwOut{An element $x$ in the search space satisfying $f(x)= 1$ (Approximately)}
$|\psi\rangle\leftarrow H^{\otimes n}|0\rangle$ \;
$|\psi\rangle\leftarrow \left(D O\right)^{r_{max}} |\psi\rangle$ \;
Measure the state $|\psi\rangle$, and output the result $x$.
\end{algorithm}

As summarized in Algorithm ~\ref{alg:GS}, the Grover search algorithm aims at finding a solution $x$ satisfying $f(x) = 1$ in the search space $S$ comprising of $N = |S| = 2^{n}$ elements, $M$ of them are the solutions we want. The algorithm operates on the $2^{n}$ - dimensional Hilbert Space $\mathscr{H}_{S}$ comprising of a set of orthonormal basis $B=\{\ket{x}|x\in S\}$.  

To illustrate the idea better, we represent the state vector $|\psi\rangle$ at any stage after initialization as a rotating vector in the 2-dimensional plane, the axes being the mean superposition of all solutions $|\varphi\rangle$ and non-solutions $|\gamma\rangle$. 
\begin{equation}
    \label{eq: mean sol non-sol}
    \left\{
    \begin{gathered}
    |\varphi\rangle = \frac{1}{\sqrt{M}}\sum_{x}^{f(x)=1} \; |x\rangle\hfill  \\
    |\gamma\rangle = \frac{1}{\sqrt{N-M}}\sum_{x}^{f(x)=0} \;  |x\rangle\hfill\\
    \end{gathered}
    \right.
\end{equation}

The mean superposition of all elements of the initial state $|\psi_{0}\rangle$ will be: 
\begin{equation}
    \label{eq: mean superposition}
    \begin{aligned}
    |\psi_{0}\rangle = |m\rangle &= \frac{1}{\sqrt{N}} \sum_{x} |x\rangle \\
    &= \sqrt{\frac{M}{N}} |\varphi\rangle + \sqrt{\frac{N-M}{N}} |\gamma\rangle \\
    \end{aligned}
\end{equation}

Originally, the initial state vector in Grover search is defined to be $\ket{m}$, the equal superposition of all basis states. However, it is suggested in \cite{Brassard_2002} that the initial state vector does not necessarily need to be the equal superposition of all basis, rather it could be in principle an arbitrary state vector $\ket{m_{0}}$ in the framework of generalized Grover search.
 
Therefore, we could interpret the process of Grover search by rotating the state vector in the two-dimensional vector space spanned by the basis vector $|\varphi\rangle$ and $|\gamma\rangle$. The state can be specified by an angle $\sigma$, which the rotating state vector makes with the axis $|\gamma\rangle$, with its initial value $\theta = \arcsin{\sqrt{\frac{M}{N}}}$.

In this regard, at any stage after initialization, the state vector $|\psi\rangle  = \sin{\sigma}|\varphi\rangle + \cos{\sigma} |\gamma\rangle $, and when applied to the Grover Operator $G=DO$, the state becomes: 
\begin{equation}
    \label{eq: apply G}
    \begin{aligned}
    G|\psi\rangle 
    &= \left(2 |\psi_{0}\rangle\langle\psi_{0}| -I\right) O \left(\sin{\sigma}|\varphi\rangle + \cos{\sigma |\gamma\rangle} \right)\\
    &= \left(2 |\psi_{0}\rangle\langle\psi_{0}| -I\right) \left(-\sin{\sigma}|\varphi\rangle + \cos{\sigma |\gamma\rangle} \right) \\
    &= \left(2 \left(\sin{\theta} |\varphi\rangle + \cos{\theta} |\gamma\rangle \right)\left(\sin{\theta} \langle\varphi| + \cos{\theta} \langle\gamma| \right) - I\right)\\
    &\phantom{=} \left(-\sin{\sigma}|\varphi\rangle + \cos{\sigma |\gamma\rangle} \right) \\
    &= \sin(2\theta + \sigma) |\varphi\rangle + \cos(2\theta + \sigma) |\gamma\rangle \\
    \end{aligned}
\end{equation}

From Eq.~(\ref{eq: apply G}), we could conclude that each time the state vector $|\psi\rangle$ is rotated by $2\theta$ towards the $|\varphi\rangle$ axis. After applying Grover Operator $r_{max}$ times, the angle which the state vector makes with the $|\gamma\rangle$ axis in the 2-dimensional plane becomes: 
\begin{equation}
    \label{eq: r times angle}
    \sigma_{max} = \left(  2r_{max} + 1 \right) \theta
\end{equation}

After conducting a measurement on the final state, the probability of finding one of the solutions is $\sin^{2}{\sigma_{max}}$, which is close to unity when $\sigma_{max}\approx \frac{\pi}{2}$. The least number of iterations needed for a near-optimal result is $r\approx \frac{\pi}{4}\sqrt{N}$, which means the time complexity of Grover search is $O(\sqrt{N})$. We require $N \gg 1$ for the above analysis to hold.

It has been shown in \cite{Boyer_1998}, \cite{Zalka_1999} that Grover search algorithm is optimal for any oracle-based quantum search scheme, with its time complexity reaching $O(\sqrt{N})$.

\section{Grover Adaptive Search}
\label{app:GAS}

Similar to what was mentioned in Sec.~\ref{sec: algorithm}, Grover adaptive search copes with the problem of finding an element $x$ in the search space, satisfying: 
\begin{equation}
    \label{eq:gas_goal}
    v(x)=\min\{v(y)| y\in S \wedge f(y)=1\}
\end{equation}

For simplicity, we will abbreviate Grover Adaptive Search for GAS.

The condition labeled in Eq.~\ref{eq:gas_goal} slightly differs from that proposed by Baritompa in the original version of GAS \cite{Baritompa_2005}, as we have not only taken into account the value of measure function $v(x)$ but also the value of the solution-indicator function $f(x)$, distinguishing between solutions and non-solutions.  

\begin{algorithm}
    \caption{Grover Adaptive Search}\label{alg:GAS}   
    \KwIn{Measure function $v$ represented in quantum arithmetic; parameters $k=1$, $\lambda > 1$}
    \KwOut{An element $\mathsf{best}$ in the search space satisfying $f(\mathsf{best})=1$ and $v(\mathsf{best})$ (approximately) minimized}
    
    Using Grover search with an unknown number of solutions as proposed in \cite{Boyer_1998} for a $\mathsf{best} \in S$ such that $f(\mathsf{best}) = 1$\;
    \Repeat{\text{a termination condition is met}}{ Randomly select rotation count $r_{i}$ from the set $\{0,1,\cdots,\lceil k-1 \rceil \}$ \;
    Initialize the state vector $|\psi_{0}\rangle \gets H^{\otimes n} |
    0\rangle$ \;
    Apply Grover search with $r_{i}$ iterations on $|\psi_{0}\rangle$ for a solution $\mathsf{new}$ such that $f(\mathsf{new}) = 1$ and $v(\mathsf{new}) < v(\mathsf{best})$ \;
    \eIf{$v(\mathsf{new}) < v(\mathsf{best})$}
    {$\mathsf{best} \gets \mathsf{new}$;  $v(\mathsf{best}) \gets v(\mathsf{new})$; $k \gets 1$ \;
    }
    {$k\leftarrow\lambda k$ \;
    }}
\end{algorithm}

Algorithm~\ref{alg:GAS} shows the process of GAS for solving combinatorial optimization problems formalized by \cite{Bulger_2003,Baritompa_2005}. In the algorithm, a best current guess $\mathsf{best}$ is repeatedly updated via a Grover search. Ignoring ancilla qubits, there are two components of memory storage in GAS:
\begin{enumerate}
    \item A quantum storage of $n$ qubits. In each iteration, it is reset for a Grover search.
    \item A classical storage documenting the current best guess.
\end{enumerate}

The classical storage stores the current best guess $\mathsf{best}$ after $i$ iterations as well as the measure function $v(\mathsf{best})$, and the quantum storage stores the state vector $\ket{\psi}$, which is initialized after each iteration. In the $i_{th}$ iteration, we apply Grover operator $G$ on the quantum storage $r_{i}$ times, followed by a measurement. And we update the current best guess $\mathsf{best}$ if the measured result is "better" than the current one stored in the classical storage. Thus, GAS is a combination of Grover and the classical adaptive search. 

According to \cite{Baritompa_2005}, the rotation counts should vary to avoid the breakdown of GAS for certain specific solutions in the search space $S$. The choice of the Grover rotation counts $r_{i}$ of the $i_{t h}$ iteration of GAS is determined by the parameter $k$, where $r_{i}$ is randomly chosen from the set $\{ 0, 1, \cdots, \lceil k-1 \rceil \}$. 

The performance of GAS is shown to be dependent upon the choice of parameters, especially the parameter $\lambda$, which serves as the scaling factor of $k$ when the result produced by one measurement is not satisfactory. Baritompa suggests in \cite{Baritompa_2005} via numerical simulations that when $\lambda = 1.34$, the performance of GAS reaches an optimal bound, with a relatively high success probability of finding the most optimized solution. He also shows that the time complexity of GAS is approximately $O(\sqrt{N})$.

To develop the idea of GAS to quantum dueling, the current best guess, originally stored in a classical memory in GAS, will be stored in a quantum register. This register thus contains an entangled state of potential solutions, with more optimized solutions having higher amplitudes. In each iteration, the states in both registers are updated, leading to increases in amplitude of more optimized solutions. Compared to GAS, quantum dueling performs its algorithm in a fully quantized system without classical memory involved. Consequently, there will no longer be a requirement for measurements at the end of each iteration. Therefore, the first register serves the same purpose as the second one, storing the current best guess inside, which leads to a symmetry in the algorithm design.

In some ways, the idea of quantum dueling is generated from the observations of the original Grover adaptive search by quantizing the entire search process using an augmented Hilbert Space.

\section{On Uniform Distribution}
\label{app:distribution}
In Figs.~\ref{fig:initial_simulation}--\ref{fig:heuristics_uniform_N} and Table~\ref{tab:initial_simulation}, the solution distribution often takes the following form:
\begin{equation}
    f(x) = \left[x \equiv t \pmod{s}\right]
\end{equation}
To make the formula more aesthetically appealing in edge cases such as when $s = 1$, we would also write it in forms such as:
\begin{equation}
    f(x) = \left[x - t \equiv 0 \pmod{s}\right]
\end{equation}
Another reason for this form is that when $t = 1$, $x-1$ will be the element's index starting from $0$, as commonly done in computer science. 

In Figs.~\ref{fig:initial_simulation_all_M}--\ref{fig:heuristics_uniform_N}, we often used the case $t=1$, mostly because quantum dueling behaves very regularly under this scenario. However, it must be noted that the $t=1$ case is overly optimistic and should always be taken with a grain of salt. 

Lastly, in our article, we have consistently used $M$ to denote the number of solutions. Without loss of generality let $1\leqslant t \leqslant s$, for $M$ and $s$ to be consistent, we must have
\begin{align}
    &(M-1)s + t \leqslant N < Ms+t \label{eq::app_uniform_original} \\
    &\Rightarrow \frac{N-t}{M} < s \leqslant \frac{N-t}{M-1}
\end{align}
For $M = 1$ case we let $\frac{N-t}{M-1} = +\infty$. For $t=1$, Eq.~(\ref{eq::app_uniform_original}) gives us
\begin{equation}
    (M-1)s < N \leqslant Ms
\end{equation}
which is reduced to:
\begin{equation}
    \frac{N}{M} \leqslant s < \frac{N}{M-1} 
\end{equation}

Since $s$ is always an integer, for a valid $s$ to exist, $M$ must satisfy $\left\lceil\frac{N}{M}\right\rceil < \frac{N}{M-1}$. In this case, it is intuitive to choose $s = \left\lceil\frac{N}{M}\right\rceil$, as have done in Figs.~\ref{fig:initial_simulation_all_M} and \ref{fig:heuristics_uniform}.

In our analysis, we often choose $n = 2^{2k}$ where $k$ is some positive integer; this allows us to set $M = \sqrt{N} = 2^k$ without further justification. In the following theorem, we expand on this fact and prove that we can find a uniform distribution as long as $M = \left\lceil\sqrt{N}\right\rceil$ without any restriction on $N$.
\begin{theorem}
    For positive integer $N$, when $M = \left\lceil\sqrt{N}\right\rceil$, $\left\lceil\frac{N}{M}\right\rceil < \frac{N}{M-1}$.
\end{theorem}
\begin{proof}
    Let $N = a^2 - b$, where $a,b\in\mathbb{N}$ and $0\leqslant b\leqslant 2a-2$. This is possible because $(a-1)^2 = a^2 - 2a + 1 = a^2 - (2a-2) - 1$. Thus, $M = \left\lceil\sqrt{N}\right\rceil = a$. When $M=a=1$, $\frac{N}{M-1} = +\infty$, so the case is trivial. Otherwise, we have:
    \begin{equation}
        \left\{\begin{gathered}
            \left\lceil\frac{N}{M}\right\rceil = \left\lceil\frac{a^2-b}{a}\right\rceil = a + \left\lceil -\frac{b}{a}\right\rceil \hfill \\ 
            \frac{N}{M-1} = \frac{a^2-b}{a-1} = a+\frac{a-b}{a-1} \hfill \\
        \end{gathered}\right.
    \end{equation}
    When $0 \leqslant b \leqslant a-1$, $\left\lceil\frac{N}{M}\right\rceil = a$ and $\frac{N}{M-1} \geqslant a+ \frac{1}{a-1} > a$. So theorem holds. When $a\leqslant b \leqslant 2a-2$, $\left\lceil\frac{N}{M}\right\rceil = a-1$ and $\frac{N}{M-1} \geqslant a - 1 + \frac{1}{a-1} > a - 1$, so theorem also holds. 
\end{proof}

\section{Mathematical Formalization of Sec.~\ref{sec:cluster_representation}}
\label{app:cluster}
In this section, we provide mathematical formalization for statements in Sec.~\ref{sec:cluster_representation}. As shown below, the relation $\sim$ defined by Definition \ref{def:cluster_relation} is restated to Definition \ref{def:cluster_relation_restate} for clarity. We then prove that this relation is an equivalence relation.

\begin{definition}
    \label{def:cluster_relation_restate}
    ${\sim}$ is a relation on $S$ such that for all $x,y \in S$, $x{\sim} y$ if and only if either of the following conditions is true:
    \begin{enumerate}
        \item $f(x) = f(y) = 1 \wedge  v(x) = v(y)$;
        \item $f(x) = f(y) = 0 \wedge \neg  \exists\;  z \in S \; f(z) = 1 \wedge \min\left\{v(x),v(y)\right\} \leqslant v(z) < \max\left\{v(x),v(y)\right\}$. 
    \end{enumerate}
\end{definition}

\begin{theorem}
    \label{thm:equivalence_relation}
    ${\sim}$ is an equivalence relation on $S$.
\end{theorem}
\begin{proof}
    It suffices to show reflexivity, connectivity, and transitivity  hold for ${\sim}$. We can show reflectivity and connectivity directly by definition. For transitivity consider $x,y,z\in S$ satisfying $x \sim y$ and $y\sim z$. First we have $f(x) = f(y) = f(z)$. If $f(x) = 1$ it is trivial that $x \sim z$. If $f(x) = 0$, without loss of generality we assume $v(x) \leqslant v(z)$. For cases $v(y)\leqslant v(x)$, $v(x) < v(y) \leqslant v(z)$, and $v(y) > v(z)$, we can show separately that $x \sim z$. 
\end{proof}

In Sec.~\ref{sec:cluster_representation}, we claimed that components of the state vector in directions that belong to the same cluster stay uniform during quantum dueling. To clarify, Theorem~\ref{thm:cluster_property} is copied to Theorem~\ref{thm:cluster_property_app}, which is proved after introducing Lemma~\ref{lem:cluster}.

\begin{lemma}
    \label{lem:cluster}
    Let $k_1,k_2,l_1,l_2 \in S$ satisfying $k_1 {\sim} k_2$ and $l_1 {\sim} l_2$. $o(k_1,l_1) = o(k_2,l_2)$.  
\end{lemma}
\begin{proof}
    We prove the following two propositions. 
    \begin{enumerate}
        \item $o(k_1,l_1) = o(k_2,l_1)$.  
        
        If $f(k_1) = 1$, this is trivial; if $f(k_1) = 0$, $o(k_1,l_1) = o(k_2,l_1) = 1$.
        
        \item $o(k_2,l_1) = o(k_2,l_2)$.  
        
        When $f(l_1) = 1$, this is trivial;when $f(l_1) = 0$, $f(k_2) = 0$, $o(k_2,l_1) = o(k_2,l_2) = 1$. 
        
        If $f(l_1) = 0$, $f(k_2) = 1$, by definition of ${\sim}$ either $v(k_2) < \min\{v(l_1),v(l_2)\}$ or $v(k_2) \geqslant \max\{v(l_1),v(l_2)\}$. 
        
        In either case $o(k_2,l_1) = o(k_2,l_2)$.
    \end{enumerate}
    Thus, $o(k_1,l_1) = o(k_2,l_2)$. 
\end{proof}

\begin{theorem}
    \label{thm:cluster_property_app}
    Let $k_1,k_2,l_1,l_2 \in S$ satisfying $k_1 {\sim} k_2$ and $l_1 {\sim} l_2$. At any stage after initialization, $\langle k_1l_1| \psi \rangle = \langle k_2l_2 | \psi \rangle$.
\end{theorem}
\begin{proof}
    It is clear that $\langle k_1l_1| \psi \rangle = \langle k_2l_2 | \psi \rangle$ holds after initialization, as for all $k,l\in S$, $\langle kl|\psi\rangle = \frac{1}{N}$. 
    
    Assume at some time after initialization in Algorithm~\ref{alg:dueling}, $\forall k_1'\sim k_2' \; \forall l_1'\sim l_2' \; \langle k_1'l_1'| \psi \rangle = \langle k_2'l_2' | \psi \rangle$. Apply Lemma~\ref{lem:cluster} gives:
    \begin{equation}
        \label{eq:cluster_property_induction_step}
        \left\{\begin{gathered}
            (\mathcal{G}_{1\leftarrow 2}|\psi\rangle)_{k_1l_1} = (\mathcal{G}_{1\leftarrow 2}|\psi\rangle)_{k_2l_2} \hfill \\
            (\mathcal{G}_{2\leftarrow 1}|\psi\rangle)_{k_1l_1} = (\mathcal{G}_{2\leftarrow 1}|\psi\rangle)_{k_2l_2} \hfill \\
        \end{gathered}\right.
    \end{equation}
    
    Thus, by mathematical induction Theorem~\ref{thm:cluster_property} holds.
\end{proof}

As mentioned in Sec.~\ref{sec:cluster_representation}, elements from the same cluster share equivalent $f$ value. Meanwhile, only on the boundary will the interval spanned by $v$ values of the elements in each cluster coincide. To prove this, we first need to define the $v$-interval of each cluster in Definition~\ref{def: cluster interval}.

\begin{definition}
\label{def: cluster interval}
    For each cluster $\mathcal{S} \in S/{\sim}$, define its $v$-minimum $\mathinner{\min}_v\mathcal{S}$ and $v$-maximum $\mathinner{\max}_v\mathcal{S}$ as:
    \begin{equation}
        \label{eq:v_extremum}
        \left\{\begin{gathered}
            \mathinner{\min}_v\mathcal{S} = \min\setbuilder{v(x)}{x\in \mathcal{S}} \hfill \\
            \mathinner{\max}_v\mathcal{S} = \max\setbuilder{v(x)}{x\in \mathcal{S}} \hfill 
        \end{gathered}\right.
    \end{equation}
    We then define the cluster's $v$-interval $I_v(\mathcal{S})$ to be $\left(\mathinner{\min}_v\mathcal{S},\mathinner{\max}_v\mathcal{S}\right)$. 
\end{definition}

It is worth noting that $\forall \mathcal{S} \in S/\sim \;$, if $f(\mathcal{S}) = 1$, $I_v(\mathcal{S}) = \varnothing$; otherwise, $I_v(\mathcal{S}) \neq \varnothing$. In general, we have outlined the range a given cluster spans in the space of measure function $v(x)$. Therefore, the non-overlapping properties of different clusters can be expressed as $I_{v} (\mathcal{S}) \cap I_{v} (\mathcal{T}) = \varnothing$ for arbitrarily different clusters $\mathcal{S}$, $\mathcal{T}$, as presented in Lemma~\ref{lem:cluster_boundary_collide}.

\begin{lemma}
    \label{lem:cluster_boundary_collide}
    Consider $\mathcal{S},\mathcal{T} \in S/{\sim}$. If $\mathcal{S} \neq \mathcal{T}$, $I_v(\mathcal{S}) \cap I_v(\mathcal{T}) = \varnothing$.
\end{lemma}
\begin{proof}
    Prove by contradiction. If $I_v(\mathcal{S}) \cap I_v(\mathcal{T}) \neq \varnothing$, then one of endpoints of either $I_v(\mathcal{S})$ or $I_v(\mathcal{T})$ must be within the other interval. Without loss of generality, assume a endpoint of $I_v(\mathcal{S})$ belongs to $I_v(\mathcal{T})$. Thus, either $\mathinner{\min}_v\mathcal{S} \in I_v(\mathcal{T})$ or $\mathinner{\max}_v\mathcal{S} \in I_v(\mathcal{T})$, implying that there is a $x \in \mathcal{S}$ such that $v(x) \in I_v(\mathcal{T})$, i.e., $\mathinner{\min}_v\mathcal{T} < v(x) < \mathinner{\max}_v\mathcal{T}$. Meanwhile, if either $f(\mathcal{S}) = 1$ or $f(\mathcal{T}) = 1$, $I_v(\mathcal{S}) \cap I_v(\mathcal{T}) = \varnothing$. Thus, $f(\mathcal{S}) = f(\mathcal{T}) = 0$. Therefore, by Definition~\ref{def:cluster_relation} applied to the elements of $\mathcal{T}$ with smallest and greatest $v$ values, no solution has $v$ value within $[\mathinner{\min}_v\mathcal{T},\mathinner{\max}_v\mathcal{T})$. Consider a $y \in \mathcal{T}$, it is clear that $\mathinner{\min}_v\mathcal{T} \leqslant v(y) \leqslant \mathinner{\max}_v\mathcal{T}$. Since $\mathinner{\min}_v\mathcal{T} < v(x) < \mathinner{\max}_v\mathcal{T}$, no solution has $v$ value within $\left[\min(v(x),v(y)),\max(v(x),v(y))\right)$, meaning that $x\sim y$. Since we know that $x \in \mathcal{S}$ and $y \in \mathcal{T}$, we must have $\mathcal{S} = \mathcal{T}$ by properties of $S/\sim$. We assume $\mathcal{S} \neq \mathcal{T}$, so we have a contradiction.
\end{proof}

In order to prove the validity of sorting clusters with respect to their indices, we need to construct a total order on the space $S/\sim$. First of all, we define the relations $\leqslant_{v}$, $\leqslant_{f}$ on $S/\sim$ in Definition~\ref{def:v_f_orders}.

\begin{definition}
    \label{def:v_f_orders} 
    Define relation $\leqslant_v$ on $S/\sim$ such that for all $\mathcal{S},\mathcal{T} \in S/\sim$, $\mathcal{S} \leqslant_v \mathcal{T}$ if and only if $\forall x\in\mathcal{S} \; \forall y\in\mathcal{T}\; v(x)\leqslant v(y)$. Define relation $\leqslant_f$ on $S/\sim$ such that for $\mathcal{S},\mathcal{T} \in S/\sim$, $\mathcal{S} \leqslant_f \mathcal{T}$ if and only if $f(\mathcal{S}) \leqslant f(\mathcal{T})$. 
\end{definition}

As in Theorem~\ref{thm:v_f_orders}, $\leq_{v}$ and $\leq_{f}$ are total orders on $S/\sim$.

\begin{theorem}
    \label{thm:v_f_orders}
    $\leqslant_v$ and $\leqslant_f$ are total orders if we define corresponding equivalence relations $\sim_v$ and $\sim_f$ to satisfy anti-symmetry.
\end{theorem}

\begin{proof}
    Lemma~\ref{lem:cluster_boundary_collide} implies that for each cluster $\mathcal{S}$, we can find a $v$-indicator $v(\mathcal{S})$, which is a generalization of the argument of $v$ to allow it to intake a register. (For example, $v(\mathcal{S}) = \frac{1}{2}\left(\mathinner{\min}_v \mathcal{S} + \mathinner{\max}_v \mathcal{S}\right)$) In this case, $\leqslant_v$ becomes $\leqslant$ in mapped structure $\setbuilder{v(\mathcal{S})}{\forall \mathcal{S} \in S/{\sim}}$. With the equivalence notation $\sim_v$, $\leqslant_v$ is a total order. Similarly, $\leqslant_f$ is a total order with properly defined relation $\sim_f$, since it is $\leqslant$ on the mapped structure $\setbuilder{f(\mathcal{S})}{\forall \mathcal{S} \in S/{\sim}}$. 
\end{proof}

After the definition of total orders $\leqslant_{v}$, $\leqslant_{f}$ on $S/\sim$ in Definition~\ref{def:v_f_orders}, we could then define the relation $\leqslant_{\sim}$ on $S/\sim$, and prove that it is a total order on $S/\sim$ in Theorem~\ref{thm:cluster_order}. Note that Definition~\ref{def:v_sim_orders_app} is a restatement of Definition~\ref{def:v_sim_orders} for the sake of clarity.

\begin{definition}
\label{def:v_sim_orders_app}
$\leqslant_\sim$ is the lexicographic combination of $\leqslant_v$ and $\leqslant_f$. That is, for $\mathcal{S},\mathcal{T} \in S/{\sim}$, $\mathcal{S} \leqslant_\sim \mathcal{T}$ if and only if $\mathcal{S} <_v \mathcal{T}$ (that is, $\neg \mathcal{T} \leqslant_v \mathcal{S})$) or ($\mathcal{S} \sim_v \mathcal{T}$ and $\mathcal{S} \leqslant_f \mathcal{T}$).
\end{definition}

\begin{theorem}
    \label{thm:cluster_order}
    $\leqslant_\sim$ is a total order on $\mathcal{S}/\sim$. The corresponding notion of equality is $=$.
\end{theorem}
\begin{proof}
    By the property of lexicographic order, $\leqslant_\sim$ is a total order. Thus, it suffices to show that anti-symmetry holds for $=$. For $\mathcal{S},\mathcal{T} \in S/{\sim}$, let $\mathcal{S}\leqslant_\sim \mathcal{T}$ and $\mathcal{S}\leqslant \mathcal{T}$. Thus, $\mathcal{S} \sim_v \mathcal{T}$ and $\mathcal{S}\sim_f \mathcal{T}$. By definition, we can find a $c_v$ and $c_f$ such that $\forall x\in\mathcal{S}\;\forall y\in\mathcal{T}\;v(x)=v(y)=c_v$ and $f(\mathcal{S}) = f(\mathcal{T}) = c_f$. Therefore, $\forall x\in\mathcal{S}\;\forall y\in\mathcal{T} x\sim y$, so $\mathcal{S} = \mathcal{T}$.
\end{proof}

As shown in Sec.~\ref{sec:cluster_representation}, we need to construct an oracle acting on clusters that meet the criterion $[v(x)<v(y)] \& f(x)$. We denote the oracle $O_{\xi}$, which satisfies Eq.~\ref{eq: oracle_clusters}, where $\eta$, $\epsilon$ are indices of the clusters. 
\begin{equation}
    \label{eq: oracle_clusters}
    O_{\xi}\ket{\eta} = (-1)^{f(\eta) \& [\eta < \epsilon]} \ket{\eta}
\end{equation}

To show that our definition of $O_{\xi}$ is correct, we need to prove Theorem~\ref{thm:o_function}.

\begin{theorem}
    \label{thm:o_function}
    Fix $\mathcal{S},\mathcal{T} \in S/{\sim}$, let $\xi = \idx \mathcal{S}, \eta = \idx \mathcal{T}$. For all $x\in\mathcal{S}$ and $y\in\mathcal{T}$, 
    \begin{equation}
        \label{eq:o_function}
        f(x)\&[v(x)<v(y)] = f(\xi)\&[\xi < \eta]
    \end{equation}
\end{theorem}
\begin{proof}
    Since $f(x) = f(\xi)$ and it suffices to show that when $f(x) = f(\xi) = 1$, $[v(x)<v(y)] = [\xi < \eta]$. If $\mathcal{S} = \mathcal{T}$, $f(x)\&[v(x)<v(y)] = f(\xi)\&[\xi < \eta] = 0$. So let $\mathcal{S} \neq \mathcal{T}$. When $v(x) \neq v(y)$, we have:
     \begin{equation}
        \begin{aligned}
            [v(x) < v(y)] & = [v(x) \leqslant v(y)]  \\
                          & = [\mathcal{S} \leqslant_v \mathcal{T}] \\
                          & = [\mathcal{S} \leqslant_v \mathcal{T}] \& [\mathcal{S} \neq \mathcal{T}] \\
                          & = [\xi \leqslant \eta] \& [\xi \neq \eta ] \\
                          & = [\xi < \eta]
        \end{aligned}
    \end{equation}
    
    So it suffices to show that $[v(x)<v(y)] = [\xi < \eta]$ when $v(x) = v(y)$. If $f(\eta) = f(\mathcal{T}) = 1$, $\mathcal{S} = \mathcal{T}$, violating previous assumption. So we must have $f(\eta) = f(\mathcal{T}) = 0$. Since $f(\xi) = f(\mathcal{S}) = 1$, $\mathcal{T} \leqslant_f \mathcal{S}$. Moreover, by Definition~\ref{def:cluster_relation}, for $z\in \mathcal{T}$, i.e., $z \sim y$, $v(z) \leqslant v(y)$ since there is a solution, in this case $x$, with measure function value $v(y)$. Thus, $v(z) \leqslant v(x)$. Note that because $f(\mathcal{S}) = 1$, all elements from $\mathcal{S}$ can only have the same $v$ value by definition, in this case, $v(x)$. Thus, $\mathcal{T} \leqslant_v \mathcal{S}$. Since $\mathcal{T} \leqslant_f \mathcal{S}$, we must have $\mathcal{T} \leqslant_\sim \mathcal{S}$. Hence, $\eta \leqslant \xi$, i.e., $\xi \geqslant \eta$; so we have: $[v(x) < v(y)] = 0 = [\xi < \eta]$.
\end{proof}



\vspace{10mm}

\bibliography{references.bib}

\end{document}